\documentclass[a4paper,12pt,english,notitlepage]{article}
\usepackage[latin1]{inputenc}
\usepackage[breaklinks,hypertexnames=false]{hyperref}
\usepackage{amsthm}
\usepackage{amsmath}
\usepackage{amssymb}
\usepackage{esint}
\makeatletter
\usepackage{graphicx}
\usepackage{epstopdf}
\usepackage{lscape}
\usepackage{rotating}
\usepackage[english]{babel}
\usepackage{latexsym}
\usepackage{fullpage}
\usepackage{bbm}
\usepackage{tikz}
\usetikzlibrary{arrows,shapes,trees,positioning,external}
\usepackage{booktabs}
\usepackage{rotating}
\usepackage{geometry}
\usepackage{pdflscape}
\usepackage{wasysym}
\usepackage{enumerate}
\usepackage{setspace}
\usepackage{enumitem} 

\hypersetup{
	colorlinks=true,
	citecolor=blue
}

\allowdisplaybreaks
\makeatother

\newtheorem{theorem}{Theorem}

\newtheorem{assumption}{Assumption}

\newtheorem{corollary}{Corollary}

\newtheorem{lemma}{Lemma}

\usepackage[compact]{titlesec}
\titlespacing{\section}{0pt}{1.5ex}{0.5ex}
\titlespacing{\subsection}{0pt}{1ex}{0.5ex}
\titlespacing{\subsubsection}{0pt}{0.5ex}{0ex}

\input{tcilatex}

\usepackage{babel}

\begin{document}

	\title{Discretizing Unobserved Heterogeneity\footnote{We thank Anna Simoni, Manuel Arellano, Neele Balke, Jesus Carro, Gary Chamberlain, Tim Christensen, Alfred Galichon, Chris Hansen, Joe Hotz, Gr\'egory Jolivet, Arthur Lewbel, Anna Mikusheva, Roger Moon, Whitney Newey, Juan Pantano, Philippe Rigollet, Martin Weidner, and seminar audiences at various places for comments. The authors acknowledge support from the NSF grant
			number SES-1658920. The usual disclaimer applies.}
	}
	\author{St\'{e}phane Bonhomme\footnote{University of Chicago, sbonhomme@uchicago.edu} \and Thibaut Lamadon\footnote{University of Chicago, lamadon@uchicago.edu}  \and Elena Manresa\footnote{New York University, elena.manresa@nyu.edu}}
	\date{Revised draft: January 2021}
	\maketitle

	\begin{abstract}

		\begin{spacing}{1.2}
			\noindent We study discrete panel data methods where unobserved heterogeneity is revealed in a first step, in environments where population heterogeneity is not discrete. We focus on \emph{two-step grouped fixed-effects} (GFE) estimators, where individuals are first classified into groups using \emph{kmeans} clustering, and the model is then estimated allowing for group-specific heterogeneity. Our framework relies on two key properties: heterogeneity is a function --- possibly nonlinear and time-varying ---  of a low-dimensional continuous latent type, and informative moments are available for classification. We illustrate the method in a model of wages and labor market participation, and in a probit model with time-varying heterogeneity. We derive asymptotic expansions of two-step GFE estimators as the number of groups grows with the two dimensions of the panel. We propose a data-driven rule for the number of groups, and discuss bias reduction and inference.  
		\end{spacing}
		

		\bigskip
		
		\textbf{JEL codes: }C23, C38.
		
		\textbf{Keywords: }Unobserved heterogeneity, panel data, kmeans clustering, dimension reduction.

	\end{abstract}
	
	\global\long\def\ind{\mathbbm{1}}
	\global\long\def\d{\mathrm{d}}
	\global\long\def\t{\intercal}
	\global\long\def\RR{\mathbb{R}}
	\global\long\def\defeq{=}

	\clearpage

	\section{Introduction\label{Intro_sec}}

	\noindent In both reduced-form and structural work in economics, it is common to model unobserved heterogeneity as a small number of discrete types. Various estimation strategies are available, including discrete-type random-effects (as in Keane and Wolpin, 1997, and many other applications) and grouped fixed-effects (as recently studied by Hahn and Moon, 2010, and Bonhomme and Manresa, 2015). These methods require the researcher to jointly estimate individual heterogeneity and model parameters.\footnote{Also related, nonparametric maximum likelihood methods (e.g., Heckman and Singer, 1984) rely on joint estimation of the distribution of heterogeneity and the parameters.} In addition, little is known about their properties when individual heterogeneity is not discrete in the population.\footnote{In a network context, Gao \textit{et al.} (2015) provide results for stochastic blockmodels under continuous heterogeneity.} In this paper, we study two-step discrete estimators for panel data, and provide conditions for their validity when heterogeneity is continuous.

	We focus on \emph{two-step grouped fixed-effects} (GFE) estimators. In a first step, we classify individuals based on a set of individual-specific moments, using the \emph{kmeans} clustering algorithm. The aim of the kmeans classification is to group together individuals whose latent types are most similar.\footnote{Buchinsky \textit{et al.} (2005) also propose to group individuals in a first step using kmeans.} In a second step, we estimate the model by allowing for group-specific heterogeneity. This second step is similar to fixed-effects (FE) estimation, albeit it involves a smaller number of parameters that are group-specific instead of individual-specific. We analyze the properties of these two-step estimators in panel data models where heterogeneity is continuous. Hence, in contrast with existing theoretical justifications for discrete-type methods, here we use discrete heterogeneity as a dimension reduction device rather than as a substantive assumption about population unobservables.

	Our approach is targeted to environments with two key properties. \emph{First}, unobserved heterogeneity is a function of a low-dimensional latent variable. We do not restrict this latent \emph{type} to be discrete. In many economic models, agents' heterogeneity in preferences or technology is driven by a low-dimensional type, which enters the model nonlinearly and may affect multiple outcomes. As an example, we study a model of participation in the labor market where the worker's utility is a function of her productivity type, which in turn determines her wage. GFE provides a tool to exploit such nonlinear factor structures. 
	
	\emph{Second}, the first-step moments satisfy an injectivity condition, which requires any two individuals with the same population moments to have the same type. The choice of moments is important to ensure good performance. In examples, we show how suitable moments arise naturally. In models with exogenous covariates, we propose and analyze the use of conditional moments to recover latent types.

	Our setup also covers models where heterogeneity varies over time. Unlike additive FE methods and interactive FE methods based on linear factor structures (Bai, 2009), GFE does not require heterogeneity to take an additive or interactive form. As an illustration, we compare GFE and FE estimators in a probit model where heterogeneity is a nonlinear function of a time-invariant factor loading and a time-specific factor.

	Our main results are large-$N,T$ asymptotic expansions of two-step GFE estimators under time-invariant and time-varying continuous heterogeneity. In both settings, GFE is consistent as the number of groups grows with the sample size, under conditions that we provide. We find that, when the population heterogeneity is not discrete, estimating group membership induces an incidental parameter bias, similarly to FE methods. Moreover, since discreteness is an approximation in our setting, GFE is affected by approximation error. We propose a simple data-driven rule for the number of groups that controls the approximation error, and discuss how to reduce incidental parameter bias for inference.

The outline of the paper is as follows. We introduce the setup and two-step GFE estimators in Section \ref{FE_sec}, study their asymptotic properties in Section \ref{Twostep_sec}, and outline several extensions in Section \ref{sec_Extens}. The main proofs may be found in the appendix, and the supplemental material contains additional results.

\section{Two-step grouped fixed-effects (GFE)\label{FE_sec}}

\noindent We consider a panel data setup, where we denote outcome variables and exogenous covariates as $Y_i{=}(Y_{i1}',...,Y_{iT}')'$ and $X_{i}{=}\left(X_{i1}',...,X_{iT}'\right)'$, respectively, for $i{=}1,...,N$. In our theory we cover two models. In the first one, unobserved heterogeneity is time-invariant. In this case, the conditional log-density of $Y_i$ given $X_i$ is given by:\footnote{In models with first-order dependence, we assume that $Y_{i0}$ is observed and we condition on it. Higher-order dependence can be accommodated similarly. In dynamic settings, $Y_{it}$ may contain sequentially exogenous covariates in addition to outcome variables.}
\begin{equation}\ln f_i(\alpha_{i0},\theta_0)=\sum_{t=1}^T\ln f(Y_{it}\,|\, Y_{i,t-1},X_{it},\alpha_{i0},\theta_0),\label{eq_ln_fi}\end{equation}
and the log-density of exogenous covariates $X_i$ takes the form:
$$\ln g_i(\mu_{i0})=\sum_{t=1}^T\ln g(X_{it}\,|\, X_{i,t-1},\mu_{i0}),$$
where $\theta_0$ is a vector of common parameters, and $\alpha_{i0}$ and $\mu_{i0}$ are individual-specific parameters. We leave the form of $g$ unrestricted, and in estimation we will use a conditional likelihood approach based on $f_i$ alone. In other words, in applications the researcher only needs to specify the parametric form of $f_i(\alpha_{i0},\theta_0)$ in (\ref{eq_ln_fi}). However, the heterogeneity $\mu_{i0}$ in covariates plays an important role in our theory. 

In the second model, unobserved heterogeneity varies over time. Such variation in unobservables over calendar time (e.g., business cycle), age (e.g., life cycle), counties, or markets, is of interest in many applications. In the time-varying case, log-densities take the form:
\begin{align*}\ln f_i(\alpha_{i0},\theta_0)&=\sum_{t=1}^T\ln f(Y_{it}\,|\, Y_{i,t-1},X_{it},\alpha_{it0},\theta_0),\\
	\ln g_i(\mu_{i0})&=\sum_{t=1}^T\ln g(X_{it}\,|\, X_{i,t-1},\mu_{it0}),\end{align*}
where $\alpha_{i0}=(\alpha_{i10}',...,\alpha_{iT0}')'$ and $\mu_{i0}=(\mu_{i10}',...,\mu_{iT0}')'$. In both models we are interested in estimating $\theta_0$, as well as average effects depending on $\alpha_{10},...,\alpha_{N0}$.


\subsection{Main assumptions}

\noindent GFE relies on two key assumptions that we now present. We defer the presentation of regularity conditions until Section \ref{Twostep_sec}. \emph{First}, we assume that unobserved heterogeneity is a function of a low-dimensional vector $\xi_{i0}$. 

\begin{assumption}{(heterogeneity)}\label{ass_alpha}
	
	\noindent	(a) Time-invariant heterogeneity: There exist $\xi_{i0}$ of fixed dimension $d$, and two Lipschitz-continuous functions $\alpha$ and $\mu$, such that $\alpha_{i0}={\alpha}(\xi_{i0})$ and $\mu_{i0}={\mu}(\xi_{i0})$. 
	
	\noindent	(b) Time-varying heterogeneity: 
	There exist $\xi_{i0}$ of fixed dimension $d$, $\lambda_{t0}$ of dimension $d_{\lambda}$, and two functions $\alpha$ and $\mu$ that are Lipschitz-continuous in their first argument, such that $\alpha_{it0}={\alpha}(\xi_{i0},\lambda_{t0})$ and $\mu_{it0}={\mu}(\xi_{i0},\lambda_{t0})$.  
\end{assumption}

We will refer to $\xi_{i0}$ as an individual \emph{type}, and to $d$ as the \emph{dimension} of heterogeneity. The researcher does not need to know $d$, $\alpha$, or $\mu$ in applications. In models with time-varying unobserved heterogeneity, Assumption \ref{ass_alpha} requires unobservables to follow a factor structure. The link between $\alpha_{it0}$, $\xi_{i0}$ and $\lambda_{t0}$ may be nonlinear, the linear structure $\alpha_{it0}=\xi_{i0}'\lambda_{t0}$ (Bai, 2009) being covered as a special case. Moreover, the dimension of $\lambda_{t0}$ is unrestricted. Our theory will show that the performance of two-step GFE crucially relies on $\xi_{i0}$ being low-dimensional, a leading case being $d=1$. We provide examples in the next subsection.

\emph{Second}, we rely on individual-specific moment vectors $h_i$ that are informative about the types $\xi_{i0}$. We state this formally as our second main assumption, where $\|\cdot\|$ denotes an Euclidean norm. 

\begin{assumption}{(injective moments)}\label{ass_inj}
	
	\noindent There exist vectors $h_i$ of fixed dimension, and a Lipschitz-continuous function $\varphi$, such that ${\limfunc{plim}}_{T\rightarrow\infty}\, h_i=\varphi(\xi_{i0})$, and $\frac{1}{N}\sum_{i=1}^N\|h_i-\varphi(\xi_{i0})\|^2=O_p\left(1/T\right)$ as $N,T$ tend to infinity. Moreover, there exists a Lipschitz-continuous function $\psi$ such that $\xi_{i0}=\psi(\varphi(\xi_{i0}))$.
\end{assumption}

Assumption \ref{ass_inj} requires the individual moment vector $h_i$ to be informative about $\xi_{i0}$, in the sense that, for large $T$, $\xi_{i0}$ can be uniquely recovered from $h_i$. Neither $\varphi$ nor $\psi$ (which may depend on $\theta_0$) need to be known to the econometrician. Intuitively, injectivity guarantees that one can separate the types of two individuals $\xi_{i0}$ and $\xi_{i'0}$ by comparing their moments $h_i$ and $h_{i'}$. For example, an average $h_i=\frac{1}{T}\sum_{t=1}^Th(Y_{it},X_{it})$ will, under Assumption \ref{ass_alpha} and suitable regularity conditions, converge as $T$ tends to infinity to a function $\varphi(\xi_{i0})$ of the type $\xi_{i0}$. We require $\varphi$ to be injective.

The convergence rate in Assumption \ref{ass_inj} requires appropriate conditions on the serial dependence of $Y_{it}$ and $X_{it}$. In models with time-varying heterogeneity, $\varphi$ will also depend on the $\lambda_{t0}$ process. In such models, Assumption \ref{ass_inj} requires the moments to be informative about $\xi_{i0}$, and not $\lambda_{t0}$. Injectivity is a key requirement for consistency of two-step GFE estimators. More generally, the choice of moments $h_i$ is important for finite-sample performance.

\subsection{Examples}

\noindent 
To illustrate the framework we now describe two examples, for which we will provide illustrative simulations in Subsection \ref{subsec_ex}. First, consider a dynamic model of wages $W_{it}^*$ and labor force participation $Y_{it}$:
\begin{eqnarray}\label{eq_ex_prob2}
\left\{\begin{array}{ccl}
Y_{it}&=&\boldsymbol{1}\left\{u(\alpha_{i0})\geq c(Y_{i,t-1};\theta_0)+U_{it}\right\},\\
W_{it}^*&=&\alpha_{i0}+V_{it},\\ W_{it}& =& Y_{it}W_{it}^*,\end{array}\right.
\end{eqnarray}
where the wage $W_{it}^*$ is only observed when $i$ works, $U_{it}$ are i.i.d. standard normal, independent of the past $Y_{it}$'s and $\alpha_{i0}$, and $V_{it}$ are i.i.d. independent of all $U_{it}$'s, $Y_{i0}$, and $\alpha_{i0}$. Here the same scalar expected payoff $\alpha_{i0}=\xi_{i0}$, unobserved to the econometrician, drives the wage and the decision to work. Individuals have common preferences denoted by the utility function $u$, the cost function $c$ is state-dependent, and both $u$ and $c$ are unknown to the econometrician. 

In this setting, GFE provides a natural approach to exploit the functional link between $\alpha_{i0}$ and $u(\alpha_{i0})$, and to learn about the type $\alpha_{i0}$ using both wages and participation. For instance, when $h_i=(\overline{W}_i,\overline{Y}_i)'$, where $\overline{Z}_{i}=\frac{1}{T}\sum_{t=1}^TZ_{it}$ denotes the individual mean of $Z_{it}$, injectivity is satisfied under mild conditions, provided $ \overline{W}_i=\alpha_{i0}\overline{Y}_i+o_p(1)$ and ${\limfunc{plim}}_{T\rightarrow\infty}\,\overline{Y}_i>0$. 

Fixed-effects (FE) is a possible approach to estimate $\theta_0$ in (\ref{eq_ex_prob2}). However, a conventional FE estimator would treat $\alpha_{i0}$ and $u_{i0}= u(\alpha_{i0})$ as unrelated parameters, so the FE estimate of $\theta_0$ would be solely based on the binary participation decisions. Another strategy would be to rely on discrete-type random-effects methods, which are typically based on joint estimation. In contrast, we implement GFE in two steps with no need for iterative estimation, and we justify the estimator in environments where heterogeneity is not restricted to be discrete.


As a second example, consider the following probit model with time-varying heterogeneity:
\begin{eqnarray}\label{eq_ex_prob}
\left\{\begin{array}{ccl}
Y_{it}&=&\boldsymbol{1}\left\{X_{it}'\theta_0+\alpha_{it0}+U_{it}\geq 0\right\},\\
X_{it}&=&\mu_{it0}+V_{it},\end{array}\right.
\end{eqnarray}
where $U_{it}$ are i.i.d. standard normal, independent of all $V_{it}$'s, $\alpha_{it0}$'s, and $\mu_{it0}$'s, and $V_{it}$ are i.i.d. independent of all $\alpha_{it0}$'s and $\mu_{it0}$'s. Under Assumption \ref{ass_alpha}, $\alpha_{it0}$ and $\mu_{it0}$ depend on a low-dimensional vector $\xi_{i0}$ of factor loadings, so $\alpha_{it0}={\alpha}(\xi_{i0},\lambda_{t0})$ and $\mu_{it0}={\mu}(\xi_{i0},\lambda_{t0})$. Here $d$ is the dimension of the type $\xi_{i0}$ governing both $\alpha_{it0}$ and $\mu_{it0}$. 

To motivate why, in static models with covariates such as (\ref{eq_ex_prob}), $\alpha_{it0}$ and $\mu_{it0}$ may depend on a common low-dimensional type $\xi_{i0}$, suppose that, in every period, agent $i$ chooses $X_{it}$ based on expected utility or profit maximization. She observes $\xi_{i0}$ and $\lambda_{t0}$ --- which enter outcomes through $\alpha_{it0}$ --- and takes her decision before the i.i.d. shock $U_{it}$ is realized. In such a case, $X_{it}$ will be a function of $\xi_{i0}$ and $\lambda_{t0}$, as well as idiosyncratic factors $V_{it}$ in the agent's information set. Here we assume that the agent's information set, and primitives such as preferences or costs, do not include other $i$-specific elements beyond $\xi_{i0}$.\footnote{This example is reminiscent of Mundlak's (1961) classic analysis of farm production functions, where soil quality $\xi_{i0}$ is observed to the farmer but latent to the analyst.}

When $\alpha(\cdot,\cdot)$ is additive or multiplicative in its arguments, model (\ref{eq_ex_prob}) can be estimated using two-way FE (Fern\'andez-Val and Weidner, 2016) or interactive FE (Bai, 2009, Chen \textit{et al.}, 2020), respectively. However, when $\alpha(\cdot,\cdot)$ is unknown, these fixed-effects estimators are inconsistent in general. In contrast, GFE will remain consistent when unobservables are unknown nonlinear functions of factor loadings $\xi_{i0}$ and factors $\lambda_{t0}$, and injectivity holds. Taking ${h}_i=(\overline{Y}_i,\overline{X}_i')'$ as moments in model (\ref{eq_ex_prob}), injectivity is satisfied when types have monotone effects on the heterogeneity components.\footnote{To see this, consider the case where $\alpha_{it0}$ is the only component of heterogeneity (i.e., $\mu_{it0}=0$ in (\ref{eq_ex_prob})), and take $h_i=\overline{Y}_i$. Letting $G$ denote the cdf of $-(V_{it}'\theta_0+U_{it})$, injectivity will hold when $\alpha(\cdot,\cdot)$ is strictly increasing in its first argument and $G$ is strictly increasing, since then $\varphi(\xi)= {\limfunc{plim}}_{T\rightarrow\infty}\,\frac{1}{T}\sum_{t=1}^TG(\alpha(\xi,\lambda_{t0}))$ is strictly increasing.} More generally, in Assumption \ref{ass_inj} we require that the latent type $\xi_{i0}$ can be asymptotically recovered from a moment vector whose dimension is not growing with the sample size.

\subsection{Estimator}

\noindent Two-step GFE consists of a \emph{classification} step and an \emph{estimation} step.

\paragraph{First step: classification.} We rely on the individual-specific moments $h_i$ to learn about the individual types $\xi_{i0}$. Specifically, we partition individuals into $K$ groups, corresponding to group indicators $\widehat{k}_i\in\{1,...,K\}$ , by computing:
\begin{equation}\label{eq_firststep}
\left(\widehat{h}(1),...,\widehat{h}(K),\widehat{k}_1,...,\widehat{k}_N\right)=\underset{\left(\widetilde{h}(1),...,\widetilde{h}(K),k_1,...,k_N\right)}{\limfunc{argmin}}\,\,\,\sum_{i=1}^N\left\|h_i-\widetilde{h}(k_i)\right\|^2,
\end{equation} 
where $\{k_i\}$ are partitions of $\{1,...,N\}$ into $K$ groups, and $\widetilde{h}(k)$ is a vector. Note that $\widehat{h}(k)$ is simply the mean of $h_i$ in group $\widehat{k}_i=k$. 

In the \emph{kmeans} optimization problem (\ref{eq_firststep}), the minimum is taken with respect to all possible partitions $\{k_i\}$. Fast and stable optimization methods such as Lloyd's algorithm are available, although computing a global minimum may be challenging; see Bonhomme and Manresa (2015) for references. Following the literature, we will focus on the asymptotic properties of the global minimum and abstract from optimization error. Lastly, note that the quadratic loss function in (\ref{eq_firststep}) can accommodate weights on different components of $h_i$, although here for simplicity we present the unweighted case. 

\paragraph{Second step: estimation.} We maximize the log-likelihood function with respect to common parameters $\theta$ and group-specific effects $\alpha$, where the groups are given by the $\widehat{k}_i$ estimated in the first step. We define the two-step GFE estimator as:
\begin{equation}\label{eq_secondstep}
\left(\widehat{\theta},\widehat{\alpha}(1),...,\widehat{\alpha}(K)\right)=\underset{\left(\theta,\alpha(1),...,\alpha(K)\right)}{\limfunc{argmax}}\,\,\,\sum_{i=1}^N\ln f_i\left(\alpha\left(\widehat{k}_i\right),\theta\right).
\end{equation} 
Note that, in contrast to fixed-effects (FE) maximum likelihood, this second step involves a maximization with respect to $K$ group-specific parameters instead of $N$ individual-specific ones. In models with time-varying heterogeneity, $\alpha(k)$ will simply be a vector $(\alpha_1(k)',...,\alpha_T(k)')'$.

\paragraph{Choice of $K$.} Two-step GFE estimation requires setting a number of groups $K$. We propose a simple data-driven selection rule based on the first step. The convergence rate of the kmeans estimator (and the rate of the GFE estimator) will be governed by two quantities: the kmeans objective function $\widehat{Q}(K)=\frac{1}{N}\sum_{i=1}^N\|h_i-\widehat{h}(\widehat{k}_i)\|^2$, which decreases as $K$ gets larger and the group approximation becomes more accurate, and the variability $V_h=\mathbb{E}[\|h_i-\varphi(\xi_{i0})\|^2]$ of the moment $h_i$, which does not depend on $K$. We take the smallest $K$ that guarantees that $\widehat{Q}(K)$ is of the same or lower order as $V_h$. That is, letting $\widehat{V}_{h}=V_h+o_p(1/T)$, we suggest setting:
\begin{equation}\label{choice_K_eq}
\widehat{K}=\underset{K\geq 1}{\limfunc{min}}\, \left\{K \, : \, \widehat{Q}(K)\leq  \gamma\widehat{V}_{h}\right\},
\end{equation}
where $\gamma\in(0,1]$ is a user-specified parameter.\footnote{When $h_i=\frac{1}{T}\sum_{t=1}^Th(Y_{it},X_{it})$ and observations are independent over time, one may take $\widehat{V}_{h}=\frac{1}{NT^2}\sum_{i=1}^N \sum_{t=1}^T \|h(Y_{it},X_{it})-h_i\|^2$. With dependent data, one can use trimming or the bootstrap to estimate $V_h$ (Hahn and Kuersteiner, 2011, Arellano and Hahn, 2007).} In the simulations in the next subsection we will set $\gamma=1$, although smaller $\gamma$ values corresponding to larger $K$'s will also be supported by our theory.

\subsection{Illustrative simulations\label{subsec_ex}}

%


\noindent To illustrate the performance of GFE in models where heterogeneity follows a nonlinear factor structure, we present the results of a small-scale simulation study based on our two examples (\ref{eq_ex_prob2}) and (\ref{eq_ex_prob}). In both cases, we assume that the type $\xi_{i0}$ governing heterogeneity is scalar. We compare the bias of GFE to that of FE and interactive FE estimators. In the supplemental material, we provide details on the simulations and report additional results.

\begin{figure}[tb!]
	\caption{Model (\ref{eq_ex_prob2}) of wages and participation\label{Fig_Earnings}}
	\begin{center}
		\includegraphics[width=100mm]{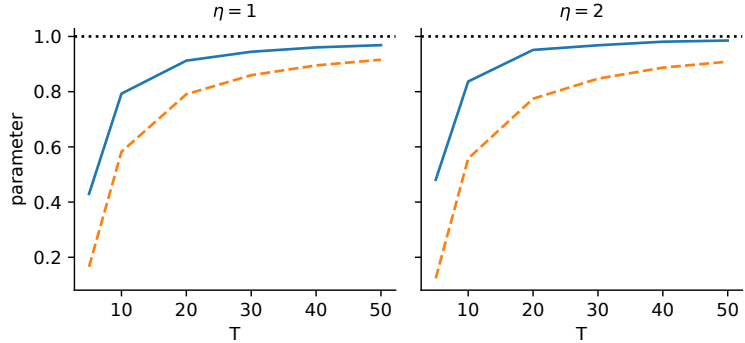}
	\end{center}
	\vspace{-20pt}
	{\footnotesize\textit{Notes: Means of $c(0;
			\widehat{\theta})-c(1;
			\widehat{\theta})$ over 1000 simulations. GFE is indicated in solid, FE is in dashed, and the truth $c(0;\theta_0)-c(1;\theta_0)=1$ is in dotted. $N=1000$, and $T$ is indicated on the x-axis. $\eta$ is the risk aversion parameter in $u(\cdot)$. See the supplemental material for details.}}
\end{figure}

\begin{figure}[b!]
	\caption{Probit model (\ref{eq_ex_prob}) with time-varying heterogeneity\label{Fig_TV}}
	\begin{center}
		\begin{tabular}{cc}
			\hspace*{-0.5cm}  \includegraphics[width=150mm]{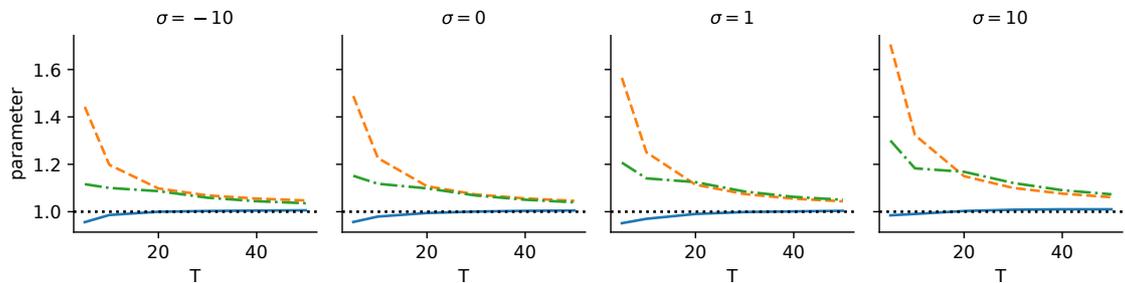}
		\end{tabular}
	\end{center}	
	\vspace{-13pt}
	{\footnotesize\textit{Notes: Means of $\widehat{\theta}$ over 1000 simulations. GFE is indicated in solid, FE is in dashed, interactive FE is in dash-dotted, and the truth $\theta_0=1$ is in dotted. $N=1000$, and $T$ is indicated on the x-axis. $\sigma$ is the substitution parameter in $\alpha(\cdot,\cdot)$. See the supplemental material for details.}}
\end{figure}


In Figure \ref{Fig_Earnings}, we compare GFE and FE in model (\ref{eq_ex_prob2}), using a CRRA functional form: $u(\alpha)=\frac{{e^{\alpha}}^{(1-\eta)}-1}{1-\eta}$, with a risk aversion parameter $\eta\in\{1,2\}$. We focus on the difference in costs $c(0;\theta_0)-c(1;\theta_0)$, which measures the degree of state dependence in participation decisions. We take $h_i=(\overline{W}_i,\overline{Y}_i)'$ as moments for GFE, and report average parameter estimates over 1000 simulations. We set $N{=}1000$ and vary T between 5 and 50. We find that FE is more biased than GFE for both values of risk aversion. This is consistent with wages and participation providing informative moments about the latent type in this setting.


In Figure \ref{Fig_TV} we compare GFE, FE, and interactive FE in model (\ref{eq_ex_prob}) with $X_{it}$ scalar, using a CES specification: $\alpha_{it0}=\left(a\xi_{i0}^\sigma +(1-a)\lambda_{t0}^{\sigma}\right)^{\frac{1}{\sigma}}$,
for $\sigma{\in}\{-10,0,1,10\}$ and $a{=}0.5$, and $\mu_{it0}=\alpha_{it0}$. The factors $\lambda_{t0}$ and the individual loadings $\xi_{i0}$ enter heterogeneity in a nonlinear way. We show estimates of ${\theta}_0$ for various estimators: GFE, FE with additive individual and time effects, and interactive FE with a single multiplicative factor. We use $(\overline{Y}_i,\overline{X}_i)'$ as moments for GFE. Note that both $\overline{Y}_i$ and $\overline{X}_i$ are informative about $\xi_{i0}$ in this data generating process. We report parameter averages over 1000 simulations, for $N{=}1000$. We find that, while GFE, FE, and interactive FE are all biased, the bias of GFE is smaller across all $\sigma$ values.\footnote{Large-$N,T$ theory implies that additive and interactive FE are consistent when $\sigma=1$ and $\sigma=0$, respectively. Figure \ref{Fig_TV} shows that, despite being large-$N,T$ consistent in these specifications, in our simulations, additive and interactive FE have larger biases than GFE for the $N$ and $T$ values we consider.  }

%

\section{Asymptotic properties\label{Twostep_sec}}

\noindent In this section we provide asymptotic expansions for two-step GFE estimators. Our first result is a rate of convergence for kmeans. Let us define the \emph{approximation error} one would make if one were to discretize the latent types $\xi_{i0}$ directly, as:
\begin{equation}
B_{\xi}(K)=\underset{\left(\widetilde{\xi}(1),...,\widetilde{\xi}(K),k_1,...,k_N\right)}{\limfunc{min}}\,\,\, \frac{1}{N}\sum_{i=1}^N \left\|\xi_{i0}-\widetilde{\xi}(k_i)\right\|^2,\label{eq_approx_error}
\end{equation}
where, similarly to (\ref{eq_firststep}), the minimum is taken with respect to all partitions $\{k_i\}$ and vectors $\widetilde{\xi}(k)$. In the asymptotic analysis we let $T=T_N$ and $K=K_N$ tend to infinity jointly with $N$. 


\begin{lemma}\label{theo1}Let Assumption \ref{ass_inj} hold. Let $\widehat{h}(1),...,\widehat{h}(K)$ and $\widehat{k}_1,...,\widehat{k}_N$ given by (\ref{eq_firststep}). Then, as $N,T,K$ tend to infinity we have:
	$$\frac{1}{N}\sum_{i=1}^N \left\|\widehat{h}(\widehat{k}_i)-\varphi(\xi_{i0})\right\|^2=O_p\left(\frac{1}{T}\right)+O_p\left(B_{\xi}(K)\right).$$
	
\end{lemma}

The bound in Lemma \ref{theo1} has two terms: an $O_p(1/T)$ term that depends on the number of periods used to construct the moments $h_i$, and an $O_p\left(B_{\xi}(K)\right)$ term that reflects the presence of an approximation error. The rate at which $B_{\xi}(K)$ tends to zero depends on the dimension of $\xi_{i0}$. Graf and Luschgy (2002, Theorem 5.3) provide explicit characterizations in the case where $\xi_{i0}$ has compact support.\footnote{See Graf and Luschgy (2002, p. 875) for a discussion of the compact support assumption.} For example, the following lemma implies that $B_{\xi}(K)=O_p(K^{-2})$ when $\xi_{i0}$ is one-dimensional, and $B_{\xi}(K)=O_p(K^{-1})$ when $\xi_{i0}$ is two-dimensional.

\begin{lemma}\label{lemma_GL}(Graf and Luschgy, 2002) Let $\xi_{i0}$ be random vectors with compact support in $\mathbb{R}^{d}$. Then, as $N,K$ tend to infinity we have $ B_{\xi}(K)=O_p(K^{-\frac{2}{d}})$.
\end{lemma}

We now use these results to study the properties of GFE in models with time-invariant and time-varying heterogeneity, in turn. We use the shorthand notation $\mathbb{E}_{Z}(W)$ and $\mathbb{E}_{Z=z}(W)$ for the conditional expectations of $W$ given $Z$ and $Z=z$, respectively. In the time-varying case, we denote as $\lambda_0$ the process of $\lambda_{t0}$'s, and as $\mathbb{E}_{\lambda_0=\lambda}(W)$ the conditional expectation of $W$ given $\lambda_0=\lambda$. We use a similar notation for variances. Finally, $\|M\|$ denotes the spectral norm of a matrix $M$.

\subsection{Time-invariant heterogeneity}

\noindent To state our first main theorem, where heterogeneity is time-invariant, we make the following assumptions, where $\ell_{it}(\alpha_{i},\theta)=\ln f(Y_{it}\,|\, Y_{i,t-1},X_{it},\alpha_{i},\theta)$, $\ell_i(\alpha_{i},\theta)=\frac{1}{T}\sum_{t=1}^T\ell_{it}(\alpha_{i},\theta)$, and $\overline{\alpha}(\theta,\xi)={\limfunc{argmax}}_{\alpha}\,\mathbb{E}_{\xi_{i0}=\xi}(\ell_{i}(\alpha,\theta))$ for all $\theta,\xi$.

\begin{assumption}{(regularity, time-invariant heterogeneity)} \label{ass_regu}

	\begin{enumerate}[itemsep=-3pt,label=(\roman*),ref=\roman*,topsep=0pt]
		\item $(Y_i',X_i',\xi_{i0}',h_i')'$ are i.i.d.; $(Y_{it}',X_{it}')'$ are stationary for all $i$; $\ell_{it}(\alpha,\theta)$ is three times differentiable in $(\alpha,\theta)$ for all $i,t$;\footnote{That is, $\ln f(y_{it}\,|\, y_{i,t-1},x_{it},\alpha,\theta)$ is three times differentiable in $(\alpha,\theta)$, for almost all $(y_{it},y_{it-1},x_{it})$.} and the parameter space $\Theta$ for $\theta_0$ is compact, the space for $\alpha_{i0}$ is compact, and $\theta_0$ belongs to the interior of $\Theta$.\label{ass_regu_i}

		\item $N,T,K$ tend jointly to infinity; $\sup_{\xi,\alpha,\theta}\, |\mathbb{E}_{\xi_{i0}=\xi}(\ell_{it}(\alpha,\theta))|=O(1)$, and similarly for the first three derivatives of $\ell_{it}$; $ \inf_{\xi,\alpha,\theta}\, \mathbb{E}_{\xi_{i0}=\xi}(-\frac{\partial^2 \ell_{it}(\alpha,\theta)}{\partial\alpha\partial{{\alpha}^\prime}})>0$; and $\max_{i}\,\sup_{\alpha,\theta}\,\left|\ell_{i}(\alpha,\theta)-\mathbb{E}_{\xi_{i0}}\left(\ell_{i}(\alpha,\theta)\right)\right|=o_p\left(1\right)$, and similarly for the first three derivatives of $\ell_{i}$.
		\label{ass_regu_iia}

		\item 
		$ \inf_{\xi,\theta}\, \mathbb{E}_{\xi_{i0}=\xi}(-\frac{\partial^2 \ell_{it}(\overline{\alpha}(\theta,\xi),\theta)}{\partial\alpha\partial{\alpha}^{\prime}})>0$; $\mathbb{E}[ \frac{1}{T}\sum_{t=1}^T\ell_{it}(\overline{\alpha}(\theta,\xi_{i0}),\theta)]$ has a unique maximum at $\theta_0$ on $\Theta$, and its matrix of second derivatives is $-H<0$; and		$\sup_{\theta}\frac{1}{NT}\sum_{i=1}^N\sum_{t=1}^T\|\frac{\partial^2\ell_{it}(\overline{\alpha}(\theta,\xi_{i0}),\theta)}{\partial \theta\partial \alpha'}\|^2=O_p(1)$.

		\label{ass_regu_iii}

		\item
		
		$\sup_{\widetilde{\xi},\alpha} \|\frac{\partial}{\partial\xi'}\big|_{\xi=\widetilde{\xi}}\,\mathbb{E}_{\xi_{i0}=\xi}(\limfunc{vec}\frac{\partial^2 \ell_{it}(\alpha,\theta_0)}{\partial \theta\partial{\alpha}^{\prime}})\|$, $ \sup_{\widetilde{\xi},\alpha} \|\frac{\partial}{\partial\xi'}\big|_{\xi=\widetilde{\xi}}\,\mathbb{E}_{\xi_{i0}=\xi}(\limfunc{vec}\frac{\partial^2 \ell_{it}(\alpha,\theta_0)}{\partial \alpha\partial{\alpha}^{\prime}})\|$, and $\sup_{\widetilde{\xi},\theta}\ \|\frac{\partial}{\partial\xi'}\big|_{\xi=\widetilde{\xi}}\,\mathbb{E}_{\xi_{i0}=\xi}(\frac{\partial \ell_{it}(\overline{\alpha}(\theta,\widetilde{\xi}),\theta)}{\partial\alpha})\|$ are $O(1)$.
		\label{ass_regu_iv}

	\end{enumerate} 
	
\end{assumption}

In part (\ref{ass_regu_i}) in Assumption \ref{ass_regu} we treat heterogeneity as random in order to use Lemma \ref{lemma_GL}, which requires $\xi_{i0}$ to be i.i.d. draws from a distribution. However, note we do not restrict how $\alpha_{i0}$ and $\mu_{i0}$ depend on each other. Moreover, while our results require asymptotic stationarity of the time-series processes, the theorem could be extended to allow for nonstationary initial conditions. 

In part (\ref{ass_regu_iia}) we require strict concavity of the log-likelihood as a function of $\alpha$. Concavity holds in a number of nonlinear panel data models such as probit and logit models, tobit, Poisson, or multinomial logit; see Fern\'andez-Val and Weidner (2016) and Chen \textit{et al.} (2020). One can show that Theorem \ref{theo2} continues to hold without concavity, under an identification condition and an assumption bounding the derivatives of the empirical GFE objective function. Importantly, note that $H^{-1}$ is the asymptotic variance of the FE estimator. As a result, $H$ being positive definite rules out models that are not identified under FE, such as a linear model with a time-invariant covariate and a heterogeneous intercept.

In part (\ref{ass_regu_iii}) we introduce the \emph{target} log-likelihood $\frac{1}{NT}\sum_{i=1}^N\sum_{t=1}^T\ell_{it}(\overline{\alpha}(\theta,\xi_{i0}),\theta)$ (Arellano and Hahn, 2007), which we will show approximates the GFE log-likelihood in large samples under our assumptions (note that $\overline{\alpha}(\theta_0,\xi_{i0}){=}\alpha_{i0}$). In part (\ref{ass_regu_iv}) we require some moments to be bounded asymptotically.

We now state our first main result, where we denote, evaluating all quantities at true values $(\theta_0,\alpha_{i0})$ and omitting the dependence from the notation: 
\begin{align}&s_i=\frac{1}{T}\sum_{t=1}^T\Bigg(\frac{\partial\ell_{it}}{\partial \theta}+\mathbb{E}_{\xi_{i0}}\left(\frac{\partial^2\ell_{it} }{\partial \theta\partial \alpha'}\right)\left[\mathbb{E}_{\xi_{i0}}\left(-\frac{\partial^2\ell_{it}}{\partial \alpha\partial \alpha'}\right)\right]^{-1}\frac{\partial\ell_{it} }{\partial \alpha}\Bigg),\label{eq_score}\\
&H=\underset{N,T\rightarrow \infty}{\limfunc{plim}}\, \frac{1}{NT}\sum_{i=1}^N\sum_{t=1}^T\Bigg(\mathbb{E}_{\xi_{i0}}\left(-\frac{\partial^2 \ell_{it} }{\partial \theta\partial \theta'}\right)\notag\\
&\quad\quad \quad \quad \quad-\mathbb{E}_{\xi_{i0}}\left(\frac{\partial^2\ell_{it} }{\partial \theta\partial \alpha'}\right)\left[\mathbb{E}_{\xi_{i0}}\left(-\frac{\partial^2\ell_{it}}{\partial \alpha\partial \alpha'}\right)\right]^{-1}\mathbb{E}_{\xi_{i0}}\left(\frac{\partial^2\ell_{it} }{\partial \alpha\partial \theta'}\right)\Bigg).\label{eq_hessian}\end{align}

\begin{theorem}\label{theo2}Let the conditions of Lemmas \ref{theo1} and \ref{lemma_GL} and Assumptions \ref{ass_alpha}, \ref{ass_inj} and \ref{ass_regu} hold. Then, as $N,T,K$ tend to infinity we have:
	\begin{eqnarray}\label{2steptheta}
	\widehat{\theta}&=&\theta_0+ H^{-1}\frac{1}{N}\sum_{i=1}^Ns_i +O_p\left(\frac{1}{T}\right)+O_p\left(K^{-\frac{2}{d}}\right)+o_p\left(\frac{1}{\sqrt{NT}}\right).
	\end{eqnarray}

\end{theorem}

The first three terms in (\ref{2steptheta}) also appear in large-$N,T$ expansions of FE estimators (e.g., Hahn and Newey, 2004).\footnote{In the supplemental material we provide a similar expansion for GFE estimators of average effects $M_0=\frac{1}{NT}\sum_{i=1}^N \sum_{t=1}^Tm\left(X_{it},\alpha_{i0},\theta_0\right)$, which are functions of both common parameters and individual heterogeneity.} Similarly to FE, GFE is subject to incidental parameter $O_p(1/T)$ bias. This contrasts with the properties of GFE estimators under discrete heterogeneity (e.g., Hahn and Moon, 2010, Bonhomme and Manresa, 2015). Indeed, when heterogeneity is \emph{not} restricted to have a small number of points of support, classification noise affects the properties of second-step estimators in general. This motivates using bias reduction techniques for inference analogous to those used in FE, as we will discuss in the next section.

The $O_p(K^{-\frac{2}{d}})$ term in (\ref{2steptheta}) reflects the approximation error, which depends on the number of groups. Setting $K=\widehat{K}$ according to (\ref{choice_K_eq}) guarantees that the approximation error is $O_p(1/T)$. Formally, we have the following result. 

\begin{corollary}\label{coro_groups}Let the conditions in Theorem \ref{theo2} hold. Let $K=\widehat{K}$ given by (\ref{choice_K_eq}), with $\gamma=O(1)$. Then, as $N,T$ tend to infinity we have:
	\begin{eqnarray}\label{2steptheta_groups}
	\widehat{\theta}&=&\theta_0+ H^{-1}\frac{1}{N}\sum_{i=1}^Ns_i +O_p\left(\frac{1}{T}\right)+o_p\left(\frac{1}{\sqrt{NT}}\right).
	\end{eqnarray}
	
\end{corollary}

Under Corollary \ref{coro_groups}, the biases of FE and GFE have the same order of magnitude. However, the required value of $K$ depends on the dimension $d$ of individual heterogeneity. Specifically, when $\xi_{i0}$ follows a continuous distribution of dimension $d$, setting $K$ proportional to or greater than $\min(T^{\frac{d}{2}},N)$ will ensure that the approximation error is $O_p(1/T)$. For small $d$ (e.g., when $d=1$) this will typically require a small number of groups (of the order of $T^{\frac{1}{2}}$). 

GFE can have advantages compared to FE, for two reasons. First, the two-step method can allow researchers to select moments that are particularly informative about the unobserved heterogeneity. To provide intuition, consider a setting where the number of groups is sufficiently large for the approximation error to be of smaller order compared to $1/T$, yet $K/N$ tends to zero. We have the following. 
\begin{corollary}\label{coro_bias}Let the conditions in Theorem \ref{theo2} hold. Let $K=\widehat{K}$ given by (\ref{choice_K_eq}), with $\gamma=o(1)$. Suppose that $K/N$ tends to zero, and that Assumption \ref{ass_coro_bias} in the appendix holds. Then the $O_p(1/T)$ term in (\ref{2steptheta_groups}) takes the explicit form $C/T+o_p(1/T)$, where:
	\begin{eqnarray}\label{1storder_GFE}
\frac{C}{T}=	H^{-1}\frac{\partial}{\partial\theta}\bigg|_{\theta_0}\mathbb{E} \left[\frac{1}{2}\left\|\widehat{\alpha}_i(\theta){-}\mathbb{E}_{\xi_{i0}}(\widehat{\alpha}_i(\theta))\right\|_{\Omega_i(\theta)}^2-\frac{1}{2}\left\|\widehat{\alpha}_i(\theta){-}\mathbb{E}_{h_i}(\widehat{\alpha}_i(\theta))\right\|^2_{\Omega_i(\theta)}\right],
	\end{eqnarray}
	with $\widehat{\alpha}_i(\theta)=\limfunc{argmax}_{\alpha} \ell_i(\theta,\alpha)$, $\Omega_i(\theta)=\mathbb{E}_{\xi_{i0}}(-\frac{\partial^2 \ell_{it}(\overline{\alpha}(\theta,\xi_{i0}),\theta)}{\partial\alpha\partial{{\alpha}^\prime}})$, and $\|V\|_{\Omega}^2=V'\Omega V$.
	
\end{corollary}

Corollary \ref{coro_bias} shows that the first-order asymptotic bias of GFE is the difference between two terms. The bias is zero when $h_i$ is an injective function of $\xi_{i0}$; i.e., when $\varepsilon_i=h_i-\varphi(\xi_{i0})=0$. More generally, the bias can be expanded in $\varepsilon_i$, and it is small when moments provide accurate estimates of the latent types. Moreover, the first term on the right-hand side of (\ref{1storder_GFE}) coincides with the bias of FE (e.g., Arellano and Hahn, 2007). The form of (\ref{1storder_GFE}) implies that the biases of FE and GFE are equal when the moments are the FE estimates $h_i=\widehat{\alpha}_i(\theta_0)$, however other moment choices can lead to smaller biases. From this perspective, GFE provides flexibility to use well-suited proxies of the latent types. As an example, our simulations of the labor force participation model (\ref{eq_ex_prob2}) show that, by jointly exploiting wages and participation to construct moments that are informative about the latent type, GFE can have smaller bias than FE (and smaller mean squared error as well, as shown in the supplemental material).

A second advantage of GFE comes from the use of grouping, and from the resulting regularization. Indeed, individual FE estimates can be highly variable whenever the number of parameters per individual is large. In such cases, reducing the number of parameters through grouping can improve performance. For instance, the ability to handle multiple components of heterogeneity is central to the performance of GFE in models with time-varying unobserved heterogeneity. This is the case we focus on next.

\subsection{Time-varying heterogeneity}

\noindent To state our second main theorem, where heterogeneity is time-varying, we make the following assumptions, where $\ell_{it}(\alpha_{it},\theta)=\ln f(Y_{it}\,|\, Y_{i,t-1},X_{it},\alpha_{it},\theta)$, $\ell_i(\alpha_{i},\theta)=\frac{1}{T}\sum_{t=1}^T\ell_{it}(\alpha_{it},\theta)$, and $\overline{\alpha}^t(\theta,\xi)={\limfunc{argmax}}_{\alpha}\,\mathbb{E}_{\xi_{i0}=\xi,\lambda_0=\lambda}(\ell_{it}(\alpha,\theta))$.\footnote{Note that $\overline{\alpha}^t(\theta,\xi_{i0})$ depends on the process $\lambda_{0}$ in addition to the type $\xi_{i0}$, although we leave the dependence on $\lambda_0$ implicit in the notation. In a static model, $\overline{\alpha}^t(\theta,\xi_{i0})$ is a function of $\xi_{i0}$ and $\lambda_{t0}$, while in a dynamic model it also depends on the history of the time effects $(\lambda_{t0},\lambda_{t-1,0},...)$.}

\begin{assumption}{(regularity, time-varying heterogeneity)} \label{ass_regu_TV}

	\begin{enumerate}[itemsep=-3pt,label=(\roman*),ref=\roman*,topsep=0pt]
		\item $(Y_i',X_i',\xi_{i0}',h_i')'$ are i.i.d. across $i$ conditional on $\lambda_{0}$; $(Y_{it}',X_{it}',\lambda_{t0}')'$ are stationary for all $i$; $\ell_{it}(\alpha_{it},\theta)$ is three times differentiable, for all $i,t$; and $\Theta$ and the space for $\alpha_{it0}$ are compact, and $\theta_0$ belongs to the interior of $\Theta$.\label{ass_regu_ib}

		\item $N,T,K$ tend jointly to infinity; $\max_{t} \,\sup_{\xi,\lambda,\alpha,\theta}\, |\mathbb{E}_{\xi_{i0}=\xi,\lambda_0=\lambda}(\ell_{it}(\alpha,\theta))|=O(1)$, and similarly for the first three derivatives of $\ell_{it}$; the minimum (respectively, maximum) eigenvalue of $(-\frac{\partial^2 \ell_{it}(\alpha,\theta)}{\partial\alpha\partial{\alpha}^{\prime}})$ is bounded away from zero (resp., infinity) with probability one, uniformly in $i,t,\alpha,\theta$; the third derivatives of $\ell_{it}(\alpha,\theta)$ are $O_p(1)$, uniformly in $i,t,\alpha,\theta$; and $\frac{1}{NT}\sum_{i=1}^N\sum_{t=1}^T[\ell_{it}(\alpha_{it0},\theta_0)-\mathbb{E}_{\xi_{i0},\lambda_0}(\ell_{it}(\alpha_{{i}t0},\theta_0))]^2=O_p(1)$, and similarly for the first three derivatives.\label{ass_regu_iib}

		\item 
		$ \min_{t} \,\inf_{\xi,\lambda,\theta}\, \mathbb{E}_{\xi_{i0}=\xi,\lambda_0=\lambda}(-\frac{\partial^2 \ell_{it}(\overline{\alpha}^t(\theta,\xi),\theta)}{\partial\alpha\partial{\alpha}^{\prime}})>0$; $\mathbb{E}[ \frac{1}{T}\sum_{t=1}^T\ell_{it}(\overline{\alpha}^t(\theta,\xi_{i0}),\theta)]$ has a unique maximum at $\theta_0$ on $\Theta$, and its matrix of second derivatives is $-H<0$; and $\sup_{\theta}\frac{1}{NT}\sum_{i=1}^N\sum_{t=1}^T\|\frac{\partial^2\ell_{it}(\overline{\alpha}^t(\theta,\xi_{i0}),\theta)}{\partial \theta\partial \alpha'}\|^2=O_p(1)$.

		\label{ass_regu_iiib}

		\item
		
		$ \|\frac{\partial}{\partial\xi'}\big|_{\xi=\widetilde{\xi}}\,\mathbb{E}_{\xi_{i0}=\xi,\lambda_0=\lambda}(\limfunc{vec}\frac{\partial^2 \ell_{it}(\alpha,\theta_0)}{\partial \theta\partial{\alpha}^{\prime}})\|$, $\|\frac{\partial}{\partial\xi'}\big|_{\xi=\widetilde{\xi}}\,\mathbb{E}_{\xi_{i0}=\xi,\lambda_0=\lambda}(\limfunc{vec}\frac{\partial^2 \ell_{it}(\alpha,\theta_0)}{\partial \alpha\partial{\alpha}^{\prime}})\|$, and \\$ \|\frac{\partial}{\partial\xi'}\big|_{\xi=\widetilde{\xi}}\,\mathbb{E}_{\xi_{i0}=\xi,\lambda_0=\lambda}(\frac{\partial \ell_{it}(\overline{\alpha}^t(\theta,\widetilde{\xi}),\theta)}{\partial\alpha})\|$ are $O(1)$, uniformly in $t$, $\widetilde{\xi}$, $\lambda$, $\alpha$, and $\theta$.
		\label{ass_regu_ivb}

		\item 
		
	$\mathbb{E}_{h_i=h,\xi_{i0}=\xi,\lambda_0=\lambda}(\frac{\partial \ell_{it}(\overline{\alpha}^{t}(\theta,\xi),\theta)}{\partial\alpha})$ and $\mathbb{E}_{h_i=h,\xi_{i0}=\xi,\lambda_0=\lambda}(\limfunc{vec}\frac{\partial}{\partial \theta'}\big|_{\theta_0}\frac{\partial \ell_{it}(\overline{\alpha}^t(\theta,\xi),\theta)}{\partial\alpha})$ are twice differentiable with respect to $h$, with first and second derivatives that are uniformly bounded in $t$, $\xi$, $\lambda$, $h$ in the support of $h_i$ given $\lambda_0=\lambda$, and $\theta\in \Theta$; and  $\|{\limfunc{Var}}_{h_i=h,\xi_{i0}=\xi,\lambda_0=\lambda}(\frac{\partial\ell_{it} (\overline{\alpha}^{t}(\theta,\xi),\theta)}{\partial\alpha})\|$ and $\|{\limfunc{Var}}_{h_i=h,\xi_{i0}=\xi,\lambda_0=\lambda}(\limfunc{vec}\frac{\partial}{\partial \theta'}\big|_{\theta_0}\frac{\partial \ell_{it}(\overline{\alpha}^t(\theta,\xi),\theta)}{\partial\alpha})\|$ are $O(1)$, uniformly in $t$, $\xi$, $\lambda$, $h$, and $\theta$.\label{ass_regu_vb}

	\end{enumerate} 
	
\end{assumption}

In part (\ref{ass_regu_iib}) in Assumption \ref{ass_regu_TV}, we impose a stronger concavity condition than in Assumption \ref{ass_regu}.\footnote{In particular, we use part (\ref{ass_regu_iib}) in Assumption \ref{ass_regu_TV} to establish consistency. Note that this condition can be restrictive in models with time-varying random coefficients.} The other parts are similar to Assumption \ref{ass_regu}, except part (\ref{ass_regu_vb}) where we require regularity of certain conditional expectations and variances.

We next state our second main result, where, differently from Theorem \ref{theo2},  $s_i$ in (\ref{eq_score}) and $H$ in (\ref{eq_hessian}) are now evaluated at $(\theta_0,\alpha_{it0})$, and expectations are conditional on $(\xi_{i0},\lambda_0)$.

\begin{theorem}\label{theo2_TV}Let the conditions of Lemmas \ref{theo1} and \ref{lemma_GL} and Assumptions \ref{ass_alpha}, \ref{ass_inj} and \ref{ass_regu_TV} hold. Then, as $N,T,K$ tend to infinity such that $K/N$ tends to zero, we have:
	\begin{equation}\label{2steptheta_TV}
	\widehat{\theta}=\theta_0+ H^{-1}\frac{1}{N}\sum_{i=1}^Ns_i +O_p\left(\frac{1}{T}\right)+O_p\left(\frac{K}{N}\right)+O_p\left(K^{-\frac{2}{d}}\right)+o_p\left(\frac{1}{\sqrt{NT}}\right).
	\end{equation}

\end{theorem}

Theorem \ref{theo2_TV} shows that GFE is consistent as $N,T,K$ tend to infinity and $K/N$ tends to zero. This requires no parametric assumption about how $\xi_{i0}$ and $\lambda_{t0}$ affect individual and time heterogeneity, unlike additive or interactive FE methods. 

To give intuition, consider the probit model (\ref{eq_ex_prob}) with time-varying unobservables. Under Assumption \ref{ass_inj}, in the first step, GFE consistently estimates an injective function $\varphi_{i0}=\varphi(\xi_{i0})$ of the type. One can then rewrite the outcome equation in (\ref{eq_ex_prob}) as $Y_{it}=\boldsymbol{1}\left\{X_{it}'\theta_0+\alpha(\psi(\varphi_{i0}),\lambda_{t0})+U_{it}\geq 0\right\}$, where $\psi$ is the function introduced in Assumption \ref{ass_inj}, and $\alpha_{it0}=\alpha(\psi(\varphi_{i0}),\lambda_{t0})$ is simply a time-varying function of $\varphi_{i0}$. In the second step, GFE estimates this function by including group-time indicators in the probit regression.

As in Theorem \ref{theo2}, the expansion in Theorem \ref{theo2_TV} features a combination of incidental parameter bias and approximation error. When using the rule (\ref{choice_K_eq}) for $K$, the approximation error is of the same or lower order compared to 1/T. However, the $O_p(K/N)$ term is a new contribution relative to the time-invariant case, which reflects the estimation of $KT$ group-specific parameters using $NT$ observations. As an example, when $d=1$ and $K$ is chosen of the order of $T^{\frac{1}{2}}$, the $O_p$ terms in (\ref{2steptheta_TV}) are $O_p(1/T+T^{\frac{1}{2}}/N)$.\footnote{When $N/T^{\frac{3}{2}}\rightarrow 0$, one could obtain a faster rate in (\ref{2steptheta_TV}) by choosing another rule for $K$.} Although this rate of convergence can be fast when $N$ is sufficiently large relative to $T$, it is too slow to apply conventional bias-reduction methods for inference. In the next section, under the additional assumption that time heterogeneity $\lambda_{t0}$ is low-dimensional, we describe how to obtain a faster convergence rate by grouping both individuals and time periods.

\section{Complements and extensions\label{sec_Extens}}

\subsection{Bias reduction and inference\label{subsec_2nd_ext}}

\noindent In models with time-invariant heterogeneity, Corollary \ref{coro_groups} can be used to characterize the asymptotic distribution of GFE estimators. However, as in FE, the presence of the $O_p(1/T)$ term in (\ref{2steptheta_groups}) shifts the distribution of $\widehat{\theta}$ away from $\theta_0$ whenever $T$ is not large relative to $N$. A variety of methods are available to bias-correct FE estimators and construct asymptotically valid confidence intervals; see Arellano and Hahn (2007) for a review. Consider the setup of Corollary \ref{coro_groups}, under the additional assumption that the $O_p(1/T)$ term in (\ref{2steptheta_groups}) is equal to $C/T+o_p(1/T)$ for some constant $C$. In this case, one can show that half-panel jackknife (Dhaene and Jochmans, 2015) gives asymptotically valid inference based on GFE as $N$ and $T$ tend to infinity at the same rate.\footnote{In particular, half-panel jackknife is valid under the conditions of Corollary \ref{coro_bias}, which requires taking $\gamma=o(1)$ in our rule (\ref{choice_K_eq}) for $K$ in order for the approximation error to be of small order. Deriving primitive conditions for the validity of half-panel jackknife and other bias-reduction methods for other choices of $K$ is left for future work.} The distribution of the bias-corrected GFE estimator is then asymptotically normal centered at the truth, and the asymptotic variance $H^{-1}$ can be consistently estimated by replacing the expectations in (\ref{eq_score}) and (\ref{eq_hessian}) by group-specific means. 

 In settings where heterogeneity varies over time, it can be desirable to group not only individuals as in (\ref{eq_firststep}), but also time periods (or alternatively counties or markets, depending on the application). We now describe such a method, and discuss its potential for performing inference in models with time-varying heterogeneity. In the \emph{two-way GFE} approach, we classify time periods based on cross-sectional moments $w_t=\frac{1}{N}\sum_{i=1}^N w(Y_{it},X_{it})$, and compute:
\begin{equation}\label{eq_firststep2}
\left(\widehat{w}(1),...,\widehat{w}(L),\widehat{l}_1,...,\widehat{l}_T\right)=\underset{\left(\widetilde{w}(1),...,\widetilde{w}(L),l_1,...,l_T\right)}{\limfunc{argmin}}\,\,\,\sum_{t=1}^T\big\|w_t-\widetilde{w}(l_t)\big\|^2,
\end{equation} 
where $\{l_t\}$ are partitions of $\{1,...,T\}$ into $L$ groups. Given the group indicators $\widehat{k}_i$ and $\widehat{l}_t$, we then maximize $\sum_{i=1}^N\sum_{t=1}^T\ln f(Y_{it}\,|\, X_{it},\alpha(\widehat{k}_i,\widehat{l}_t),\theta)$, with respect to $\theta$ and the $KL$ group-specific parameters $\alpha(k,l)$. 

Two-way GFE estimators can be expanded similarly to Theorem \ref{theo2_TV}, under two main additional assumptions: the model is \emph{static} and observations are independent across $i$ and $t$, and the dimensions $d_{\lambda}$ of time heterogeneity $\lambda_{t0}$ and $d$ of individual heterogeneity $\xi_{i0}$ are \emph{both} small. Then, for $s_{i}$ and $H$ as in Theorem \ref{theo2_TV}, we show in the supplemental material that: 
\begin{align*}
\widehat{\theta}=\theta_0+ H^{-1}\frac{1}{N}\sum_{i=1}^Ns_{i} {+}O_p\left(\frac{1}{T}{+}\frac{1}{N}{+}\frac{KL}{NT}\right){+}O_p\left(K^{-\frac{2}{d}}+L^{-\frac{2}{d_{\lambda}}}\right){+}o_p\left(\frac{1}{\sqrt{NT}}\right).
\end{align*} 
Suppose $d=d_{\lambda}=1$, and $K$ is given by (\ref{choice_K_eq}) with $\gamma$ asymptotically constant, with an analogous choice for $L$. Then the $O_p$ term in this expansion can be shown to be $O_p(1/T+1/N)$. We leave to future work the formal study of the validity of bias reduction methods for inference, such as two-way split panel jackknife (Fern\'andez-Val and Weidner, 2016), as $N$ and $T$ tend to infinity at the same rate. 
\subsection{GFE with conditional moments\label{subsec_1st_ext}}

\noindent Our theory shows that the dimension $d$ of heterogeneity plays a key role in the properties of GFE. While models with scalar latent types $\xi_{i0}$, such as model (\ref{eq_ex_prob2}) of wages and labor force participation, are not uncommon in economics, many applications involve conditioning covariates. Under Assumptions \ref{ass_alpha} and \ref{ass_inj}, the moments $h_i$ should, asymptotically, be injective functions of all the heterogeneity coming from both $Y_i$ and $X_i$. However, when $X_i$ depends on multiple components of heterogeneity, this might lead to a large dimension $d$.  

We now show that GFE can still perform well under a weaker form of injectivity. Consider the case where Assumption \ref{ass_alpha} is replaced by $\alpha_{i0}={\alpha}(\xi_{i0})$ and $\mu_{i0}={\mu}(\xi_{i0},\nu_{i0})$, where $\nu_{i0}$ is another latent component that affects covariates. Moreover, instead of requiring injectivity for both $\xi_{i0}$ and $\nu_{i0}$, let us maintain Assumption \ref{ass_inj}, which only requires $h_i$ to be injective for $\xi_{i0}$. In other words, $h_i$ needs to be directly informative about the unobserved heterogeneity component $\xi_{i0}$ that appears in the \emph{conditional distribution} of $Y_i$ given $X_i$. We show in the supplemental material that, under regularity conditions otherwise similar to those of Corollary \ref{coro_groups}, the convergence rate of GFE is unaffected by the dimension of $\nu_{i0}$. Specifically, when $K=\widehat{K}$ is given by (\ref{choice_K_eq}) with $\gamma=O(1)$ (which adapts to the dimension of $\xi_{i0}$ and not the one of $\nu_{i0}$), we have:
\begin{equation}\label{eq_condGFE_rate}
\widehat{\theta}=\theta_0 +O_p\left(\frac{1}{T}\right)+O_p\left(\frac{1}{\sqrt{NT}}\right).
\end{equation} 
To prove (\ref{eq_condGFE_rate}) we assume that the rate condition $T^{1+\frac{d}{2}}=O(N)$ holds, where $d$ is the (small) dimension of $\xi_{i0}$.\footnote{In the supplemental material, we provide an asymptotic expansion for GFE in a linear homoskedastic model under a small approximation error, as in Corollary \ref{coro_bias}. The argument requires no restriction on the relative rates of $N$ and $T$. Interestingly, in this case the asymptotic variances of GFE and FE differ, since the within-group variation in $\nu_{i0}$ tends to decrease the variance, yet the expansion features an additional score term compared to Theorem \ref{theo2}.}

In models with time-varying conditioning covariates, a simple way to target moments to $\xi_{i0}$ is to construct $h_i$ using the {conditional distribution} of $Y_{i}$ given $X_{i}$. To see this, consider a static model $f(Y_{it}\,|\, X_{it},\alpha_{i0},\theta_0)$ where $X_{it}$ has finite support. In this case, we have under appropriate conditions:
$$\underset{=h_i(x)}{\underbrace{\frac{\sum_{t=1}^T\boldsymbol{1}\{X_{it}=x\}h(Y_{it},X_{it})}{\sum_{t=1}^T\boldsymbol{1}\{X_{it}=x\}}}}=\underset{=\varphi(x,\xi_{i0})}{\underbrace{\mathbb{E}_{X_{it}=x,\xi_{i0}}[h(Y_{it},X_{it})]}}+o_p\left(1\right),$$
where $h_i(x)$ is only defined when $\sum_{t=1}^T\boldsymbol{1}\{X_{it}=x\}\neq 0$, and, importantly, $\varphi(x,\xi_{i0})$ does not depend on $\nu_{i0}$. In the supplemental material we discuss implementation, and we report simulation results in a probit model with binary covariates. We find that using conditional moments can enhance the performance of GFE in such settings. We leave the analysis of conditional moments in the presence of continuous covariates to future work. 

\section{Conclusion\label{ConcSec}}

In this paper, we analyze some properties of two-step grouped fixed-effects (GFE) methods in settings where population heterogeneity is not discrete. Our framework relies on two main assumptions: low-dimensional individual heterogeneity, and the availability of moments to approximate the latent types. In many economic models, individual types are low-dimensional. By taking advantage of this feature, GFE can allow for flexible forms of heterogeneity across individuals and over time. 

GFE methods are of interest in various applied settings. In a previous version of this paper, we used two-step GFE to estimate a dynamic structural model of location choice in the spirit of Kennan and Walker (2011), and we analyzed the performance of the discrete estimator of Bonhomme \textit{et al.} (2019) for matched employer-employee data in the presence of continuous firm heterogeneity. Other potential applications include nonlinear factor models, nonparametric and semi-parametric panel data models such as quantile regression with individual effects, and network models.

\begin{singlespace}

\end{singlespace}

\appendix

\setcounter{section}{0}\renewcommand{\thesection}{A\arabic{section}}

\setcounter{figure}{0}\renewcommand{\thefigure}{A\arabic{figure}}

\setcounter{table}{0}\renewcommand{\thetable}{A\arabic{table}}

\setcounter{footnote}{0}\renewcommand{\thefootnote}{\arabic{footnote}}

\setcounter{assumption}{0}\renewcommand{\theassumption}{A\arabic{assumption}}

\setcounter{equation}{0}\renewcommand{\theequation}{A\arabic{equation}}

\setcounter{lemma}{0}\renewcommand{\thelemma}{A\arabic{lemma}}

\setcounter{proposition}{0}\renewcommand{\theproposition}{A\arabic{proposition}}

\setcounter{corollary}{0}\renewcommand{\thecorollary}{A\arabic{corollary}}

\setcounter{theorem}{0}\renewcommand{\thetheorem}{A\arabic{theorem}}

\vskip 1cm

\begin{center}
	{ {\LARGE APPENDIX} }
\end{center}


\paragraph{\underline{Proof of Lemma \ref{theo1}.}} Define $B_{\varphi(\xi)}(K)={\limfunc{min}}_{\left(\widetilde{h},\{k_i\}\right)} \frac{1}{N}\sum_{i=1}^N \|\varphi(\xi_{i0})-\widetilde{h}(k_i)\|^2$, similarly to (\ref{eq_approx_error}), and denote: 
$(\underline{h},\{\underline{k}_i\})={\limfunc{argmin}}_{\left(\widetilde{h},\{k_i\}\right)} \sum_{i=1}^N\|\varphi(\xi_{i0})-\widetilde{h}(k_i)\|^2$. By definition of $(\widehat{h},\{\widehat{k}_i\})$, we have: $\sum_{i=1}^N \|h_i-\widehat{h}(\widehat{k}_i)\|^2\leq \sum_{i=1}^N \left\|h_i-\underline{h}(\underline{k}_i)\right\|^2$ (almost surely). Letting ${\varepsilon}_{i}=h_i-\varphi(\xi_{i0})$, we thus have, using the triangle inequality twice:
\begin{align*}
&\frac{1}{N}\sum_{i=1}^N \left\|\varphi(\xi_{i0})-\widehat{h}(\widehat{k}_i)\right\|^2\leq
\frac{2}{N}\sum_{i=1}^N\left\|h_i-\widehat{h}(\widehat{k}_i)\right\|^2+\frac{2}{N}\sum_{i=1}^N \left\|h_i-\varphi(\xi_{i0})\right\|^2\\
&{\leq} 	\frac{2}{N}\sum_{i=1}^N\left\|h_i-\underline{h}(\underline{k}_i)\right\|^2+\frac{2}{N}\sum_{i=1}^N \left\|\varepsilon_i\right\|^2{\leq} 		
4\underset{= B_{\varphi(\xi)}(K)}{\underbrace{\left(\frac{1}{N} \sum_{i=1}^N\left\|\varphi(\xi_{i0})-\underline{h}(\underline{k}_i)\right\|^2\right)}}{+}\frac{6}{N}\sum_{i=1}^N\|{\varepsilon}_{i}\|^2.
\end{align*}
By Assumption \ref{ass_inj}, $\frac{1}{N}\sum_{i=1}^N\|{\varepsilon}_{i}\|^2=O_p(1/T)$. In addition, since $\varphi$ is Lipschitz-continuous, there exists a constant $\tau$ such that $\|\varphi(\xi')-\varphi(\xi)\|\leq \tau \|\xi'-\xi\|$ for all $(\xi,\xi')$. This implies that $B_{\varphi(\xi)}(K)\leq {\tau^2} B_{\xi}(K)$, and Lemma \ref{theo1} follows.

\paragraph{\underline{Proofs of Theorems \ref{theo2} and \ref{theo2_TV}.}} It is convenient to use a common notation for Theorems \ref{theo2} and \ref{theo2_TV}. Let $p$ denote the number of individual-specific vectors $\alpha_i^j$, $j\in\{1,...,p\}$. In the time-invariant case: $p=1$, $j=1$, and $\alpha_{i}^j=\alpha_i$. In the time-varying case: $p=T$, $j\in\{1,...,T\}$, and $\alpha_{i}^j=\alpha_{it}$. Denote $\ell_{ij}= \ell_{i}$ in the time-invariant case, and $\ell_{ij}= \ell_{it}$ in the time-varying case. Let $v_{ij}=\frac{\partial\ell_{ij}}{\partial \alpha}$, $v_{ij}^{\alpha}=\frac{\partial^2\ell_{ij}}{\partial \alpha\partial \alpha'}$, $v_{ij}^{\theta}=\frac{\partial^2\ell_{ij}}{\partial \theta\partial \alpha'}$, and $v_{ij}^{\alpha\alpha}=\frac{\partial^3\ell_{ij}}{\partial \alpha\partial \alpha'\otimes \partial \alpha'}$ (which is a $\limfunc{dim}\alpha_{i0}^j\times (\limfunc{dim}\alpha_{i0}^j)^2$ matrix). Let, for all $\theta\in\Theta$, $j\in \{1,...,p\}$, and $k\in\{1,...,K\}$, $\widehat{\alpha}^j(k,\theta){=}{\limfunc{argmax}}_{\alpha}\sum_{i=1}^N\boldsymbol{1}\{\widehat{k}_i=k\}\ell_{ij}\left(\alpha,\theta\right)$. Likewise, denote $\overline{\alpha}^j(\theta,\xi){=}{\limfunc{argmax}}_{\alpha}\mathbb{E}_{\xi_{i0}=\xi,\lambda_{0}=\lambda}(\ell_{ij}(\alpha,\theta))$. We will index expectations by $\xi_{i0}$ and $\lambda_0$, although the conditioning on $\lambda_0$ is not needed in the time-invariant case of Theorem \ref{theo2}. Finally, let $\delta=\frac{1}{T}+K^{-\frac{2}{d}}$ in the time-invariant case, and let $\delta=\frac{1}{T}+\frac{K}{N}+K^{-\frac{2}{d}}$ in the time-varying case. 

\underline{To show consistency} of $\widehat{\theta}$, we first establish the next technical lemma (see the supplemental material for the proof):
\begin{lemma}\label{lem_sup1}
	Under the conditions of either Theorem \ref{theo2} or Theorem \ref{theo2_TV} we have:
	\begin{align}\label{Rate_alpha_hat_theta} &\frac{1}{Np}\sum_{i=1}^N\sum_{j=1}^p\left\|\widehat{\alpha}^j(\widehat{k}_i,\theta)-\overline{\alpha}^{j}(\theta,\xi_{i0})\right\|^2=O_p(\delta), \quad \forall \theta\in\Theta, \\
	&\sup_{\theta\in\Theta}\,\frac{1}{Np}\sum_{i=1}^N\sum_{j=1}^p\left\|\widehat{\alpha}^j(\widehat{k}_i,\theta)-\overline{\alpha}^{j}(\theta,\xi_{i0})\right\|^2=o_p\left(1\right).\label{eq_alpha_hat_sup}
	\end{align}	
\end{lemma}
\noindent From (\ref{eq_alpha_hat_sup}) we then verify using a Taylor expansion that:
\begin{equation*}
\sup_{\theta\in\Theta}\,\left|\frac{1}{Np}\sum_{i=1}^N\sum_{j=1}^p\ell_{ij}\left(\widehat{\alpha}^j(\widehat{k}_i,\theta),\theta\right)-\frac{1}{Np}\sum_{i=1}^N\sum_{j=1}^p\ell_{ij}\left(\overline{\alpha}^{j}\left(\theta,\xi_{i0}\right),\theta\right)\right|=o_p(1).
\end{equation*}
Consistency of $\widehat{\theta}$ then follows by standard arguments.

Next, the \underline{two key steps} in the proof consist in showing the following two expansions:
\begin{align}
&\frac{1}{Np}\sum_{i=1}^N\sum_{j=1}^p\frac{\partial \ell_{ij}(\widehat{\alpha}^j(\widehat{k}_i,\theta_0),\theta_0) }{\partial \theta}{=}\frac{1}{Np}\sum_{i=1}^N\sum_{j=1}^p\frac{\partial}{\partial\theta}\bigg|_{\theta_0}\, \ell_{ij}\left(\overline{\alpha}^{j}(\theta,\xi_{i0}),\theta\right){+}O_p\left(\delta\right),\label{maineq1} \\
&\frac{1}{Np}\sum_{i=1}^N\sum_{j=1}^p\frac{\partial^2}{\partial \theta\partial \theta'}\bigg|_{\theta_0}\,  \left(\ell_{ij}\left(\widehat{\alpha}^j(\widehat{k}_i,\theta),\theta\right)-\ell_{ij}\left(\overline{\alpha}^j(\theta,\xi_{i0}),\theta\right)\right)=o_p(1).\label{maineq2}
\end{align}

\underline{To show (\ref{maineq1})}, we show the following technical lemma, where we omit references to the evaluation points $\theta_0$ and $\alpha_{i0}^j$ for conciseness:
\begin{lemma}\label{lem_sup2}
	Under the conditions of either Theorem \ref{theo2} or Theorem \ref{theo2_TV} we have:
	\begin{align*}
	&\frac{1}{Np}\sum_{i=1}^N\sum_{j=1}^p\mathbb{E}_{\xi_{i0},\lambda_0}\left(v_{ij}^{\theta}\right)\left[\mathbb{E}_{\xi_{i0},\lambda_0}\left(v_{ij}^{\alpha}\right)\right]^{-1}v_{ij}^{\alpha}\left(\widehat{\alpha}^j(\widehat{k}_i,\theta_0)-\alpha_{i0}^j+(v_{ij}^{\alpha})^{-1}v_{ij}\right)=O_p(\delta),\\
	&\frac{1}{Np}\sum_{i=1}^N\sum_{j=1}^p\left(v_{ij}^{\theta}\left(v_{ij}^{\alpha}\right)^{-1}{-}\mathbb{E}_{\xi_{i0},\lambda_0}\left(v_{ij}^{\theta}\right)\left[\mathbb{E}_{\xi_{i0},\lambda_0}\left(v_{ij}^{\alpha}\right)\right]^{-1}\right)v_{ij}^{\alpha}\left(\widehat{\alpha}^j(\widehat{k}_i,\theta_0){-}\alpha_{i0}^j\right){=}O_p(\delta).
	\end{align*}
\end{lemma}
\noindent Now, expanding $v_{ij}^{\theta}(\widehat{\alpha}_j(\widehat{k}_i,\theta_0),\theta_0)$ around $\overline{\alpha}^{j}(\theta_0,\xi_{i0}){=}\alpha_{i0}^j$, and using the identity $\frac{\partial \overline{\alpha}^{j}(\theta_0,\xi_{i0})}{\partial \theta'}{=}\left[\mathbb{E}_{\xi_{i0},\lambda_0}\left(-v_{ij}^{\alpha}\right)\right]^{-1}\mathbb{E}_{\xi_{i0},\lambda_0}\left(v_{ij}^{\theta}\right)'$, we obtain:
\begin{align*}
&\frac{1}{Np}\sum_{i=1}^N\sum_{j=1}^p\frac{\partial \ell_{ij}(\widehat{\alpha}^j(\widehat{k}_i,\theta_0),\theta_0) }{\partial \theta}-\frac{1}{Np}\sum_{i=1}^N\sum_{j=1}^p\frac{\partial}{\partial\theta}\bigg|_{\theta_0}\, \ell_{ij}\left(\overline{\alpha}^{j}(\theta,\xi_{i0}),\theta\right)\\
& {=}\frac{1}{Np}\sum_{i=1}^N\sum_{j=1}^p\left\{v_{ij}^{\theta}\left(\widehat{\alpha}^j(\widehat{k}_i,\theta_0){-}\alpha_{i0}^j\right){+}\mathbb{E}_{\xi_{i0},\lambda_0}\left(v_{ij}^{\theta}\right)\left[\mathbb{E}_{\xi_{i0},\lambda_0}\left(v_{ij}^{\alpha}\right)\right]^{-1}v_{ij}\right\}+O_p(\delta),
\end{align*}and summing the two parts in Lemma \ref{lem_sup2} shows that the last expression is $O_p(\delta)$. It follows that (\ref{maineq1}) is satisfied.

\underline{To show (\ref{maineq2})}, we show the next technical lemma:
\begin{lemma}\label{lem_sup3}
	Under the conditions of either Theorem \ref{theo2} or Theorem \ref{theo2_TV} we have:
	\begin{equation}\frac{1}{Np}\sum_{i=1}^N\sum_{j=1}^p\left\|\frac{\partial \widehat{\alpha}^j(\widehat{k}_i,\theta_0)}{\partial \theta'}-\frac{\partial \overline{\alpha}^j(\theta_0,\xi_{i0})}{\partial \theta'}\right\|^2=o_p\left(1\right).\label{eq_der_alpha}\end{equation}
\end{lemma}
\noindent Using (\ref{Rate_alpha_hat_theta}) and the identity $\frac{\partial \overline{\alpha}^{j}(\theta_0,\xi_{i0})}{\partial \theta'}{=}\left[\mathbb{E}_{\xi_{i0},\lambda_0}\left(-v_{ij}^{\alpha}\right)\right]^{-1}\mathbb{E}_{\xi_{i0},\lambda_0}\left(v_{ij}^{\theta}\right)'$, we thus have, under the conditions of either Theorem \ref{theo2} or \ref{theo2_TV}:
\begin{align*}
&	\frac{1}{Np}\sum_{i=1}^N\sum_{j=1}^p\frac{\partial^2}{\partial \theta\partial \theta'}\bigg|_{\theta_0}  \ell_{ij}\left(\widehat{\alpha}^j(\widehat{k}_i,\theta),\theta\right)-\frac{1}{Np}\sum_{i=1}^N\sum_{j=1}^p\frac{\partial^2}{\partial \theta\partial \theta'}\bigg|_{\theta_0}\ell_{ij}\left(\overline{\alpha}^{j}(\theta,\xi_{i0}),\theta\right)\\&=\frac{1}{Np}\sum_{i=1}^N \sum_{j=1}^pv_{ij}^{\theta}\left(\frac{\partial \widehat{\alpha}^j(\widehat{k}_i,\theta_0)}{\partial \theta'}-\frac{\partial \overline{\alpha}^{j}(\theta_0,\xi_{i0})}{\partial \theta'}\right)+o_p\left(1\right)=o_p(1),
\end{align*}
where we have used Lemma \ref{lem_sup3} in the last equality.

Finally, \underline{to show Theorems \ref{theo2} and \ref{theo2_TV}} we expand the GFE score as:
\begin{align*}
\frac{1}{Np}\sum_{i=1}^N\sum_{j=1}^p\frac{\partial \ell_{ij}(\widehat{\alpha}^j(\widehat{k}_i,\theta_0),\theta_0) }{\partial \theta}{+}\left(\frac{\partial }{\partial\theta'}\Big|_{\widetilde{\theta}} \frac{1}{Np}\sum_{i=1}^N\sum_{j=1}^p \frac{\partial \ell_{ij}(\widehat{\alpha}^j(\widehat{k}_i,\theta),{\theta})}{\partial \theta}\right)\left(\widehat{\theta}{-}\theta_0\right){=}0,
\end{align*}
where $\widetilde{\theta}$ lies between $\theta_0$ and $\widehat{\theta}$, and further expand $\frac{\partial }{\partial\theta'}\big|_{\widetilde{\theta}} \frac{1}{Np}\sum_{i=1}^N \sum_{j=1}^p\frac{\partial \ell_{ij}(\widehat{\alpha}^j(\widehat{k}_i,\theta),{\theta})}{\partial \theta}$ around $\theta_0$ using that $\widetilde{\theta}$ is consistent. Lastly, we use (\ref{maineq1}) and (\ref{maineq2}), and note that, if $\overline{\ell}_i(\theta)=\frac{1}{p}\sum_{j=1}^p\ell_{ij}\left(\overline{\alpha}^{j}\left(\theta,\xi_{i0}\right),\theta\right)$ denotes the individual target log-likelihood, then $s_i=\frac{\partial \overline{\ell}_i(\theta_0) }{\partial\theta}$ and $H={\limfunc{plim}}_{N,T\rightarrow \infty} \frac{1}{N}\sum_{i=1}^N\mathbb{E}_{\xi_{i0},\lambda_0}(-\frac{\partial^2 \overline{\ell}_i(\theta_0) }{\partial\theta\partial\theta'})$.

\paragraph{\underline{Proof of Corollary \ref{coro_groups}.}} By the triangle inequality:
$\frac{1}{N}\sum_{i=1}^N  \|\widehat{h}(\widehat{k}_i)-\varphi(\xi_{i0})\|^2\leq 2\widehat{Q}(K)+O_p(\frac{1}{T})=O_p(\frac{1}{T})$. The proof of Theorem \ref{theo2} is then unchanged, simply redefining $\delta{=}1/T$ (since heterogeneity is time-invariant here). This shows (\ref{2steptheta_groups}).

\paragraph{\underline{Proof of Corollary \ref{coro_bias}.}} To prove Corollary \ref{coro_bias}, we follow a likelihood approach (see Arellano and Hahn, 2007). Consider the difference between the GFE and FE profile log-likelihoods:
$\Delta L(\theta)=\frac{1}{N}\sum_{i=1}^N\ell_i(\widehat{\alpha}(\widehat{k}_i,\theta),\theta)-\frac{1}{N}\sum_{i=1}^N\ell_i(\widehat{\alpha}_i(\theta),\theta)$. 
\begin{assumption}{(regularity)} \label{ass_coro_bias} Let $\widehat{w}_i=-\frac{\partial^2\ell_i(\widehat{\alpha}_i(\theta_0),\theta_0)}{\partial\alpha\partial\alpha'}$, and $\widehat{g}_i=\frac{\partial^2\ell_i(\widehat{\alpha}_i(\theta_0),\theta_0)}{\partial\theta\partial\alpha'}\widehat{w}_i^{-1}$.
	\begin{enumerate}[itemsep=-3pt,label=(\roman*),ref=\roman*,topsep=0pt]
		\item $\ell_{it}(\alpha_i,\theta)$ is four times differentiable, and its fourth derivatives satisfy similar properties to the first three.\label{corotheo2i}
		\item $\gamma(h){=}\{\mathbb{E}_{h_i=h}\left(\widehat{w}_i\right)\}^{-1}\mathbb{E}_{h_i=h}\left(\widehat{w}_i\widehat{\alpha}_i(\theta_0)\right)$ and $\lambda(h){=}\mathbb{E}_{h_i=h}\left(\widehat{g}_i\widehat{w}_i\right)\{\mathbb{E}_{h_i=h}\left(\widehat{w}_i\right)\}^{-1}$ are Lipschitz-continuous in $h$; and $\limfunc{Var}_{h_i=h}\left(\widehat{w}_i(\widehat{\alpha}_i(\theta_0)-\gamma(h_i))\right)=O(\frac{1}{T})$ and $\limfunc{Var}_{h_i=h}\left((\widehat{g}_i-\lambda(h_i))\widehat{w}_i\right)=O(\frac{1}{T})$, uniformly in $h$.\label{corotheo2iii}
	\end{enumerate}

\end{assumption}

\begin{lemma}\label{lem_coro_bias}Let the conditions of Corollary \ref{coro_bias} hold, and let $\nu_i(\theta){=}\widehat{\alpha}_i(\theta){-}{\mathbb{E}}_{h_i}(\widehat{\alpha}_i(\theta))$. We have: 
	\begin{equation}
	\frac{\partial}{\partial\theta}\Big|_{\theta_0}\, \Delta L(\theta){=}-\frac{\partial}{\partial\theta}\Big|_{\theta_0}\, \frac{1}{2N}\sum_{i=1}^N\nu_i(\theta)'\mathbb{E}_{\xi_{i0}}\left[-v_i^\alpha \left(\overline{\alpha}(\theta,\xi_{i0}),\theta\right)\right] \nu_i(\theta)+o_p\left(\frac{1}{T}\right).\label{eq_Delta_L}
	\end{equation}
\end{lemma}
\noindent Corollary \ref{coro_bias} follows, since the bias of the FE score is:
$\frac{\partial}{\partial\theta}\big|_{\theta_0}\,\big[\frac{1}{N}\sum_{i=1}^N\ell_i(\widehat{\alpha}_i(\theta),\theta)-\frac{1}{N}\sum_{i=1}^N\ell_i(\overline{\alpha}(\theta,\xi_{i0}),\theta)\big] = \frac{\partial}{\partial\theta}\big|_{\theta_0}\frac{1}{2N}\sum_{i=1}^N\widehat{\nu}_{i}(\theta)'\mathbb{E}_{\xi_{i0}}[-v_i^\alpha (\overline{\alpha}(\theta,\xi_{i0}),\theta)]\widehat{\nu}_{i}(\theta)+o_p(\frac{1}{T})$, where $\widehat{\nu}_{i}(\theta)=\widehat{\alpha}_i(\theta)-{\mathbb{E}}_{\xi_{i0}}\left(\widehat{\alpha}_i(\theta)\right)$; see, e.g., Arellano and Hahn (2007).

\clearpage

\setcounter{section}{0}\renewcommand{\thesection}{S\arabic{section}}

\setcounter{figure}{0}\renewcommand{\thefigure}{S\arabic{figure}}

\setcounter{table}{0}\renewcommand{\thetable}{S\arabic{table}}

\setcounter{footnote}{0}\renewcommand{\thefootnote}{\arabic{footnote}}

\setcounter{assumption}{0}\renewcommand{\theassumption}{S\arabic{assumption}}

\setcounter{equation}{0}\renewcommand{\theequation}{S\arabic{equation}}

\setcounter{lemma}{0}\renewcommand{\thelemma}{S\arabic{lemma}}

\setcounter{proposition}{0}\renewcommand{\theproposition}{S\arabic{proposition}}

\setcounter{corollary}{0}\renewcommand{\thecorollary}{S\arabic{corollary}}

\setcounter{theorem}{0}\renewcommand{\thetheorem}{S\arabic{theorem}}

\begin{center}
	{\LARGE SUPPLEMENTAL MATERIAL \\``Discretizing Unobserved Heterogeneity''}
\end{center}

\section{Proofs of technical lemmas}

\paragraph{\underline{Lemma \ref{lem_sup1}.}}

From Assumption \ref{ass_regu} $(\ref{ass_regu_iii})$-$(\ref{ass_regu_iv})$ or \ref{ass_regu_TV} $(\ref{ass_regu_iiib})$-$(\ref{ass_regu_ivb})$, both $\frac{\partial \overline{\alpha}^j\left(\theta,{\xi}\right)}{\partial \theta'}$ and $\frac{\partial \overline{\alpha}^j\left(\theta,{\xi}\right)}{\partial {\xi}^{\prime}}$ are uniformly bounded (in probability in the time-varying case). Let $a^j(k,\theta)=\overline{\alpha}^j(\theta,\psi(\widehat{h}(k)))$. We thus have, using Lemmas \ref{theo1} and \ref{lemma_GL}:
\begin{align}
&{\sup_{\theta\in\Theta}} \frac{1}{Np}\sum_{i,j}\left\|a^j(\widehat{k}_i,\theta){-}\overline{\alpha}^{j}(\theta,\xi_{i0})\right\|^2 {=}{\sup_{\theta\in\Theta}}\frac{1}{Np}\sum_{i,j}\left\|\overline{\alpha}^j(\theta,\psi(\widehat{h}(\widehat{k}_i))){-}\overline{\alpha}^j(\theta,\psi(\varphi(\xi_{i0})))\right\|^2\notag\\
&{=}O_p\left(\frac{1}{N}\sum_{i}\|\widehat{h}(
\widehat{k}_i)-\varphi(\xi_{i0})\|^2\right){=}O_p(\delta). \label{Rate_a1}
\end{align}

Let $\theta\in\Theta$. Expanding: $\sum_{i,j}\ell_{ij}(a^j(\widehat{k}_i,\theta),\theta)\leq \sum_{i,j}\ell_{ij}(\widehat{\alpha}^j(\widehat{k}_i,\theta),\theta)$ to second order around $\overline{\alpha}^{j}(\theta,\xi_{i0})$, and using:
\begin{align}{\max}_{i,j}\,{\sup}_{(\alpha,\theta)}\,\| v_{ij}^{\alpha}(\alpha,\theta)\|=O_p(1),\label{eq_bounded_hessian}\end{align}
we have, for some $a_{ij}(\theta)$ between $\widehat{\alpha}^j(\widehat{k}_i,\theta)$ and $\overline{\alpha}^{j}(\theta,\xi_{i0})$:
\begin{align} &\frac{1}{2Np}\sum_{i,j}\left(\widehat{\alpha}^j(\widehat{k}_i,\theta)-\overline{\alpha}^{j}(\theta,\xi_{i0})\right)'[-v_{ij}^{\alpha}(a_{ij}(\theta),\theta)]\left(\widehat{\alpha}^j(\widehat{k}_i,\theta)-\overline{\alpha}^{j}(\theta,\xi_{i0})\right)\notag\\
&\leq \frac{1}{Np}\sum_{i,j}v_{ij}(\overline{\alpha}^{j}(\theta,\xi_{i0}),\theta)'\left(\widehat{\alpha}^j(\widehat{k}_i,\theta)-a^j(\widehat{k}_i,\theta)\right)+O_p(\delta)\notag\\
&=\frac{1}{Np}\sum_{i,j}\overline{v}_j(\widehat{k}_i,\theta)'\left(\widehat{\alpha}^j(\widehat{k}_i,\theta)-a^j(\widehat{k}_i,\theta)\right)+O_p(\delta),\label{eq_score_exp}
\end{align}
where $\overline{v}_j(k,\theta)$ denotes the mean over $i$ of $v_{ij}(\overline{\alpha}^{j}(\theta,\xi_{i0}),\theta)$ in group $\widehat{k}_i=k$, and the $O_p(\delta)$ terms are uniform in $\theta$ by (\ref{Rate_a1}).

Now, by Assumption \ref{ass_regu} (\ref{ass_regu_iia}) or \ref{ass_regu_TV} (\ref{ass_regu_iib}) there exists a constant $\underline{c}>0$ such that:
\begin{equation}
{\min}_{i,j}\,{\inf}_{(\alpha,\theta)}\, \limfunc{mineig}\left[-v_{ij}^{\alpha}(\alpha,\theta)\right]\geq \underline{c}+o_p(1),\label{eq_mineig}
\end{equation} 
where $\limfunc{mineig}(M)$ is the minimum eigenvalue of $M$. Let $A=\frac{1}{Np}\sum_{i,j}\|\widehat{\alpha}^j(\widehat{k}_i,\theta)-\overline{\alpha}^{j}(\theta,\xi_{i0})\|^2$. By (\ref{eq_score_exp}) and the Cauchy Schwarz inequality, we have:\\
$$A\leq O_p\left[\left(\frac{1}{Np}\sum_{i,j}\left\| \overline{v}_j(\widehat{k}_i,\theta)\right\|^2\right)^{\frac{1}{2}}\left(\frac{1}{Np}\sum_{i,j}\left\|\widehat{\alpha}^j(\widehat{k}_i,\theta)-a^j(\widehat{k}_i,\theta)\right\|^2\right)^{\frac{1}{2}}\right]+O_p(\delta).$$ By (\ref{Rate_a1}) and the triangle inequality: $(\frac{1}{Np}\sum_{i,j}\|\widehat{\alpha}^j(\widehat{k}_i,\theta)-a^j(\widehat{k}_i,\theta)\|^2)^{\frac{1}{2}}\leq {A}^{\frac{1}{2}}+O_p({\delta}^{\frac{1}{2}})$. Hence:
$A=O_p\left[ \left(\frac{1}{Np}\sum_{i,j}\|\overline{v}_j(\widehat{k}_i,\theta)\|^2\right)^{\frac{1}{2}}\left(A^{\frac{1}{2}}+O_p({\delta}^{\frac{1}{2}})\right)\right]+O_p(\delta)$,
which implies:
\begin{equation}A= O_p\left(\frac{1}{Np}\sum_{i,j}\|\overline{v}_j(\widehat{k}_i,\theta)\|^2\right)+O_p(\delta).\label{eq_inter}\end{equation}

We are now going to show that, for all $\theta\in\Theta$: 
\begin{align}\label{rate_vbar} &\frac{1}{Np}\sum_{i,j}\left\| \overline{v}_j(\widehat{k}_i,\theta)\right\|^2=O_p\left(\delta\right).
\end{align}
Using (\ref{eq_inter}) and (\ref{rate_vbar}) will then imply (\ref{Rate_alpha_hat_theta}). Under the conditions of Theorem \ref{theo2}, it is easy to see that (\ref{rate_vbar}) holds. We are now going to show (\ref{rate_vbar}) under the conditions of Theorem \ref{theo2_TV}. Let, for all $j,\theta,h,\xi,\lambda$: $\rho_{j}(h,\xi,\lambda,\theta)=\mathbb{E}_{h_i=h,\xi_{i0}=\xi,\lambda_0=\lambda}( v_{ij}(\overline{\alpha}^{j}(\theta,\xi),\theta))$, and, for all $i,j,\theta$: $\zeta_{ij}(\theta)=v_{ij}(\overline{\alpha}^{j}(\theta,\xi_{i0}),\theta)- \rho_{j}(h_i,\xi_{i0},\lambda_0,\theta)$.
By Assumption \ref{ass_regu_TV} (\ref{ass_regu_vb}), and letting $h_i=\varphi(\xi_{i0})+\varepsilon_i$, we can expand $\rho_{j}(h_i,\xi_{i0},\lambda_0,\theta)$ twice around $\varphi(\xi_{i0})$ as:
$\rho_{j}(\varphi(\xi_{i0}),\xi_{i0},\lambda_0,\theta)+\frac{\partial\rho_{j}(\varphi(\xi_{i0}),\xi_{i0},\lambda_0,\theta)}{\partial h'}{\varepsilon}_i+\frac{1}{2}{\varepsilon}_i'\frac{\partial^2\rho_{j}(a_{i\theta}^j,\xi_{i0},\lambda_0,\theta)}{\partial h\partial h'}{\varepsilon}_i$,
where $a_{i\theta}^j$ lies between $h_i$ and $\varphi(\xi_{i0})$. Hence, taking expectations, using that $\mathbb{E}_{\xi_{i0},\lambda_0}\left[\rho_{j}(h_i,\xi_{i0},\lambda_0,\theta)\right]=0$, and using Assumptions \ref{ass_inj} and \ref{ass_regu_TV} (\ref{ass_regu_vb}), we have: $$\frac{1}{Np}\sum_{i,j}\| \rho_{j}(\varphi(\xi_{i0}),\xi_{i0},\lambda_{0},\theta)\|^2{=}\frac{1}{Np}\sum_{i,j}\left\| \frac{\partial\rho_{j}(\varphi(\xi_{i0}),\xi_{i0},\lambda_0,\theta)}{\partial h'}\mathbb{E}_{\xi_{i0},\lambda_0}\left[{\varepsilon}_i\right]\right\|^2{+}o_p\left(\frac{1}{T}\right),$$
which is $O_p(\frac{1}{T})$. Hence: $\frac{1}{Np}\sum_{i,j}\| \rho_{j}(h_i,\xi_{i0},\lambda_{0},\theta)\|^2=O_p(\frac{1}{T})$. It thus follows from the triangle inequality that:
\begin{equation}\frac{1}{Np}\sum_{i,j}\| \overline{v}_j(\widehat{k}_i,\theta)\|^2\leq O_p\left(\frac{1}{T}\right)+\frac{2}{Np}\sum_{i,j} \|\overline{\zeta}_j(\widehat{k}_i,\theta)\|^2,\label{eq_sum_v}\end{equation}
where $\overline{\zeta}_j(k,\theta)$ denotes the mean of $\zeta_{ij}(\theta)$ in group $\widehat{k}_i=k$. Now, using that $\widehat{k}_1,...,\widehat{k}_N$ are functions of $h_1,...,h_N$, we have:
\begin{align*}
&\mathbb{E}\left[\frac{1}{Np}\sum_{i,j} \|\overline{\zeta}_j(\widehat{k}_i,\theta)\|^2\right]\\
& {=}\frac{1}{Np}\sum_{k,j} \mathbb{E}\left[\frac{\sum_{i}\sum_{i'}\boldsymbol{1}\{\widehat{k}_{i}{=}k\}\boldsymbol{1}\{\widehat{k}_{i'}{=}k\}\mathbb{E}_{h_1,...,h_N,\xi_{10},...,\xi_{N0},\lambda_0}\left(\zeta_{ij}(\theta)'\zeta_{i'j}(\theta)\right)}{\sum_{i}\boldsymbol{1}\{\widehat{k}_i=k\}}\right].
\end{align*}
Furthermore, since observations are independent across $i$ given $\lambda_0$:
\begin{align*}&\mathbb{E}_{h_1,...,h_N,\xi_{10},...,\xi_{N0},\lambda_0}\left(\zeta_{i_1,j}(\theta)'\zeta_{i_2,j}(\theta)\right)\\
&=\mathbb{E}_{h_{i_1},\xi_{i_1,0},\lambda_0}\left(\zeta_{i_1,j}(\theta)\right)'\mathbb{E}_{h_{i_2},\xi_{i_2,0},\lambda_0}\left(\zeta_{i_2,j}(\theta)\right)=0\quad \mbox{for all }i_1\neq i_2 \mbox{ and } j.\end{align*}
Hence:
\begin{align*}
&\mathbb{E}\left[\frac{1}{Np}\sum_{i,j} \|\overline{\zeta}_j(\widehat{k}_i,\theta)\|^2\right]=\frac{1}{Np}\sum_{k,j}\mathbb{E}\left[\frac{\sum_{i}\boldsymbol{1}\{\widehat{k}_{i}{=}k\}\mathbb{E}_{h_i,\xi_{i0},\lambda_0}\left(\zeta_{ij}(\theta)'\zeta_{ij}(\theta)\right)}{\sum_{i}\boldsymbol{1}\{\widehat{k}_i=k\}}\right].
\end{align*}
Finally, using that $\mathbb{E}_{h_i,\xi_{i0},\lambda_0}\left(\zeta_{ij}(\theta)\right)=0$, and using part (\ref{ass_regu_vb}) in Assumption \ref{ass_regu_TV}:
$$\mathbb{E}_{h_i=h,\xi_{i0}=\xi,\lambda_0=\lambda}\left(\zeta_{ij}(\theta)'\zeta_{ij}(\theta)\right)=\limfunc{Tr} \left[{\limfunc{Var}}_{h_i=h,\xi_{i0}=\xi,\lambda_0=\lambda}(v_{ij}(\overline{\alpha}^{j}(\theta,\xi_{i0}),\theta))\right]=O\left(1\right),$$
uniformly in $h,\xi,\lambda$.\footnote{Note that the dimension of $v_{ij}$ is fixed throughout, independent of the sample size.} This implies that $\mathbb{E}\left[\frac{1}{Np}\sum_{i,j} \|\overline{\zeta}_j(\widehat{k}_i,\theta)\|^2\right]=O\left(\frac{K}{N}\right)$, and shows (\ref{rate_vbar}) and (\ref{Rate_alpha_hat_theta}).

We are now going to show:
\begin{equation}\label{eq_sup_theta_v}
\sup_{\theta\in\Theta}\,\frac{1}{Np}\sum_{i,j} \| \overline{v}_j(\widehat{k}_i,\theta)\|^2=o_p\left(1\right).\end{equation}
Using a bounding argument similar to the one we used to show (\ref{Rate_alpha_hat_theta}), (\ref{eq_alpha_hat_sup}) will then follow. To see that (\ref{eq_sup_theta_v}) holds, let $Z(\theta)=\frac{1}{Np}\sum_{i,j} \|\overline{v}_j(\widehat{k}_i,\theta)\|^2$. By (\ref{rate_vbar}), $Z(\theta)=O_p(\delta)$ for all $\theta\in\Theta$. Moreover: $\frac{\partial Z(\theta)}{\partial\theta}=\frac{2}{Np}\sum_{i,j} \overline{v}_j^{\theta}(\widehat{k}_i,\theta)\overline{v}_j(\widehat{k}_i,\theta)=O_p\left(\sqrt{ \sup_{\theta\in\Theta}\, Z(\theta)}\right)$ 
uniformly in $\theta$, using the Cauchy Schwarz inequality with either Assumption \ref{ass_regu} (\ref{ass_regu_iia}) or \ref{ass_regu_TV} (\ref{ass_regu_iib}), where $\overline{v}_j^{\theta}(k,\widetilde{\theta})$ is the mean of $\frac{\partial}{\partial \theta}\big|_{\theta=\widetilde{\theta}}\,v_{ij}(\overline{\alpha}^{j}(\theta,\xi_{i0}),\theta)'$ in group $\widehat{k}_i=k$. Since $\Theta$ is compact, it follows that $\sup_{\theta\in\Theta}\, Z(\theta)=o_p(1)$.\footnote{Let $\upsilon>0,\epsilon>0$. There is $M>0$ such that $\Pr\left(\sup_{\theta\in\Theta}\,\left\|\frac{\partial Z(\theta)}{\partial\theta}\right\|>M\sqrt{ \sup_{\theta\in\Theta}\, Z(\theta)}\right)<\frac{\epsilon}{2}$. Take a finite cover of $\Theta=B_1\cup...\cup B_R$, where $B_r$ are balls with centers $\theta_r$ and diameters $\limfunc{diam} B_r\leq \frac{1}{2M}\sqrt{\upsilon}$. Since: $\sup_{\theta\in\Theta}\, Z(\theta)\leq \max_{r}Z(\theta_r)+\sup_\theta\,\left\|\frac{\partial Z(\theta)}{\partial\theta}\right\| \frac{1}{2M}\sqrt{\upsilon}$, and since: $a>\upsilon\Rightarrow a-\sqrt{a}\frac{1}{2}\sqrt{\upsilon}>\frac{\upsilon}{2}$, we have: $\Pr\left(\sup_{\theta\in\Theta}Z(\theta)>\upsilon\right)\leq \frac{\epsilon}{2}+\Pr\left(\max_{r}Z(\theta_r)>\frac{\upsilon}{2}\right)$, which, by (\ref{rate_vbar}), is smaller than $\epsilon$ for $N,T,K$ large enough.}

\paragraph{\underline{Lemma \ref{lem_sup2}.}}

Let us omit references to $\theta_0$ and $\alpha_{i0}^j$ throughout, and let:
	\begin{align*}
	&A=\frac{1}{Np}\sum_{i=1}^N\sum_{j=1}^p\mathbb{E}_{\xi_{i0},\lambda_0}\left(v_{ij}^{\theta}\right)\left[\mathbb{E}_{\xi_{i0},\lambda_0}\left(v_{ij}^{\alpha}\right)\right]^{-1}v_{ij}^{\alpha}\left(\widehat{\alpha}^j(\widehat{k}_i,\theta_0)-\alpha_{i0}^j+(v_{ij}^{\alpha})^{-1}v_{ij}\right),\\
	&B=\frac{1}{Np}\sum_{i=1}^N\sum_{j=1}^p\left(v_{ij}^{\theta}\left(v_{ij}^{\alpha}\right)^{-1}{-}\mathbb{E}_{\xi_{i0},\lambda_0}\left(v_{ij}^{\theta}\right)\left[\mathbb{E}_{\xi_{i0},\lambda_0}\left(v_{ij}^{\alpha}\right)\right]^{-1}\right)v_{ij}^{\alpha}\left(\widehat{\alpha}^j(\widehat{k}_i,\theta_0){-}\alpha_{i0}^j\right).
	\end{align*}

 \underline{We first bound $A$}. Expanding: $\sum_{i}\boldsymbol{1}\{\widehat{k}_i{=}k\}v_{ij}(\widehat{\alpha}^j(k)){=}0$ for all $k,j$, we have, for $a_{ij}$ between $\alpha_{i0}^j$ and $\widehat{\alpha}^j(\widehat{k}_i)$:
\begin{align*}&\sum_{i}\boldsymbol{1}\{\widehat{k}_i=k\}v_{ij}(\alpha_{i0}^j)+\sum_{i}\boldsymbol{1}\{\widehat{k}_i=k\}v_{ij}^{\alpha}(\alpha_{i0}^j)(\widehat{\alpha}^j(\widehat{k}_i)-\alpha_{i0}^j)\\
&+\frac{1}{2}\sum_{i}\boldsymbol{1}\{\widehat{k}_i=k\}v_{ij}^{\alpha\alpha}(a_{ij})\left(\widehat{\alpha}^j(\widehat{k}_i)-\alpha_{i0}^j\right)\otimes \left(\widehat{\alpha}^j(\widehat{k}_i)-\alpha_{i0}^j\right)=0.\end{align*}
It follows that $\widehat{\alpha}^j(\widehat{k}_i)=\widetilde{\alpha}_j(\widehat{k}_i)+\widetilde{v}_j(\widehat{k}_i)+\widetilde{w}_j(\widehat{k}_i)$, where:
\begin{align*}&\widetilde{\alpha}_j(k)=\left(\sum_{i}\boldsymbol{1}\{\widehat{k}_i=k\}(-v_{ij}^{\alpha})\right)^{-1}\left(\sum_{i}\boldsymbol{1}\{\widehat{k}_i=k\}(-v_{ij}^{\alpha})\alpha_{i0}^j\right),\notag\\
&\widetilde{v}_j(k)=\left(\sum_{i}\boldsymbol{1}\{\widehat{k}_i=k\}(-v_{ij}^{\alpha})\right)^{-1}\left(\sum_{i}\boldsymbol{1}\{\widehat{k}_i=k\}v_{ij}\right),\notag\\&\widetilde{w}_j(k)=\frac{1}{2}\left(\sum_{i}\boldsymbol{1}\{\widehat{k}_i=k\}(-v_{ij}^{\alpha})\right)^{-1}\left(\sum_{i}\boldsymbol{1}\{\widehat{k}_i=k\}v_{ij}^{\alpha\alpha}(a_{ij})\left(\widehat{\alpha}^j(\widehat{k}_i)-\alpha_{i0}^j\right)^{\otimes 2}\right),\end{align*}
where $a^{\otimes 2}=a\otimes a$. Hence, we have:
\begin{eqnarray*}
	A{=}\frac{1}{Np}\sum_{i,j}\mathbb{E}_{\xi_{i0},\lambda_0}\left(v_{ij}^{\theta}\right)\left[\mathbb{E}_{\xi_{i0},\lambda_0}\left(v_{ij}^{\alpha}\right)\right]^{-1}v_{ij}^{\alpha}\left(\widetilde{w}_j(\widehat{k}_i){+}\widetilde{\alpha}_j(\widehat{k}_i){-}\alpha_{i0}^j{+}\widetilde{v}_j(\widehat{k}_i){+}(v_{ij}^{\alpha})^{-1}v_{ij}\right).
\end{eqnarray*}

Note first that:
\begin{eqnarray*}
\frac{1}{Np}\sum_{i,j}\mathbb{E}_{\xi_{i0},\lambda_0}\left(v_{ij}^{\theta}\right)\left[\mathbb{E}_{\xi_{i0},\lambda_0}\left(v_{ij}^{\alpha}\right)\right]^{-1}v_{ij}^{\alpha}\widetilde{w}_j(\widehat{k}_i){=}O_p(\frac{1}{Np}\sum_{i,j}\|\widehat{\alpha}^j(\widehat{k}_i){-}\alpha_{i0}^j\|^2)=O_p(\delta),
\end{eqnarray*}
where we have used (\ref{eq_bounded_hessian}), (\ref{Rate_alpha_hat_theta}), and either Assumption \ref{ass_regu} (\ref{ass_regu_iia}) or Assumption \ref{ass_regu_TV} (\ref{ass_regu_iib}).

Next, let $z_{j}(\xi_{i0})'=\mathbb{E}_{\xi_{i0},\lambda_0}\left(v_{ij}^{\theta}\right)\left[\mathbb{E}_{\xi_{i0},\lambda_0}\left(v_{ij}^{\alpha}\right)\right]^{-1}$. We have:
\begin{eqnarray}
&&\frac{1}{Np}\sum_{i,j}\mathbb{E}_{\xi_{i0},\lambda_0}\left(v_{ij}^{\theta}\right)\left[\mathbb{E}_{\xi_{i0},\lambda_0}\left(v_{ij}^{\alpha}\right)\right]^{-1}v_{ij}^{\alpha}\left(\widetilde{\alpha}_j(\widehat{k}_i)-\alpha_{i0}^j\right)\notag\\&&=\frac{1}{Np}\sum_{i,j}\left(z_{j}(\xi_{i0})'-\widetilde{z}_j\left(\widehat{k}_i\right)'\right)v_{ij}^{\alpha}\left(\widetilde{\alpha}_j(\widehat{k}_i)-\alpha_{i0}^j\right),\label{eq_1star}
\end{eqnarray}
where, for all $k,j$:
\begin{equation}\widetilde{z}_j(k)=\left(\sum_{i}\boldsymbol{1}\{\widehat{k}_i=k\}(-v_{ij}^{\alpha})\right)^{-1}\left(\sum_{i}\boldsymbol{1}\{\widehat{k}_i=k\}(-v_{ij}^{\alpha})z_{j}(\xi_{i0})\right).\label{eq_ztilde}
\end{equation}
Now we have, using that: ${\alpha}^j\mapsto \sum_{i}\left({\alpha}^j(\widehat{k}_i)-\alpha_{i0}^j\right)'(-v_{ij}^{\alpha})\left({\alpha}^j(\widehat{k}_i)-\alpha_{i0}^j\right)$ is minimized at ${\alpha}^j=\widetilde{\alpha}_j$, and using (\ref{eq_bounded_hessian}) and (\ref{eq_mineig}):
\begin{align*}
&\frac{1}{Np}\sum_{i,j}\left\|\widetilde{\alpha}_j(\widehat{k}_i)-\alpha_{i0}^j\right\|^2{=}O_p\left(\frac{1}{Np}\sum_{i,j}\left(\widetilde{\alpha}_j(\widehat{k}_i)-\alpha_{i0}^j\right)'(-v_{ij}^{\alpha})\left(\widetilde{\alpha}_j(\widehat{k}_i)-\alpha_{i0}^j\right)\right)\\
&{=} O_p\left(\frac{1}{Np}{\sum_{i,j}}\left(\widehat{\alpha}^j(\widehat{k}_i){-}\alpha_{i0}^j\right)'(-v_{ij}^{\alpha})\left(\widehat{\alpha}^j(\widehat{k}_i){-}\alpha_{i0}^j\right)\right){=} O_p\left(\frac{1}{Np}{\sum_{i,j}}\left\|\widehat{\alpha}^j(\widehat{k}_i){-}\alpha_{i0}^j\right\|^2\right),
\end{align*}
where the last expression is $O_p(\delta)$ by (\ref{Rate_alpha_hat_theta}). Likewise, since by Assumption \ref{ass_regu} (\ref{ass_regu_iv}) or \ref{ass_regu_TV} (\ref{ass_regu_ivb}) $\frac{\partial \limfunc{vec}z_{j}(\xi)}{\partial \xi'}$ is bounded (in probability) uniformly in $j$ and $\xi$, we have:
\begin{align}
&\frac{1}{Np}\sum_{i,j}\left\|\widetilde{z}_j(\widehat{k}_i){-}z_{j}(\xi_{i0})\right\|^2{=}
O_p\left(\frac{1}{Np}\sum_{i,j}\left(\widetilde{z}_j(\widehat{k}_i){-}z_{j}(\xi_{i0})\right)'(-v_{ij}^{\alpha})\left(\widetilde{z}_j(\widehat{k}_i){-}z_{j}(\xi_{i0})\right)\right)\notag\\
&=O_p\left(\frac{1}{Np}\sum_{i,j}\left(z_j\left(\psi\left(\widehat{h}(\widehat{k}_i)\right)\right)-z_j(\xi_{i0})\right)'(-v_{ij}^{\alpha})\left(z_j\left(\psi\left(\widehat{h}(\widehat{k}_i)\right)\right)-z_j(\xi_{i0})\right)\right)\notag\\
&=O_p\left(\frac{1}{Np}\sum_{i,j}\left\|\widehat{h}(\widehat{k}_i)-\varphi(\xi_{i0})\right\|^2\right)=O_p(\delta),\label{z_tilde_dist}
\end{align}
where we have used (\ref{eq_bounded_hessian}), (\ref{eq_mineig}), Lemmas \ref{theo1} and \ref{lemma_GL}, and that $\psi$ is Lipschitz-continuous. Combining results, and using the Cauchy Schwarz inequality in (\ref{eq_1star}), we obtain:
\begin{eqnarray*}
	\frac{1}{Np}\sum_{i,j}\mathbb{E}_{\xi_{i0},\lambda_0}\left(v_{ij}^{\theta}\right)\left[\mathbb{E}_{\xi_{i0},\lambda_0}\left(v_{ij}^{\alpha}\right)\right]^{-1}v_{ij}^{\alpha}\left(\widetilde{\alpha}_j(\widehat{k}_i)-\alpha_{i0}^j\right)=O_p(\delta).
\end{eqnarray*}

The last term in $A$ is:
$$A_3= \frac{1}{Np}\sum_{i,j}\mathbb{E}_{\xi_{i0},\lambda_0}\left(v_{ij}^{\theta}\right)\left[\mathbb{E}_{\xi_{i0},\lambda_0}\left(v_{ij}^{\alpha}\right)\right]^{-1}(-v_{ij}^{\alpha})\left((-v_{ij}^{\alpha})^{-1}v_{ij}-\widetilde{v}_j(\widehat{k}_i)\right).$$
Since $\widetilde{v}_j(k)=(\sum_{i}\boldsymbol{1}\{\widehat{k}_i=k\}(-v_{ij}^{\alpha}))^{-1}(\sum_{i}\boldsymbol{1}\{\widehat{k}_i=k\}(-v_{ij}^{\alpha})(-v_{ij}^{\alpha})^{-1}v_{ij})$, we have:\begin{align}
A_3&{=}\frac{1}{Np}\sum_{i,j}\left(z_{j}(\xi_{i0})'{-}\widetilde{z}_j\left(\widehat{k}_i\right)'\right)(-v_{ij}^{\alpha})(-v_{ij}^{\alpha})^{-1}v_{ij}{=}\frac{1}{Np}\sum_{i,j}\left(z_{j}(\xi_{i0})'{-}\widetilde{z}_j\left(\widehat{k}_i\right)'\right)v_{ij}\notag\\
&{=}\frac{1}{Np}\sum_{i,j}\left(z_{j}(\xi_{i0})'-{z}_j^*\left(\widehat{k}_i\right)'\right)v_{ij}+\frac{1}{Np}\sum_{i,j}\left({z}_j^*\left(\widehat{k}_i\right)'-\widetilde{z}_j\left(\widehat{k}_i\right)'\right)v_{ij},\label{eq_for_A3}
\end{align}
where $\widetilde{z}_j\left(k\right)$ is given by (\ref{eq_ztilde}), and:
\begin{align}&z_j^*(k){=}\left(\sum_{i}\boldsymbol{1}\{\widehat{k}_i{=}k\}\mathbb{E}_{\xi_{i0},\lambda_0}\left(-v_{ij}^{\alpha}\right)\right)^{-1}\left(\sum_{i}\boldsymbol{1}\{\widehat{k}_i{=}k\}\mathbb{E}_{\xi_{i0},\lambda_0}\left(-v_{ij}^{\alpha}\right)z_{j}(\xi_{i0})\right).\label{eq_zbar}\end{align}

Under the conditions of Theorem \ref{theo2}, it is easy to see that $A_3=O_p(\delta)$. We are now going to show that $A_3=O_p(\delta)$ under the conditions of Theorem \ref{theo2_TV}. To see that the first term on the right-hand-side of (\ref{eq_for_A3}) is $O_p(\delta)$, we use an argument similar to the one we used to show (\ref{rate_vbar}). Let $\zeta_{ij}=v_{ij}-\mathbb{E}_{h_i,\xi_{i0},\lambda_0}(v_{ij})$. Following the same steps as the ones leading to (\ref{eq_sum_v}), we obtain: 
\begin{equation}\label{eq_mean_expec}\frac{1}{Np}\sum_{i,j}\left\| \mathbb{E}_{h_i,\xi_{i0},\lambda_0}(v_{ij})\right\|^2=O_p\left(\frac{1}{T}\right).\end{equation}
Moreover, by an argument similar to (\ref{z_tilde_dist}), since $\mathbb{E}_{\xi_{i0},\lambda_0}(-v_{ij}^\alpha)$ is bounded away from zero with probability one, we have:
\begin{equation}\frac{1}{Np}\sum_{i,j}\left\|z_{j}(\xi_{i0})-{z}_j^*(\widehat{k}_i)\right\|^2=O_p(\delta).\label{eq_two_stars}\end{equation}
Let $z'=(z_{1}',...,z_{p}')$, and $z^*(k)'=(z_1^*(k)',...,z_p^*(k)')$. Since $\zeta_{ij}$ are independent across $i$, with zero mean, conditional on $h_1,...,h_N,\xi_{10},...,\xi_{N0},\lambda_0$, we thus have, denoting $\zeta_{i}=(\zeta_{i1}',...,\zeta_{ip}')'$:
\begin{align*}
&\mathbb{E}\left[\left\|\frac{1}{Np}\sum_{i,j}\left(z_{j}(\xi_{i0})'-{z}_j^*\left(\widehat{k}_i\right)'\right)v_{ij}\right\|^2\right]\\
&{\leq} 2O\left(\frac{1}{T}\right)\mathbb{E}\left[\frac{1}{Np}\sum_{i,j}\left\|z_{j}(\xi_{i0}){-}{z}_j^*\left(\widehat{k}_i\right)\right\|^2\right]
\\&\quad \quad \quad \quad \quad \quad {+}2\mathbb{E}\left[\left\|\frac{1}{Np}\sum_{i,j}\left(z_{j}(\xi_{i0})'{-}{z}_j^*\left(\widehat{k}_i\right)'\right)\zeta_{ij}\right\|^2\right]\\
&=O\left(\frac{\delta}{T}\right)+\frac{2}{N^2p^2}\sum_{i}\mathbb{E}\left[\left(z_i'-{z}^*\left(\widehat{k}_i\right)'\right)\mathbb{E}_{h_i,\xi_{i0},\lambda_0}\left[\zeta_i\zeta_i'\right]\left(z_i-{z}^*\left(\widehat{k}_i\right)\right)\right]\\
&=O\left(\frac{\delta}{T}\right)+O\left(\frac{\delta p}{NT}\right)=O(\delta^2),
\end{align*}
where we have used, in turn, the triangle and Cauchy Schwarz inequalities, (\ref{eq_mean_expec}), (\ref{eq_two_stars}), conditional independence of the $\zeta_{i}$ across $i$, part (\ref{ass_regu_vb}) in Assumption \ref{ass_regu_TV}, and (\ref{eq_two_stars}) one more time. Note that, by part (\ref{ass_regu_vb}) in Assumption \ref{ass_regu_TV}, $\|\mathbb{E}_{h_i,\xi_{i0},\lambda_0}\left[\zeta_i\zeta_i'\right]\|\leq  \limfunc{Tr}\mathbb{E}_{h_i,\xi_{i0},\lambda_0}\left[\zeta_i\zeta_i'\right]\leq p\, {\limfunc{max}}_j\, \limfunc{Tr}\mathbb{E}_{h_i,\xi_{i0},\lambda_0}\left[\zeta_{ij}\zeta_{ij}'\right]=O_p(p^2/T)$.

Turning to the second term in (\ref{eq_for_A3}), we have:
\begin{eqnarray*}
	\frac{1}{Np}\sum_{i,j}\left({z}_j^*\left(\widehat{k}_i\right)'-\widetilde{z}_j\left(\widehat{k}_i\right)'\right)v_{ij}=\frac{1}{Np}\sum_{i,j}\left({z}^*_j\left(\widehat{k}_i\right)'-\widetilde{z}_j\left(\widehat{k}_i\right)'\right)\overline{v}_j\left(\widehat{k}_i\right),
\end{eqnarray*}
where by (\ref{rate_vbar}) we have: $\frac{1}{Np}\sum_{i,j}\| \overline{v}_j(\widehat{k}_i)\|^2=O_p(\delta)$. Moreover:
\begin{align*}
\frac{1}{Np}\sum_{i,j}\left\|{z}_j^*\left(\widehat{k}_i\right)-\widetilde{z}_j\left(\widehat{k}_i\right)\right\|^2{\leq}& \frac{2}{Np}\sum_{i,j}\left\|z_{j}(\xi_{i0}){-}{z}_j^*\left(\widehat{k}_i\right)\right\|^2\\&{+}\frac{2}{Np}\sum_{i,j}\left\|z_{j}(\xi_{i0}){-}\widetilde{z}_j\left(\widehat{k}_i\right)\right\|^2,
\end{align*}
where the second term on the right-hand side is $O_p(\delta)$ due to (\ref{z_tilde_dist}), and the first term is $O_p(\delta)$ due to (\ref{eq_two_stars}). This shows that $A_3=O_p(\delta)$, hence that $A=O_p(\delta)$.

Let us now turn to $B$. Letting: $\pi_{ij}'=v_{ij}^{\theta}\left(v_{ij}^{\alpha}\right)^{-1}-\mathbb{E}_{\xi_{i0},\lambda_0}\left(v_{ij}^{\theta}\right)\left[\mathbb{E}_{\xi_{i0},\lambda_0}\left(v_{ij}^{\alpha}\right)\right]^{-1}$, we have:
\begin{eqnarray*}
	B&=&\frac{1}{Np}\sum_{i,j}\pi_{ij}'v_{ij}^{\alpha}\left(\widetilde{w}_j(\widehat{k}_i)+\widetilde{v}_j(\widehat{k}_i)+\widetilde{\alpha}_j(\widehat{k}_i)-\alpha_{i0}^j\right).
\end{eqnarray*}
First, we have: $\frac{1}{Np}\sum_{i,j}\pi_{ij}'v_{ij}^{\alpha}\widetilde{w}_j(\widehat{k}_i)=O_p(\delta)$. Next, we have:
$	\frac{1}{Np}\sum_{i,j}\pi_{ij}'v_{ij}^{\alpha}\widetilde{v}_j(\widehat{k}_i)=\frac{1}{Np}\sum_{i,j}\widetilde{\pi}_j(\widehat{k}_i)'v_{ij}^{\alpha}\widetilde{v}_j(\widehat{k}_i)$, where $\widetilde{\pi}_j(k)$ is defined similarly to $\widetilde{\alpha}_j(k)$. To see that this quantity is $O_p(\delta)$, note that, by the definition of $\widetilde{v}_j(k)$ and using (\ref{eq_mineig}) and (\ref{rate_vbar}):\begin{eqnarray*}
	\frac{1}{Np}\sum_{i,j}\left\|\widetilde{v}_j(\widehat{k}_i)\right\|^2{=}	O_p\left(\frac{1}{Np}\sum_{i,j}\left\|\overline{v}_j(\widehat{k}_i)\right\|^2\right)=O_p(\delta).
\end{eqnarray*}
Moreover, letting $\tau_{ij}=\pi_{ij}'v_{ij}^{\alpha}$, we have:
\begin{eqnarray*}
	\frac{1}{Np}\sum_{i,j}\left\|\widetilde{\pi}_j(\widehat{k}_i)\right\|^2{=}	O_p\left(\frac{1}{Np}\sum_{i,j}\left\| \overline{\tau}_j(\widehat{k}_i)\right\|^2\right).
\end{eqnarray*}
Now, the $\tau_{ij}$ are independent across $i$, with zero conditional mean given $\xi_{i0},\lambda_0$:
$$\mathbb{E}_{\xi_{i0},\lambda_0}\left(\pi_{ij}'v_{ij}^{\alpha}\right)=\mathbb{E}_{\xi_{i0},\lambda_0}\left(\left(v_{ij}^{\theta}\left(v_{ij}^{\alpha}\right)^{-1}-\mathbb{E}_{\xi_{i0},\lambda_0}\left(v_{ij}^{\theta}\right)\left[\mathbb{E}_{\xi_{i0},\lambda_0}\left(v_{ij}^{\alpha}\right)\right]^{-1}\right)v_{ij}^{\alpha}\right)=0.$$
Using an argument similar to the one we used to show (\ref{rate_vbar}), and using Assumption \ref{ass_regu_TV} (\ref{ass_regu_vb}) in the time-varying case, it thus follows that $\frac{1}{Np}\sum_{i,j}\|\widetilde{\pi}_j(\widehat{k}_i)\|^2=O_p(\delta)$. Hence, by the Cauchy Schwarz inequality: $\frac{1}{Np}\sum_{i,j}\pi_{ij}'v_{ij}^{\alpha}\widetilde{v}_j(\widehat{k}_i)=O_p(\delta)$. 

We lastly bound the third term $B_3$ in $B$:
$$\frac{1}{Np}\sum_{i,j}\pi_{ij}'v_{ij}^{\alpha}\left(\widetilde{\alpha}_j(\widehat{k}_i)-\alpha_{i0}^j\right){=}\frac{1}{Np}\sum_{i,j}\pi_{ij}'v_{ij}^{\alpha}\left[\left({\alpha}_j^*(\widehat{k}_i){-}\alpha_{i0}^j\right){+}\left(\widetilde{\alpha}_j(\widehat{k}_i){-}{\alpha}_j^*(\widehat{k}_i)\right)\right],$$
where $\widetilde{\alpha}_j(k)$ and ${\alpha}^*_j(k)$ are given by expressions similar to (\ref{eq_ztilde}) and (\ref{eq_zbar}), with $\alpha_{i0}^j$ in place of $z_{j}(\xi_{i0})$ in those formulas. The first term is $O_p(\delta)$ since, similarly to (\ref{eq_two_stars}): $\frac{1}{Np}\sum_{i,j}\|{\alpha}_j^*(\widehat{k}_i)-\alpha_{i0}^j\|^2=O_p(\delta)$, and the $\tau_{ij}=\pi_{ij}'v_{ij}^{\alpha}$ are conditionally independent across $i$ with zero mean given $\xi_{i0}$ and $\lambda_0$ (using a similar argument to the first term in (\ref{eq_for_A3})). The second term is:
\begin{eqnarray*}
	\frac{1}{Np}\sum_{i,j}\pi_{ij}'v_{ij}^{\alpha}\left(\widetilde{\alpha}_j(\widehat{k}_i)-{\alpha}_j^*(\widehat{k}_i)\right)&=&\frac{1}{Np}\sum_{i,j}\widetilde{\pi}_j(\widehat{k}_i)'v_{ij}^{\alpha}\left(\widetilde{\alpha}_j(\widehat{k}_i)-{\alpha}^*_j(\widehat{k}_i)\right).
\end{eqnarray*}
We have already shown that: $\frac{1}{Np}\sum_{i,j}\|\widetilde{\pi}_j(\widehat{k}_i)\|^2=O_p(\delta)$. Moreover, using similar arguments to the ones we used to bound $\frac{1}{Np}\sum_{i,j}\|{z}_j^*(\widehat{k}_i)-\widetilde{z}_j(\widehat{k}_i)\|^2$ above, we have: $\frac{1}{Np}\sum_{i,j}\|\widetilde{\alpha}_j(\widehat{k}_i)-{\alpha}_j^*(\widehat{k}_i)\|^2=O_p(\delta)$. This shows that $B_3=O_p(\delta)$, hence that $B=O_p(\delta)$.

\paragraph{\underline{Lemma \ref{lem_sup3}.}}

For given $k,j$, $\theta$-differentiating: $\sum_{i}\boldsymbol{1}\{\widehat{k}_i{=}k\}v_{ij}(\widehat{\alpha}^j(k,\theta),\theta){=}0$, and using (\ref{eq_mineig}), we obtain:
\begin{equation}\frac{\partial \widehat{\alpha}^j(k,\theta)}{\partial \theta'}{=}\left(\sum_i\boldsymbol{1}\{\widehat{k}_{i}{=}k\}\left(-v_{ij}^\alpha\left(\widehat{\alpha}^j(\widehat{k}_{i},\theta),\theta\right)\right)\right)^{-1}\sum_i\boldsymbol{1}\{\widehat{k}_{i}{=}k\}v_{ij}^\theta\left(\widehat{\alpha}^j(\widehat{k}_{i},\theta),\theta\right)'.\label{eq_dalpha}\end{equation}

Let us define, at $\theta=\theta_0$ (and omitting $\theta_0$ and $\alpha_{i0}^j$ from the notation):
\begin{align*}\frac{\partial \widetilde{\alpha}^j(k)}{\partial \theta'}&=\left(\sum_{i}\boldsymbol{1}\{\widehat{k}_i{=}k\}(-v_{ij}^{\alpha})\right)^{-1}\sum_{i}\boldsymbol{1}\{\widehat{k}_i{=}k\}(v_{ij}^{\theta})',\\
\frac{\partial\widetilde{\alpha}^j_*(k)}{\partial \theta'}&=\left(\sum_{i}\boldsymbol{1}\{\widehat{k}_i{=}k\}({-}v_{ij}^{\alpha})\right)^{-1}\sum_{i}\boldsymbol{1}\{\widehat{k}_i{=}k\}({-}v_{ij}^{\alpha})\underset{=\frac{\partial \overline{\alpha}^j(\xi_{i0})}{\partial \theta'}}{\underbrace{\left[\mathbb{E}_{\xi_{i0},\lambda_0}(-v_{ij}^{\alpha})\right]^{-1}\mathbb{E}_{\xi_{i0},\lambda_0}(v_{ij}^{\theta})'}}.\end{align*}
Using (\ref{Rate_alpha_hat_theta}) and (\ref{eq_mineig}), we have:
$\frac{1}{Np}\sum_{i,j}\|\frac{\partial \widehat{\alpha}^j(\widehat{k}_i)}{\partial \theta'}-\frac{\partial \widetilde{\alpha}^j(\widehat{k}_i)}{\partial \theta'}\|^2=o_p\left(1\right)$. Moreover:
\begin{eqnarray*}
	\frac{\partial \widetilde{\alpha}^j(k)}{\partial \theta'}-\frac{\partial\widetilde{\alpha}_*^j(k)}{\partial \theta'}
	&=&\left(\frac{\sum_{i}\boldsymbol{1}\{\widehat{k}_i{=}k\}(-v_{ij}^{\alpha})}{\sum_{i}\boldsymbol{1}\{\widehat{k}_i{=}k\}}\right)^{-1}\left(\frac{\sum_{i}\boldsymbol{1}\{\widehat{k}_i{=}k\}\tau_{ij}'}{\sum_{i}\boldsymbol{1}\{\widehat{k}_i{=}k\}}\right),
\end{eqnarray*}
where $\tau_{ij}'{=}(v_{ij}^{\theta})'{-}({-}v_{ij}^{\alpha})\left[\mathbb{E}_{\xi_{i0},\lambda_0}({-}v_{ij}^{\alpha})\right]^{-1}\mathbb{E}_{\xi_{i0},\lambda_0}(v_{ij}^{\theta})'$ are conditionally independent across $i$, with zero mean given $\xi_{i0}$ and $\lambda_0$. Hence, using (\ref{eq_mineig}), and a similar argument to the one we used to show (\ref{rate_vbar}), we have:
$\frac{1}{Np}\sum_{i,j}\|\frac{\partial \widetilde{\alpha}^j(\widehat{k}_i)}{\partial \theta'}-\frac{\partial\widetilde{\alpha}_*^j(\widehat{k}_i)}{\partial \theta'}\|^2=o_p\left(1\right)$. Lastly, using (\ref{eq_mineig}) we have, as in (\ref{z_tilde_dist}): $\frac{1}{Np}\sum_{i,j}\|\frac{\partial \widetilde{\alpha}_*^j(\widehat{k}_i)}{\partial \theta'}-\frac{\partial \overline{\alpha}^j(\xi_{i0})}{\partial \theta'}\|^2=o_p\left(1\right)$. Combining results shows (\ref{eq_der_alpha}).

\paragraph{\underline{Lemma \ref{lem_coro_bias}.}}

In the following we again evaluate all functions at $\theta_0$, and omit $\theta_0$ for the notation. In particular, $\widehat{\alpha}_{i}$ is a shorthand for $\widehat{\alpha}_{i}(\theta_0)$. We will use the notation $\widehat{w}_i=-v_i^{\alpha}(\widehat{\alpha}_i)$. The choice of $K=\widehat{K}$ with $\gamma=o(1)$ implies that:
\begin{equation}\frac{1}{N}\sum_{i} \|h_i-\widehat{h}(\widehat{k}_i)\|^2=o_p\left(\frac{1}{T}\right).\label{eq_kmeans_manygroups}\end{equation} 
We also have: $\frac{1}{N}\sum_{i}\|\widehat{\alpha}(\widehat{k}_i)-\widehat{\alpha}_{i}\|^2=O_p\left(\frac{1}{T}\right)$. Let, for all $k$:
\begin{align}\label{alpha_hat_alpha_tilde}
&\widetilde{\alpha}(k)=\left(\sum_{i}\boldsymbol{1}\{\widehat{k}_i{=}k\}\widehat{w}_i\right)^{-1}\sum_{i}\boldsymbol{1}\{\widehat{k}_i{=}k\}\widehat{w}_i\widehat{\alpha}_i.
\end{align}
Expanding: $\sum_{i}\boldsymbol{1}\{\widehat{k}_i{=}k\}v_i(\widehat{\alpha}(k))=0$
around $\widehat{\alpha}_i$, using that $v_i(\widehat{\alpha}_i)=0$, we obtain:
\begin{align*}
\widehat{\alpha}(k)=\widetilde{\alpha}(k){+}\frac{1}{2}\left[\sum_{i} \boldsymbol{1}\{\widehat{k}_{i}{=}k\}\widehat{w}_i\right]^{-1}\sum_{i} \boldsymbol{1}\{\widehat{k}_{i}{=}k\}v_{i}^{\alpha\alpha}\left(a_{i}(k)\right)\left(\widehat{\alpha}(\widehat{k}_{i})-\widehat{\alpha}_{i}\right)^{\otimes 2},
\end{align*} 
where $a_i(k)$ lies between $\widehat{\alpha}_{i}$ and $\widehat{\alpha}(k)$, and $v_i^{\alpha\alpha}(a_i(k))$ is a matrix of third derivatives with $(\limfunc{dim}\alpha_{i0})^2$ columns.

To see that (\ref{eq_Delta_L}) holds, we rely on the following decomposition:
\begin{align*}
&\frac{\partial}{\partial\theta}\Big|_{\theta_0}\, \Delta L(\theta)=\frac{1}{N}\sum_{i}\frac{\partial \ell_i(\widehat{\alpha}(\widehat{k}_i)) }{\partial \theta}-\frac{1}{N}\sum_{i}\frac{\partial \ell_i(\widehat{\alpha}_{i}) }{\partial \theta}\\
& =
\frac{1}{N}\sum_{i}v_i^{\theta}(\widehat{\alpha}_i)\left(\widehat{\alpha}(\widehat{k}_i)-\widehat{\alpha}_{i}\right)+\frac{1}{2N}\sum_{i}v_i^{\theta\alpha}(a_i)\left(\widehat{\alpha}(\widehat{k}_i)-\widehat{\alpha}_{i}\right)^{\otimes 2}\\
& =
\frac{1}{N}\sum_{i}v_i^{\theta}(\widehat{\alpha}_i)\left(\widetilde{\alpha}(\widehat{k}_i)-\widehat{\alpha}_{i}\right)+\frac{1}{2N}\sum_{i}v_i^{\theta\alpha}(a_i)\left(\widehat{\alpha}(\widehat{k}_i)-\widehat{\alpha}_{i}\right)^{\otimes 2}\\
& +\frac{1}{2N}\sum_{i}v_i^{\theta}(\widehat{\alpha}_i)\left(\mathbb{E}_{\widehat{k}_i}\left[\widehat{w}_{i'}\right]\right)^{-1}\mathbb{E}_{\widehat{k}_i}\left[v_{i'}^{\alpha\alpha}\left(a_{i'}(\widehat{k}_{i'})\right)\left(\widehat{\alpha}(\widehat{k}_{i'})-\widehat{\alpha}_{i'}\right)^{\otimes 2}\right]\\
& =
\underset{= A_1}{\underbrace{\frac{1}{N}\sum_{i}v_i^{\theta}(\widehat{\alpha}_i)\left(\widetilde{\alpha}(\widehat{k}_i)-\widehat{\alpha}_{i}\right)}}+\underset{= A_2}{\underbrace{\frac{1}{2N}\sum_{i}v_i^{\theta\alpha}(a_i)\left(\widetilde{\alpha}(\widehat{k}_i)-\widehat{\alpha}_{i}\right)^{\otimes 2}}}+o_p\left(\frac{1}{T}\right)\\
& \quad \underset{= A_3}{\underbrace{+\frac{1}{2N}\sum_{i}v_i^{\theta}(\widehat{\alpha}_i)\left(\mathbb{E}_{\widehat{k}_i}\left[\widehat{w}_{i'}\right]\right)^{-1}\mathbb{E}_{\widehat{k}_i}\left[v_{i'}^{\alpha\alpha}\left(a_{i'}(\widehat{k}_{i'})\right)\left(\widetilde{\alpha}(\widehat{k}_{i'})-\widehat{\alpha}_{i'}\right)^{\otimes 2}\right]}},
\end{align*}
where $a_i$ and $a_i(\widehat{k}_i)$ lie between $\widehat{\alpha}_{i}$ and $\widehat{\alpha}(\widehat{k}_i)$, and $\mathbb{E}_{k}$ denotes a mean in group $\widehat{k}_i=k$. Let $\gamma(h)=\{\mathbb{E}_{h_i=h}\left(\widehat{w}_i\right)\}^{-1}\mathbb{E}_{h_i=h}\left(\widehat{w}_i\widehat{\alpha}_i\right)$, and $\nu_i=\widehat{\alpha}_i-\gamma(h_i)$. Let $\widehat{g}_i=v_i^{\theta}(\widehat{\alpha}_i)(\widehat{w}_i)^{-1}$, $\lambda(h)=\mathbb{E}_{h_i=h}\left(\widehat{g}_i\widehat{w}_i\right)\{\mathbb{E}_{h_i=h}\left(\widehat{w}_i\right)\}^{-1}$, and $\tau_i=\widehat{g}_i'-\lambda(h_i)'$. Using (\ref{eq_kmeans_manygroups}) we can show, using that $\gamma$ is Lipschitz-continuous, that $
\frac{1}{N}\sum_{i}\|{\gamma}(h_i)-\widetilde{\gamma}(\widehat{k}_i)\|^2=o_p\left(\frac{1}{T}\right)$. Moreover, we have: $\mathbb{E}\left[\widehat{w}_i\nu_i\,|\, h_1,...,h_N\right]{=}\mathbb{E}_{h_i}\left[\widehat{w}_i\widehat{\alpha}_i\right]-\mathbb{E}_{h_i}\left[\widehat{w}_i\widehat{\alpha}_i\right]{=}0$. Similar arguments to the proof of Lemma \ref{lem_sup1} give:
$
\frac{1}{N}\sum_{i}\|\mathbb{E}_{\widehat{k}_i}\left[\widehat{w}_{i'}\nu_{i'}\right]\|^2=O_p(\frac{K}{NT}){=}o_p\left(\frac{1}{T}\right)$. Hence: $
\frac{1}{N}\sum_{i}\|\widetilde{\nu}(\widehat{k}_i)\|^2=\frac{1}{N}\sum_{i}\|\left(\mathbb{E}_{\widehat{k}_i}\left[\widehat{w}_{i'}\right]\right)^{-1}\mathbb{E}_{\widehat{k}_i}\left[\widehat{w}_{i'}\nu_{i'}\right]\|^2=o_p\left(\frac{1}{T}\right)$. Likewise, we have: $
\frac{1}{N}\sum_{i}\|{\lambda}(h_i)-\widetilde{\lambda}(\widehat{k}_i)\|^2=o_p\left(\frac{1}{T}\right)$, and:
$
\frac{1}{N}\sum_{i}\|\widetilde{\tau}(\widehat{k}_i)\|^2=o_p\left(\frac{1}{T}\right)$.\footnote{Here $\widetilde{\gamma}(k)$, $\widetilde{\lambda}(k)$, $\widetilde{\nu}(k)$, and $\widetilde{\tau}(k)$ are defined similarly to $\widetilde{\alpha}(k)$ in (\ref{alpha_hat_alpha_tilde}), with $\gamma(h_i)$, $\lambda(h_i)$, $\nu_i$, and $\tau_i$, respectively, replacing $\widehat{\alpha}_i$ in that formula.}

Let us now expand the three terms $A_1,A_2,A_3$ in the above decomposition: 
\begin{align*}
&A_1=\frac{1}{N}\sum_{i}\widehat{g}_i \widehat{w}_i\left(\widetilde{\alpha}(\widehat{k}_i){-}\widehat{\alpha}_{i}\right)=\frac{1}{N}\sum_{i}\left(\widehat{g}_i{-}\widetilde{g}(\widehat{k}_i)\right) \widehat{w}_i\left(\widetilde{\alpha}(\widehat{k}_i){-}\widehat{\alpha}_{i}\right)\\
&=-\frac{1}{N}\sum_{i}\left(\lambda(h_i)-\widetilde{\lambda}(\widehat{k}_i)+\tau_i'-\widetilde{\tau}(\widehat{k}_i)'\right) \widehat{w}_i\left(\gamma(h_i)-\widetilde{\gamma}(\widehat{k}_i)+\nu_i-\widetilde{\nu}(\widehat{k}_i)\right)\\
&=-\frac{1}{N}\sum_{i}\tau_i' \widehat{w}_i\nu_i+o_p\left(\frac{1}{T}\right)=-\frac{1}{N}\sum_{i}\tau_i' \mathbb{E}_{\xi_{i0}}(-{v}_i^{\alpha}(\alpha_{i0}))\nu_i+o_p\left(\frac{1}{T}\right),\\
&A_2=\frac{1}{2N}\sum_{i}\mathbb{E}_{\xi_{i0}}\left(v_i^{\theta\alpha}(\alpha_{i0})\right)\left(\widetilde{\alpha}(\widehat{k}_i)-\widehat{\alpha}_{i}\right)^{\otimes 2}+o_p\left(\frac{1}{T}\right)
\\
&=\frac{1}{2N}\sum_{i}\mathbb{E}_{\xi_{i0}}\left(v_i^{\theta\alpha}(\alpha_{i0})\right)\left(\widetilde{\gamma}(\widehat{k}_i)-\gamma(h_i)+\widetilde{\nu}(\widehat{k}_i)-\nu_i\right)^{\otimes 2}+o_p\left(\frac{1}{T}\right)\\
&=\frac{1}{2N}\sum_{i}\mathbb{E}_{\xi_{i0}}\left(v_i^{\theta\alpha}(\alpha_{i0})\right)\nu_i^{\otimes 2}+o_p\left(\frac{1}{T}\right),\\
&A_3=\frac{1}{2N}\sum_{i}\mathbb{E}_{\xi_{i0}}\left(v_i^{\theta}(\alpha_{i0})\right)\left[\mathbb{E}_{\xi_{i0}}(-{v}_i^{\alpha}(\alpha_{i0}))\right]^{-1}\mathbb{E}_{\xi_{i0}}\left[v_i^{\alpha\alpha}\left(\alpha_{i0}\right)\right]\nu_i^{\otimes 2}+o_p\left(\frac{1}{T}\right).
\end{align*}

Combining, we get:
\begin{align*}
&\frac{\partial}{\partial\theta}\Big|_{\theta_0}\, \Delta L(\theta)=-\frac{1}{N}\sum_{i}\tau_i' \mathbb{E}_{\xi_{i0}}(-{v}_i^{\alpha}(\alpha_{i0}))\nu_i+o_p\left(\frac{1}{T}\right)\\
& +\frac{1}{2N}\sum_{i}\left[\mathbb{E}_{\xi_{i0}}\left(v_i^{\theta\alpha}(\alpha_{i0})\right){+}\mathbb{E}_{\xi_{i0}}\left(v_i^{\theta}(\alpha_{i0})\right)\left[\mathbb{E}_{\xi_{i0}}(-{v}_i^{\alpha}(\alpha_{i0}))\right]^{-1}\mathbb{E}_{\xi_{i0}}\left[v_i^{\alpha\alpha}\left(\alpha_{i0}\right)\right]\right]\nu_i^{\otimes 2}.
\end{align*}
Now,  $\frac{\partial\widehat{\alpha}_i(\theta_0)}{\partial\theta'}=\widehat{g}_i'$, and: \begin{align*}&\frac{\partial}{\partial\theta'}\Big|_{\theta_0}\limfunc{vec}\mathbb{E}_{\xi_{i0}}\left[-v_i^\alpha \left(\overline{\alpha}(\theta,\xi_{i0}),\theta\right)\right]\\
&= -\left(\mathbb{E}_{\xi_{i0}}\left(v_i^{\theta\alpha}(\alpha_{i0})\right)+\mathbb{E}_{\xi_{i0}}\left(v_i^{\theta}(\alpha_{i0})\right)\left[\mathbb{E}_{\xi_{i0}}(-{v}_i^{\alpha}(\alpha_{i0}))\right]^{-1}\mathbb{E}_{\xi_{i0}}\left[v_i^{\alpha\alpha}\left(\alpha_{i0}\right)\right]\right)'.\end{align*}
Let $\omega_i=\{\mathbb{E}_{h_i}\left(\widehat{w}_i\right)\}^{-1} \widehat{w}_i,$ and $\widetilde{\nu}_i(\theta)=\widehat{\alpha}_i(\theta)-{\mathbb{E}}_{h_i}\left(\omega_i\widehat{\alpha}_i(\theta)\right)$. Combining the above with the expression of the bias of the FE score, we obtain:
\begin{align}\label{eq_deltaL}
&	\frac{\partial}{\partial\theta}\Big|_{\theta_0} \Delta L(\theta){=}-\frac{\partial}{\partial\theta}\Big|_{\theta_0} \frac{1}{2N}\sum_{i}\widetilde{\nu}_i(\theta)'\mathbb{E}_{\xi_{i0}}\left[-v_i^\alpha \left(\overline{\alpha}(\theta,\xi_{i0}),\theta\right)\right] \widetilde{\nu}_i(\theta){+}o_p\left(\frac{1}{T}\right).
\end{align}

Lastly, let $\widehat{\alpha}_i(\theta)=\mathbb{E}_{h_i}(\widehat{\alpha}_i(\theta))+\nu_i(\theta)$, and $\omega_i=\mathbb{E}_{h_i}(\omega_i)+\eta_i=1+\eta_i$. We have:
$\widetilde{\nu}_i(\theta)=\nu_i(\theta)-\mathbb{E}_{h_i}(\eta_i\nu_i(\theta))$, from which it follows that: $\frac{1}{N}\sum_{i}\|\widetilde{\nu}_i(\theta_0)-\nu_i(\theta_0)\|^2=o_p\left(1/T\right)$. Likewise:	$\frac{1}{N}\sum_{i}\|\frac{\partial\widetilde{\nu}_i(\theta_0)}{\partial\theta'}-\frac{\partial\nu_i(\theta_0)}{\partial\theta'}\|^2=o_p\left(1/T\right)$. Hence, (\ref{eq_deltaL}) implies (\ref{eq_Delta_L}).

\section{Complements and extensions}

\subsection{Average effects}

\noindent Let $m_i(\alpha_i,\theta)=\frac{1}{T}\sum_{t=1}^Tm\left(X_{it},\alpha_{i},\theta\right)$ in the time-invariant case, and $m_i(\alpha_i,\theta)=\frac{1}{T}\sum_{t=1}^Tm\left(X_{it},\alpha_{it},\theta\right)$ in the time-varying case. Let $\widehat{M}=\frac{1}{N}\sum_{i}m_{i}\left(\widehat{\alpha}(\widehat{k}_i),\widehat{\theta}\right)$ be the GFE estimator of $M_0=\frac{1}{N}\sum_{i} m_{i}\left(\alpha_{i0},\theta_0\right)$. We use a common notation as in the proofs of Theorems \ref{theo2} and \ref{theo2_TV}, and denote $m_{ij}(\alpha_i^j,\theta)=m_i(\alpha_i,\theta)$ in the time-invariant case, and  $m_{ij}(\alpha_i^j,\theta)=m\left(X_{it},\alpha_{it},\theta\right)$ in the time-varying case. 

\begin{assumption}{(average effects)} \label{ass_average}
\begin{enumerate}[itemsep=-3pt,label=(\roman*),ref=\roman*,topsep=0pt]
		\item $m_{ij}(\alpha,\theta)$ is twice differentiable in both its arguments, for all $i,j$. \label{ass_average_i}
		\item ${\max}_{i,j}\,{\sup}_{\alpha,\theta}\,\| m_{ij}(\alpha,\theta)\|=O_p(1)$, and similarly for the first two derivatives of $m_{ij}$; $\max_{j} \, \sup_{\widetilde{\xi},\lambda}\, \|\frac{\partial}{\partial\xi'}\big|_{\xi=\widetilde{\xi}}\,\mathbb{E}_{\xi_{i0}=\xi,\lambda_0=\lambda}(\frac{\partial m_{ij}(\alpha_{i0}^j,\theta_0)}{\partial\alpha})\|=O(1)$; and, letting $\tau^m_{ij}=\frac{\partial m_{ij}(\alpha_{i0}^j,\theta_0)}{\partial \alpha'}-\mathbb{E}_{\xi_{i0},\lambda_0}[\frac{\partial m_{ij}(\alpha_{i0}^j,\theta_0)}{\partial \alpha'}]\mathbb{E}_{\xi_{i0},\lambda_0}[v_{ij}^{\alpha}(\alpha_{i0}^j,\theta_0)]v_{ij}^{\alpha}(\alpha_{i0}^j,\theta_0)$, the function $\mathbb{E}_{h_i=h,\xi_{i0}=\xi,\lambda_0=\lambda}(\limfunc{vec}\tau_{ij}^m)$ is twice differentiable with respect to $h$, with first and second derivatives that are uniformly bounded in $j$, $\xi$, $\lambda$, and $h$, and   $\|{\limfunc{Var}}_{h_i=h,\xi_{i0}=\xi,\lambda_0=\lambda}(\limfunc{vec}\tau_{ij}^m)\|=O(\frac{p}{T})$, uniformly in $j$, $\xi$, $\lambda$, and $h$. \label{ass_average_ii}
	\end{enumerate}
\end{assumption}

Let $s_i$ and $H$ as in Theorem \ref{theo2} or \ref{theo2_TV}, and let $\overline{s}=\frac{1}{N}\sum_{i}s_{i}$. Define:
\begin{align*}
s_i^m{=}&\frac{1}{p}\sum_{j}\Bigg\{\mathbb{E}_{\xi_{i0},\lambda_0}\left(\frac{\partial m_{ij}}{\partial\alpha'}\right)\left[\mathbb{E}_{\xi_{i0},\lambda_0}\left({-}\frac{\partial^2 \ell_{ij}}{\partial\alpha\partial\alpha'}\right)\right]^{-1}\frac{\partial\ell_{ij} }{\partial \alpha}{+}\mathbb{E}_{\xi_{i0},\lambda_0}\left(\frac{\partial m_{ij}}{\partial\theta'}\right)H^{-1} \overline{s}\notag\\&+ \mathbb{E}_{\xi_{i0},\lambda_0}\left(\frac{\partial m_{ij}}{\partial\alpha'}\right)\left[\mathbb{E}_{\xi_{i0},\lambda_0}\left({-}\frac{\partial^2 \ell_{ij}}{\partial\alpha\partial\alpha'}\right)\right]^{-1}\mathbb{E}_{\xi_{i0},\lambda_0}\left(\frac{\partial^2 \ell_{ij}}{\partial\alpha\partial\theta'}\right)H^{-1}\overline{s}\Bigg\}.
\end{align*}

\begin{corollary}\label{coro_average}
	
	Let the conditions of Theorem \ref{theo2} or \ref{theo2_TV} hold, and let Assumption \ref{ass_average} hold. Then, as $N,T,K$ tend to infinity such that $Kp/(NT)$ tends to zero:
	\begin{eqnarray*}
		\widehat{M}&=&M_0+ \frac{1}{N}\sum_{i}s_i^m +O_p\left(\frac{1}{T}\right)+O_p\left(\frac{Kp}{NT}\right)+O_p\left(K^{-\frac{2}{d}}\right)+o_p\left(\frac{1}{\sqrt{NT}}\right).
	\end{eqnarray*}
\end{corollary}

\begin{proof}

We have, by a Taylor expansion:
\begin{align*}
	&\widehat{M}-M_0=\frac{1}{Np}\sum_{i,j}m_{ij}\left(\widehat{\alpha}^j(\widehat{k}_i,\widehat{\theta}),\widehat{\theta}\right)-\frac{1}{Np}\sum_{i,j}m_{ij}\left(\alpha_{i0}^j,\theta_0\right)\\&{=}\frac{1}{Np}\sum_{i,j}\frac{\partial m_{ij}\left(\alpha^j_{i0},\theta_0\right)}{\partial \alpha'}\left(\widehat{\alpha}^j(\widehat{k}_i,\widehat{\theta}){-}\alpha^j_{i0}\right){+}\frac{1}{Np}\sum_{i,j}\frac{\partial m_{ij}\left(\alpha^j_{i0},\theta_0\right)}{\partial \theta'}\left(\widehat{\theta}{-}\theta_0\right) {+}O_p\left(\delta\right),
\end{align*}
where $\delta$ is defined as in the proofs of Theorems \ref{theo2} and \ref{theo2_TV}.

Using similar arguments to the ones we used to establish Lemma \ref{lem_sup2}, under Assumption \ref{ass_average} we have (recall that $\overline{\alpha}^j({\theta}_0,\xi_{i0})=\alpha_{i0}^j$):
\begin{align*}&\frac{1}{Np}\sum_{i,j}\frac{\partial m_{ij}\left(\alpha_{i0}^j,\theta_0\right)}{\partial \alpha'}\left(\widehat{\alpha}^j(\widehat{k}_i,{\theta}_0)-\overline{\alpha}^j({\theta}_0,\xi_{i0})\right)\notag\\
&+\frac{1}{Np}\sum_{i,j}\mathbb{E}_{\xi_{i0},\lambda_0}\left[\frac{\partial m_{ij}\left(\alpha^j_{i0},\theta_0\right)}{\partial \alpha'}\right]\mathbb{E}_{\xi_{i0},\lambda_0}\left[v_{ij}^{\alpha}(\alpha_{i0}^j,\theta_0)\right]^{-1}v_{ij}(\alpha_{i0}^j,\theta_0)=O_p(\delta).\end{align*}
Moreover, using (\ref{eq_der_alpha}) and Assumption \ref{ass_average} we obtain:
\begin{align*}& \frac{1}{Np}\sum_{i,j}\frac{\partial m_{ij}\left(\alpha_{i0}^j,\theta_0\right)}{\partial \alpha'}\left\{\left(\widehat{\alpha}^j(\widehat{k}_i,\widehat{\theta})-\overline{\alpha}^j(\widehat{\theta},\xi_{i0})\right)-\left(\widehat{\alpha}^j(\widehat{k}_i,{\theta}_0)-\overline{\alpha}^j({\theta}_0,\xi_{i0})\right)\right\}\\
& =o_p\left(\|\widehat{\theta}-\theta_0\|\right)+O_p(\delta)=o_p\left(\frac{1}{\sqrt{NT}}\right)+O_p(\delta).\end{align*}

Combining, we obtain:
\begin{align*} &\widehat{M}-M_0=\frac{1}{Np}\sum_{i,j}\frac{\partial m_{ij}\left(\alpha^j_{i0},\theta_0\right)}{\partial \alpha'}\left(\widehat{\alpha}^j(\widehat{k}_i,\widehat{\theta})-\widehat{\alpha}^j(\widehat{k}_i,\theta_0)\right)\\
&{+}\frac{1}{Np}\sum_{i,j}\frac{\partial m_{ij}\left(\alpha^j_{i0},\theta_0\right)}{\partial \alpha'}\left(\widehat{\alpha}^j(\widehat{k}_i,\theta_0){-}\alpha_{i0}^j\right){+}\frac{1}{Np}\sum_{i,j}\frac{\partial m_{ij}\left(\alpha^j_{i0},\theta_0\right)}{\partial \theta'}\left(\widehat{\theta}{-}\theta_0\right) {+}O_p\left(\delta\right)\\
&=\frac{1}{Np}\sum_{i,j}\frac{\partial m_{ij}\left(\alpha^j_{i0},\theta_0\right)}{\partial \alpha'}\left(\overline{\alpha}^j(\widehat{\theta},\xi_{i0})-\overline{\alpha}^j({\theta}_0,\xi_{i0})\right)\\
&+\frac{1}{Np}\sum_{i,j}\mathbb{E}_{\xi_{i0},\lambda_0}\left[\frac{\partial m_{ij}\left(\alpha^j_{i0},\theta_0\right)}{\partial \alpha'}\right]\mathbb{E}_{\xi_{i0},\lambda_0}\left[-v_{ij}^{\alpha}(\alpha_{i0}^j,\theta_0)\right]^{-1}v_{ij}(\alpha_{i0}^j,\theta_0)\\
&+\frac{1}{Np}\sum_{i,j}\frac{\partial m_{ij}\left(\alpha^j_{i0},\theta_0\right)}{\partial \theta'}\left(\widehat{\theta}-\theta_0\right) +O_p\left(\delta\right)+o_p\left(\frac{1}{\sqrt{NT}}\right).\end{align*}

The result comes from expanding $\overline{\alpha}^j(\widehat{\theta},\xi_{i0})$ around $\theta_0$, and then substituting $\widehat{\theta}-\theta_0$ by its influence function.\end{proof}

\subsection{Two-way GFE}

\noindent We have the following lemma, whose proof is analogous to that of Lemma \ref{theo1}. 

\begin{lemma}\label{theo1_2way} Suppose that there exist random vectors $h_i=\frac{1}{T}\sum_{t}h(Y_{it},X_{it})$ and $w_t=\frac{1}{N}\sum_{i}w(Y_{it},X_{it})$, with fixed dimensions, and Lipschitz-continuous functions $\varphi$ and $\phi$, such that $h_i=\varphi(\xi_{i0})+o_p(1)$, $\frac{1}{N}\sum_{i}\|h_i-\varphi(\xi_{i0})\|^2=O_p\left(1/T\right)$, $w_t=\phi(\lambda_{t0})+o_p(1)$, and $\frac{1}{T}\sum_{t}\|w_t-\phi(\lambda_{t0})\|^2=O_p\left(1/N\right)$ as $N,T$ tend to infinity. Then we have, as $N,T,K$ tend to infinity: $\frac{1}{N}\sum_{i} \|\widehat{h}(\widehat{k}_i)-\varphi(\xi_{i0})\|^2=O_p\left(\frac{1}{T}\right)+O_p\left(B_{\xi}(K)\right)$, and, as $N,T,L$ tend to infinity: $\frac{1}{T}\sum_{t} \|\widehat{w}(\widehat{l}_t)-\phi(\lambda_{t0})\|^2=O_p\left(\frac{1}{N}\right)+O_p\left(B_{\lambda}(L)\right)$,	where $B_{\lambda}(L)$ is defined analogously to $B_{\xi}(K)$.
	
\end{lemma}

For all $\theta$, $\xi$, and $\lambda$, let $\overline{\alpha}(\theta,\xi,\lambda)=\limfunc{argmax}_{\alpha}\,\mathbb{E}_{\xi_{i0}=\xi,\,\lambda_{t0}=\lambda}(\ell_{it}(\alpha,\theta))$. In addition, let $\xi_0=(\xi_{10}',...,\xi_{N0}')'$.

\begin{assumption}{(regularity, two-way)} \label{ass_regu_2way}

\begin{enumerate}[itemsep=-3pt,label=(\roman*),ref=\roman*,topsep=0pt]
		\item $(Y_{it}',X_{it}')'$, $i=1,..,N$, $t=1,...,T$, are i.i.d. given $\xi_0$ and $\lambda_0$, $\xi_{i0}$ are i.i.d., and $\lambda_{t0}$ are i.i.d.; $\ell_{it}(\alpha,\theta)$ is three times differentiable in $(\theta,\alpha)$; $\Theta$ is compact, the spaces for $\xi_{i0}$ and $\lambda_{t0}$ are compact, and $\theta_0$ belongs to the interior of $\Theta$.\label{ass_regu_2way_i}

		\item $N,T,K,L$ tend jointly to infinity; $\sup_{\xi,\lambda,\alpha,\theta}\, |\mathbb{E}_{\xi_{i0}=\xi,\lambda_{t0}=\lambda}(\ell_{it}(\alpha,\theta))|=O(1)$, and similarly for the first three derivatives of $\ell_{it}$; the minimum (resp., maximum) eigenvalue of $(-\frac{\partial^2 \ell_{it}(\alpha,\theta)}{\partial\alpha\partial{\alpha}^{\prime}})$ is bounded away from zero (resp., infinity) with probability one uniformly in $i,t,\alpha,\theta$, and the third derivatives of $\ell_{it}(\alpha,\theta)$ are $O_p(1)$, uniformly in $i,t,\alpha,\theta$; $\frac{1}{NT}\sum_{i,t}[\ell_{it}(\alpha_{it0},\theta_0)-\mathbb{E}_{\xi_{i0},\lambda_{t0}}(\ell_{it}(\alpha_{{it}0},\theta_0))]^2=O_p(1)$, and similarly for the first three derivatives of $\ell_{it}$.\label{ass_regu_2way_ii}

		\item 
		$\inf_{\xi,\lambda,\theta}\, \mathbb{E}_{\xi_{i0}=\xi,\, \lambda_{t0}=\lambda}(-\frac{\partial^2 \ell_{it}(\overline{\alpha}(\theta,\xi,\lambda),\theta)}{\partial\alpha\partial{\alpha}^{\prime}})>0$; $\mathbb{E}\left[\frac{1}{NT}\sum_{i,t}\ell_{it}(\overline{\alpha}(\theta,\xi_{i0},\lambda_{t0}),\theta)\right]$ has a unique maximum at $\theta_0$ on $\Theta$, and its second derivative is $-H<0$.  \label{ass_regu_2way_iii}

		\item $\frac{\partial}{\partial\xi'}\big|_{\widetilde{\xi}}\,\mathbb{E}_{\xi_{i0}=\xi, \lambda_{t0}=\lambda}(\limfunc{vec}\frac{\partial^2 \ell_{it}(\alpha,\theta_0)}{\partial \theta\partial{\alpha}^{\prime}}){=}O(1)$;	$\frac{\partial}{\partial\lambda'}\big|_{\widetilde{\lambda}}\,\mathbb{E}_{\xi_{i0}=\xi, \lambda_{t0}=\lambda}(\limfunc{vec}\frac{\partial^2 \ell_{it}(\alpha,\theta_0)}{\partial \theta\partial{\alpha}^{\prime}}){=}O(1)$;\\ $\frac{\partial}{\partial\xi'}\big|_{\widetilde{\xi}}\,\mathbb{E}_{\xi_{i0}=\xi, \lambda_{t0}=\lambda}(\limfunc{vec}\frac{\partial^2 \ell_{it}(\alpha,\theta_0)}{\partial \alpha\partial{\alpha}^{\prime}}){=}O(1)$; $\frac{\partial}{\partial\lambda'}\big|_{\widetilde{\lambda}}\,\mathbb{E}_{\xi_{i0}=\xi, \lambda_{t0}=\lambda}(\limfunc{vec}\frac{\partial^2 \ell_{it}(\alpha,\theta_0)}{\partial \alpha\partial{\alpha}^{\prime}}){=}O(1)$;\\ $	\frac{\partial}{\partial\xi'}\big|_{\widetilde{\xi}}\mathbb{E}_{\xi_{i0}=\xi, \lambda_{t0}=\lambda}(\frac{\partial \ell_{it}(\overline{\alpha}(\theta,\xi,\lambda),\theta)}{\partial\alpha}){=}O(1)${;}\,\,\, $\frac{\partial}{\partial\lambda'}\big|_{\widetilde{\lambda}}\mathbb{E}_{\xi_{i0}=\xi, \lambda_{t0}=\lambda}(\frac{\partial \ell_{it}(\overline{\alpha}(\theta,\xi,\lambda),\theta)}{\partial\alpha}){=}O(1)${,}\\ uniformly in $\xi,\widetilde{\xi},\lambda,\widetilde{\lambda},\alpha,\theta$.
		
		\label{ass_regu_2way_iv}

		\item ${\mathbb{E}_{h_i{=}h,\xi_{i0}{=}\xi, w_t{=}w,\lambda_{t0}{=}\lambda}}(\frac{\partial \ell_{it}(\overline{\alpha}(\theta,\xi,\lambda),\theta)}{\partial\alpha})$, ${\mathbb{E}_{h_i{=}h,\xi_{i0}{=}\xi, w_t{=}w,\lambda_{t0}{=}\lambda}}(\limfunc{vec}\frac{\partial}{\partial \theta'}\big|_{\theta_0}\frac{\partial \ell_{it}(\overline{\alpha}(\theta,\xi,\lambda),\theta)}{\partial\alpha})$ are twice differentiable with respect to $h$ and $w$, with first and second derivatives that are uniformly bounded in $h\in {\cal{H}}$, $w\in{\cal{W}}$, $\xi$, $\lambda$, and $\theta\in \Theta$, where ${\cal{H}}$ and ${\cal{W}}$ are the supports of $h_i$ and $w_t$; $\|{\limfunc{Var}}_{h_i=h,\xi_{i0}=\xi, w_t=w,\lambda_{t0}=\lambda}(\frac{\partial\ell_{it} (\overline{\alpha}(\theta,\xi,\lambda),\theta)}{\partial\alpha})\|$ and $\|{\limfunc{Var}}_{h_i=h,\xi_{i0}=\xi, w_t=w,\lambda_{t0}=\lambda}(\limfunc{vec}\frac{\partial}{\partial \theta'}\big|_{\theta_0}\frac{\partial \ell_{it}(\overline{\alpha}(\theta,\xi,\lambda),\theta)}{\partial\alpha})\|$ are $O(1)$,  uniformly in $h$, $w$, $\xi$, $\lambda$, $\theta$.\label{ass_regu_2way_v}

	\end{enumerate} 
	
\end{assumption}

\begin{theorem}$\quad$\label{theo2_2way}Let the conditions in Lemma \ref{theo1_2way} hold. Suppose that $B_{\xi}(K)=O_p(K^{-\frac{2}{d}})$ and $B_{\lambda}(L)=O_p(L^{-\frac{2}{d_{\lambda}}})$. Suppose that $\alpha$ and $\mu$ are Lipschitz-continuous in both arguments, and that there exist two Lipschitz-continuous functions $\psi$ and $\varPsi$ such that $\xi_{i0}=\psi(\varphi(\xi_{i0}))$ and $\lambda_{t0}=\varPsi(\phi(\lambda_{t0}))$. Lastly, let Assumption \ref{ass_regu_2way} hold. Then, as $N,T,K,L$ tend to infinity such that $KL/(NT)$ tends to zero, we have:
	\begin{align*}
	\widehat{\theta}=\theta_0{+} H^{-1}\frac{1}{N}\sum_{i}s_{i} {+}O_p\left(\frac{1}{T}{+}\frac{1}{N}{+}\frac{KL}{NT}\right){+}O_p\left(K^{-\frac{2}{d}}{+}L^{-\frac{2}{d_{\lambda}}}\right){+}o_p\left(\frac{1}{\sqrt{NT}}\right).
	\end{align*} 	
\end{theorem}

\begin{proof}

The proof closely follows the steps of that of Theorem \ref{theo2_TV}. Here we simply highlight the main differences. Let $\delta=\frac{1}{T}+\frac{1}{N}+\frac{KL}{NT}+K^{-\frac{2}{d}}+L^{-\frac{2}{d_{\lambda}}}$. To show \underline{consistency}, a key step is to show, for all $\theta\in\Theta$: 
\begin{align}\label{rate_vbar_2way} &\frac{1}{NT}\sum_{i,t}\left\| \overline{v}(\widehat{k}_i,\widehat{l}_t,\theta)\right\|^2=O_p\left(\delta\right), 
\end{align}
where $\overline{v}(k,l,\theta)$ denotes the mean of $v_{it}(\overline{\alpha}(\theta,\xi_{i0},\lambda_{t0}),\theta)$ in the intersection of groups $\widehat{k}_i{=}k$ and $\widehat{l}_t{=}l$. Let: $\rho(h,\xi,w,\lambda,\theta)=\mathbb{E}_{h_i=h,\xi_{i0}=\xi, w_t=w,\lambda_{t0}=\lambda}( v_{it}(\overline{\alpha}(\theta,\xi,\lambda),\theta))$, and let, for all $i,t,\theta$: $\zeta_{it}(\theta)=v_{it}(\overline{\alpha}(\theta,\xi_{i0},\lambda_{t0}),\theta)- \rho(h_i,\xi_{i0},w_t,\lambda_{t0},\theta)$. Proceeding as in the proof of Lemma \ref{lem_sup1} we have:
$$\frac{1}{NT}\sum_{i,t}\| \rho(h_i,\xi_{i0},w_t,\lambda_{t0},\theta)\|^2=O_p\left(\frac{1}{T}\right)+O_p\left(\frac{1}{N}\right).$$
We thus only need to bound:
\begin{align*}
&\mathbb{E}\left[\frac{1}{NT}\sum_{i,t} \|\overline{\zeta}(\widehat{k}_i,\widehat{l}_t,\theta)\|^2\right]{=}\frac{1}{NT}\sum_{k,\ell} \mathbb{E}_{k\ell}\left[\mathbb{E}_{h_i,\xi_{i0},w_t,\lambda_{t0}}\left(\zeta_{it}(\theta)'\zeta_{it}(\theta)\right)\right],
\end{align*}
where we have used that observations are independent across $i$ and $t$ given $\xi_0$ and $\lambda_0$, and $\mathbb{E}_{k\ell}$ denotes a mean in groups $\widehat{k}_i=k$ and $\widehat{l}_t=l$. To bound this quantity, we use part $(\ref{ass_regu_2way_v})$ in Assumption \ref{ass_regu_2way}. We thus obtain (\ref{rate_vbar_2way}).

Similarly to the proof of Lemma \ref{lem_sup2}, \underline{we then show}:
\begin{equation}\label{eq_score_toshow_2way}
\frac{1}{NT}\sum_{i,t}\left\{v_{it}^{\theta}\left(\widehat{\alpha}(\widehat{k}_i,\widehat{l}_t){-}\alpha_{it0}\right){+}\mathbb{E}_{\xi_{i0}, \lambda_{t0}}\left(v_{it}^{\theta}\right)\left[\mathbb{E}_{\xi_{i0}, \lambda_{t0}}\left(v_{it}^{\alpha}\right)\right]^{-1}v_{it}\right\}{=}O_p\left(\delta\right),
\end{equation}
where we omit references to $\theta_0$ and $\alpha_{it0}$. \underline{The first key term} is:
$$A_3=\frac{1}{NT}\sum_{i,t}\mathbb{E}_{\xi_{i0}, \lambda_{t0}}\left(v_{it}^{\theta}\right)\left[\mathbb{E}_{\xi_{i0}, \lambda_{t0}}\left(v_{it}^{\alpha}\right)\right]^{-1}(-v_{it}^{\alpha})\left((-v_{it}^{\alpha})^{-1}v_{it}-\widetilde{v}(\widehat{k}_i,\widehat{l}_t)\right),$$
where $\widetilde{v}$ is defined analogously to the proof of Lemma \ref{lem_sup2}. To show that $A_3=O_p(\delta)$, we use that the $\zeta_{it}(\theta_0)$ are independent across $i$ and $t$, with zero mean conditional on $h_1,...,h_N$, $w_1,...,w_T$, $\xi_{0}$, and $\lambda_{0}$. 

Let $\pi_{it}'=v_{it}^{\theta}\left(v_{it}^{\alpha}\right)^{-1}-\mathbb{E}_{\xi_{i0},\lambda_{t0}}\left(v_{it}^{\theta}\right)\left[\mathbb{E}_{\xi_{i0},\lambda_{t0}}\left(v_{it}^{\alpha}\right)\right]^{-1}$. \underline{The second key term} is:
\begin{align*}
B_3&=\frac{1}{NT}\sum_{i,t}\pi_{it}'v_{it}^{\alpha}\left(\widetilde{\alpha}(\widehat{k}_i,\widehat{l}_t)-\alpha_{it0}\right)\\&=\frac{1}{NT}\sum_{i,t}\pi_{it}'v_{it}^{\alpha}\left({\alpha}^*(\widehat{k}_i,\widehat{l}_t)-\alpha_{it0}\right)+\frac{1}{NT}\sum_{i,t}\pi_{it}'v_{it}^{\alpha}\left(\widetilde{\alpha}(\widehat{k}_i,\widehat{l}_t)-{\alpha}^*(\widehat{k}_i,\widehat{l}_t)\right),
\end{align*}
where $\widetilde{\alpha}(k,l)$ and ${\alpha}^*(k,l)$ are defined analogously to the proof of Lemma \ref{lem_sup2}. To show that $B_3=O_p(\delta)$, we use that $\tau_{it}=\pi_{it}'v_{it}^{\alpha}$ are independent across $i$ and $t$ with zero mean given $\xi_{0},\lambda_{0}$. 

\underline{The final step}, as in the proof of Lemma \ref{lem_sup3}, is to show that:
\begin{equation}\frac{1}{NT}\sum_{i,t}\left\|\frac{\partial \widehat{\alpha}(\widehat{k}_i,\widehat{l}_t,\theta_0)}{\partial \theta'}-\frac{\partial \overline{\alpha}(\theta_0,\xi_{i0},\lambda_{t0})}{\partial \theta'}\right\|^2=o_p\left(1\right).\label{eq_der_alpha_2way}\end{equation}
The proof of (\ref{eq_der_alpha_2way}) follows similar arguments to the proof of Lemma \ref{lem_sup3}. \end{proof}

\subsection{GFE based on conditional moments}

	\begin{assumption}{(heterogeneity, conditional case)}\label{ass_alphacond}
		
		\noindent There exist vectors $\xi_{i0}$ of fixed dimension $d$, and $\nu_{i0}$ of dimension $d_{\nu}$, and functions $\alpha$ and $\mu$ Lipschitz-continuous in $\xi$, such that $\alpha_{i0}={\alpha}(\xi_{i0})$ and $\mu_{i0}={\mu}(\xi_{i0},\nu_{i0})$. 
		
	\end{assumption}
	
	Differently from Assumption \ref{ass_alpha}, here $\mu_{i0}$ depends on an additional heterogeneity component $\nu_{i0}$, and by Assumption \ref{ass_inj} the moment $h_i$ is only injective for $\xi_{i0}$.
	
		\begin{assumption}{(regularity, conditional case)} \label{ass_regucond}

	\begin{enumerate}[itemsep=-3pt,label=(\roman*),ref=\roman*,topsep=0pt]
				\item $(Y_i',X_i',\xi_{i0}',\nu_{i0}',h_i')'$ are i.i.d.; $(Y_{it}',X_{it}')'$ are stationary for all $i$; $\ell_{it}(\alpha,\theta)$ is three times differentiable in both its arguments for all $i,t$; and $\Theta$ is compact, the space for $\alpha_{i0}$ is compact, and $\theta_0$ belongs to the interior of $\Theta$.\label{ass_regucond_i}

				\item $N,T,K$ tend jointly to infinity; $\sup_{\xi,\nu,\alpha,\theta}\, |\mathbb{E}_{\xi_{i0}=\xi,\nu_{i0}=\nu}(\ell_{it}(\alpha,\theta))|=O(1)$, and similarly for the first three derivatives of $\ell_{it}$; $ \inf_{\xi,\nu,\alpha,\theta}\, \mathbb{E}_{\xi_{i0}=\xi,\nu_{i0}=\nu}(-\frac{\partial^2 \ell_{it}(\alpha,\theta)}{\partial\alpha\partial{{\alpha}^\prime}})$ is positive definite; and $\max_{i}\,\sup_{\alpha,\theta}\,\left|\ell_{i}(\alpha,\theta)-\mathbb{E}_{\xi_{i0},\nu_{i0}}\left(\ell_{i}(\alpha,\theta)\right)\right|=o_p\left(1\right)$, and similarly for the first three derivatives of $\ell_{i}$.
				\label{ass_regucond_iia}

				\item 
				$ \inf_{\xi,\nu,\theta}\, \mathbb{E}_{\xi_{i0}=\xi,\nu_{i0}=\nu}(-\frac{\partial^2 \ell_{it}(\overline{\alpha}(\theta,\xi),\theta)}{\partial\alpha\partial{\alpha}^{\prime}})>0$; $\mathbb{E}[ \frac{1}{T}\sum_{t=1}^T\ell_{it}(\overline{\alpha}(\theta,\xi_{i0}),\theta)]$ has a unique maximum at $\theta_0$ on $\Theta$, and its matrix of second derivatives is $-H^{\rm cond}<0$; and		$\sup_{\theta}\frac{1}{NT}\sum_{i,t}\|\frac{\partial^2\ell_{it}(\overline{\alpha}(\theta,\xi_{i0}),\theta)}{\partial \theta\partial \alpha'}\|^2=O_p(1)$.

				\label{ass_regucond_iii}

				\item
				
				$\sup_{\widetilde{\xi},\alpha} \|\frac{\partial}{\partial\xi'}\big|_{\xi=\widetilde{\xi}}\mathbb{E}_{\xi_{i0}=\xi}(\limfunc{vec}\frac{\partial^2 \ell_{it}(\alpha,\theta_0)}{\partial \theta\partial{\alpha}^{\prime}})\|$; $ \sup_{\widetilde{\xi},\alpha} \|\frac{\partial}{\partial\xi'}\big|_{\xi=\widetilde{\xi}}\mathbb{E}_{\xi_{i0}=\xi}(\limfunc{vec}\frac{\partial^2 \ell_{it}(\alpha,\theta_0)}{\partial \alpha\partial{\alpha}^{\prime}})\|$; and $\sup_{\widetilde{\xi},\theta} \|\frac{\partial}{\partial\xi'}\big|_{\xi=\widetilde{\xi}}\mathbb{E}_{\xi_{i0}=\xi}(\frac{\partial \ell_{it}(\overline{\alpha}(\theta,\widetilde{\xi}),\theta)}{\partial\alpha})\|$ are $O(1)$.
				\label{ass_regucond_iv}

				\item 
				
				$\mathbb{E}_{h_i=h,\xi_{i0}=\xi}(\frac{\partial \ell_{it}(\overline{\alpha}(\theta,\xi),\theta)}{\partial\alpha})$ is twice differentiable with respect to $h$ and $\xi$, with first and second derivatives that are uniformly bounded in $\xi$, $h\in{\cal{H}}$, and $\theta\in \Theta$; and  $\|{\limfunc{Var}}_{h_i=h,\xi_{i0}=\xi}(\frac{\partial\ell_{i}(\overline{\alpha}(\theta,\xi),\theta)}{\partial\alpha})\|=O(1)$, uniformly in $\xi$, $h$ and $\theta$.\label{ass_regucond_v}
				
			\end{enumerate} 
			
		\end{assumption}

	\begin{corollary}$\quad$\label{theo2cond} Let the conditions of Lemmas \ref{theo1} and \ref{lemma_GL} hold. Let Assumptions \ref{ass_inj}, \ref{ass_alphacond}, and \ref{ass_regucond} hold. Let $K$ be given by (\ref{choice_K_eq}), with $\gamma=O(1)$. Then, as $N,T,K$ tend to infinity such that $T^{1+\frac{d}{2}}=O(N)$ we have:
		\begin{eqnarray}\label{2stepthetacond}
		\widehat{\theta}&=&\theta_0+ O_p\left(\frac{1}{T}\right)+O_p\left(\frac{1}{\sqrt{NT}}\right).
		\end{eqnarray}
		
	\end{corollary}

	\begin{proof}
		Let $\delta=\frac{1}{T}+\frac{K}{N}+K^{-\frac{2}{d}}$.\footnote{Note that if $K=\widehat{K}$ is given by (\ref{choice_K_eq}) with $\gamma{=}O(1)$, then $K{=}O(T^{\frac{d}{2}})$ and $\delta{=}O(\frac{1}{T}{+}\frac{T^{\frac{d}{2}}}{N})$, so if $T^{1+\frac{d}{2}}=O(N)$ then $\delta{=}O(\frac{1}{T})$.} 
	To show \underline{consistency}, the key step is to show: 
		\begin{align}\label{rate_vbar2} &\frac{1}{N}\sum_{i}\left\| \overline{v}(\widehat{k}_i,\theta)\right\|^2=O_p\left(\delta\right),\quad \forall\theta\in\Theta.
		\end{align}
		Let, for all $\theta,h,\xi$: $\rho(h,\xi,\theta)=\mathbb{E}_{h_i=h,\xi_{i0}=\xi}( v_i(\overline{\alpha}(\theta,\xi),\theta))$, and let, for all $i,\theta$: $\zeta_{i}(\theta)=v_i(\overline{\alpha}(\theta,\xi_{i0}),\theta)- \rho(h_i,\xi_{i0},\theta)$. One can show, using similar techniques to the proof of Lemma \ref{lem_sup1}, that: $\frac{1}{N}\sum_{i} \|\overline{\zeta}(\widehat{k}_i,\theta)\|^2=O_p(\frac{K}{N})$, and that this implies (\ref{rate_vbar2}).\footnote{Note that, in the case of Theorem \ref{theo2} (i.e., in the absence of additional heterogeneity $\nu_{i0}$), the left-hand side in (\ref{rate_vbar2}) is $O_p(\frac{1}{T})$.}
		
	\underline{We then show}: $\frac{1}{N}\sum_{i}\frac{\partial \ell_{i}(\widehat{\alpha}(\widehat{k}_i,\theta_0),\theta_0) }{\partial \theta}=O_p(\delta)$, which will follow from: \begin{equation}\label{eq_score_toshow2}
	\frac{1}{N}\sum_{i}v_i^{\theta}\left(\widehat{\alpha}(\widehat{k}_i)-\alpha_{i0}\right)=O_p\left(\delta\right),
	\end{equation}
	where from now on we omit references to $\theta_0$ and $\alpha_{i0}$. We have:
	\begin{eqnarray*}
	\frac{1}{N}\sum_{i}v_i^{\theta}\left(\widehat{\alpha}(\widehat{k}_i)-\alpha_{i0}\right)=\frac{1}{N}\sum_{i}v_i^{\theta}\left(\widetilde{\alpha}(\widehat{k}_i)-\alpha_{i0}+\widetilde{v}(\widehat{k}_i)\right)+O_p\left(\delta\right),
	\end{eqnarray*}
	where $\widetilde{\alpha}(k)$ and $\widetilde{v}(k)$ are as in the proof of Lemma \ref{lem_sup2}; that is, denoting $w_i=(-v_i^{\alpha})$, we have $\widetilde{\alpha}(k)=\overline{w}(k)^{-1}\overline{ w \alpha}(k)$ and $\widetilde{v}(k)=\overline{w}(k)^{-1}\overline{ v}(k)$.

Let $\gamma_v(h_i)=\mathbb{E}_{h_i}(v_i)$, $\zeta_i^v=v_i-\gamma_v(h_i)$, $\gamma_w(h_i)=\mathbb{E}_{h_i}(w_i)$, $\zeta_i^w=w_i-\gamma_w(h_i)$, $\gamma_{v^{\theta}}(h_i)=\mathbb{E}_{h_i}(v_i^{\theta})$, and $\zeta_i^{v^{\theta}}=v_i^{\theta}-\gamma_{v^{\theta}}(h_i)$. \underline{First}, we have:
\begin{align*}
&\frac{1}{N}\sum_{i}v_i^{\theta}\widetilde{v}(\widehat{k}_i)=\frac{1}{N}\sum_{i}v_i^{\theta}\overline{w}(\widehat{k}_i)^{-1}\overline{v}(\widehat{k}_i)=\frac{1}{N}\sum_{i}\overline{v}^{\theta}(\widehat{k}_i)\overline{w}(\widehat{k}_i)^{-1}	v_i\\&{=}\frac{1}{N}\sum_{i}(\overline{\gamma}_{v^{\theta}}(\widehat{k}_i)+\overline{\zeta}^{v^{\theta}}(\widehat{k}_i))(\overline{\gamma}_w(\widehat{k}_i)+\overline{\zeta}^w(\widehat{k}_i))^{-1}	v_i{=}\frac{1}{N}\sum_{i}\overline{\gamma}_{v^{\theta}}(\widehat{k}_i)\overline{\gamma}_w(\widehat{k}_i)^{-1}	v_i+O_p(\delta),
\end{align*}
where for example $\overline{\gamma}_w(k)$ is the mean of $\gamma_w(h_i)$ in group $\widehat{k}_i=k$, and we have used that $\frac{1}{N}\sum_{i}\|\overline{\zeta}^{v^{\theta}}(\widehat{k}_i)\|^2=O_p(K/N)$, $\frac{1}{N}\sum_{i}\|\overline{\zeta}^{w}(\widehat{k}_i)\|^2=O_p(K/N)$, and $\frac{1}{N}\sum_{i}\|v_i\|^2=O_p(1/T)$. Moreover:
\begin{align*}
\frac{1}{N}\sum_{i}\overline{\gamma}_{v^{\theta}}(\widehat{k}_i)\overline{\gamma}_w(\widehat{k}_i)^{-1}	v_i&=\frac{1}{N}\sum_{i}\overline{\gamma}_{v^{\theta}}(\widehat{k}_i)\overline{\gamma}_w(\widehat{k}_i)^{-1}	\gamma_v(h_i)+O_p(\delta),
\end{align*}
where we have used that $\frac{1}{N}\sum_{i}\|\overline{\zeta}^{v}(\widehat{k}_i)\|^2=O_p(K/(NT))$. Lastly, we have:
\begin{align*}
\frac{1}{N}\sum_{i}\overline{\gamma}_{v^{\theta}}(\widehat{k}_i)\overline{\gamma}_w(\widehat{k}_i)^{-1}	\gamma_v(h_i)&=\frac{1}{N}\sum_{i}{\gamma}_{v^{\theta}}(h_i)	{\gamma}_w(h_i)^{-1}\gamma_v(h_i)\\&+\frac{1}{N}\sum_{i}\left[\overline{\gamma}_{v^{\theta}}(\widehat{k}_i)\overline{\gamma}_w(\widehat{k}_i)^{-1}-{\gamma}_{v^{\theta}}(h_i){\gamma}_w(h_i)^{-1}\right]	\gamma_v(h_i),
\end{align*}
where
the first term is $O_p(\delta)$ since it is a mean of i.i.d. terms with mean $O(1/T)$ and variance $O(1/T)$, and the second term is $O_p(\delta)$ since $\frac{1}{N}\sum_{i} \|h_i-\overline{h}(\widehat{k}_i)\|^2=O_p(\delta)$ and the $\gamma$ functions are Lipschitz-continuous. 

\underline{Second}, let $v_i^{\theta}w_i^{-1}=\eta(h_i,\xi_{i0})+e_i$, where $\mathbb{E}_{h_i=h,\xi_{i0}=\xi}(e_iw_i)=0$. We have:
\begin{align*}
&\frac{1}{N}{\sum_{i}}v_i^{\theta}\left(\widetilde{\alpha}(\widehat{k}_i){-}\alpha_{i0}\right){=}\frac{1}{N}{\sum_{i}}\eta(h_i,\xi_{i0})w_i\left(\widetilde{\alpha}(\widehat{k}_i){-}\alpha_{i0}\right){+}\frac{1}{N}{\sum_{i}}e_iw_i\left(\widetilde{\alpha}(\widehat{k}_i){-}\alpha_{i0}\right),
\end{align*}
where the first term is $O_p(\delta)$ since $\frac{1}{N}\sum_{i} \|h_i-\overline{h}(\widehat{k}_i)\|^2=O_p(\delta)$, $\frac{1}{N}\sum_{i} \|\xi_{i0}-\overline{\xi}(\widehat{k}_i)\|^2=O_p(\delta)$, $\frac{1}{N}\sum_{i} \|\widetilde{\alpha}(\widehat{k}_i)-\alpha_{i0}\|^2=O_p(\delta)$, $\eta$ is Lipschitz-continuous, and $w_i$ is uniformly bounded (as in the proof of Lemma \ref{lem_sup2}), and the second term is:
\begin{align*}
\frac{1}{N}\sum_{i}e_iw_i\left(\widetilde{\alpha}(\widehat{k}_i)-\alpha_{i0}\right)&{=}\frac{1}{N}\sum_{i}e_iw_i\left(\widetilde{\alpha}(\widehat{k}_i){-}\overline{\alpha}(\widehat{k}_i)\right)+\frac{1}{N}\sum_{i}e_iw_i\left(\overline{\alpha}(\widehat{k}_i){-}\alpha_{i0}\right)\\
&{=}\frac{1}{N}\sum_{i}\overline{e w}(\widehat{k}_i)\left(\widetilde{\alpha}(\widehat{k}_i)-\overline{\alpha}(\widehat{k}_i)\right)+O_p(\delta)=O_p(\delta),
\end{align*}
where we have used that the $(e_iw_i)$'s have zero mean given $h_1,...,h_N,\xi_{10},...,\xi_{N0}$ with bounded conditional variance, and $\frac{1}{N}\sum_{i} \|\overline{e w}(\widehat{k}_i)\|^2=O_p(K/N)=O_p(\delta)$.

\underline{Finally, to show}: $\frac{1}{N}\sum_{i}\frac{\partial^2}{\partial \theta\partial \theta'}\big|_{\theta_0}\,  (\ell_{i}(\widehat{\alpha}(\widehat{k}_i,\theta),\theta)-\ell_{i}(\overline{\alpha}(\theta,\xi_{i0}),\theta))=o_p(1)$, we use similar arguments to the proof of Lemma \ref{lem_sup3}.\footnote{Although the arguments are as in the proof of  Lemma \ref{lem_sup3}, the target log-likelihood is different since here $\overline{\alpha}(\theta,\xi_{i0})$ only depends on $\xi_{i0}$, not on $(\xi_{i0}',\nu_{i0}')'$. In particular, the matrix $H^{\rm cond}$ in Assumption \ref{ass_regucond} differs from the matrix $H$ in Assumption \ref{ass_regu}; see (\ref{eq1}) for an example.}
\end{proof}

\paragraph{Example: a linear homoskedastic model.} Consider the model $Y_{it}=X_{it}\theta_0+\alpha_{i0}+U_{it}$, where $X_{it}$ are scalar and $U_{it}$ are i.i.d. with mean zero and variance $\sigma^2$ given $X_{i1},...,X_{iT},\alpha_{i0}$. Let $\widehat{\theta}$ be the GFE estimator based on a moment $h_i=\varphi(\alpha_{i0})+\varepsilon_i$ that satisfies Assumptions \ref{ass_alpha} and \ref{ass_inj} for $\xi_{i0}=\alpha_{i0}$; that is, $h_i$ is only informative about $\alpha_{i0}$, but not about the heterogeneity in $X_{it}$. Let $\zeta_i^X=\overline{X}_i-\mathbb{E}_{h_i}(\overline{X}_i)$, $\zeta_i^{\alpha}={\alpha}_{i0}-\mathbb{E}_{h_i}(\alpha_{i0})$, and $\zeta_i^U=\overline{U}_i-\mathbb{E}_{h_i}(\overline{U}_i)$. We assume that $K$ is large enough for the approximation error to be of smaller order, and that $K/N$ tends to zero, as in Corollary \ref{coro_bias}. Under appropriate conditions in the regression model, using similar arguments to the proof of Corollary \ref{theo2cond} (though with no need for any restriction on the relative rates of $N$ and $T$), one can show that $\widehat{\theta}$ admits the following expansion:
	\begin{align}&\widehat{\theta}{=}\theta_0{+}\frac{\frac{1}{N}\sum_{i}\zeta^X_i(\zeta^\alpha_i+\zeta^U_i)+\frac{1}{NT}\sum_{i,t}(X_{it}-\overline{X}_i)(U_{it}-\overline{U}_i)}{\mathbb{E}[(X_{it}-\overline{X}_i)^2]+\limfunc{Var}(\zeta^X_i)}{+}o_p\left(\frac{1}{T}\right){+}o_p\left(\frac{1}{\sqrt{NT}}\right).\label{eq1}
	\end{align}
Notice two differences between (\ref{eq1}) and the expansion of the FE estimator: the presence of $\limfunc{Var}(\zeta^X_i)$ in the denominator, and the presence of $\frac{1}{N}\sum_{i}\zeta^X_i(\zeta^\alpha_i+\zeta^U_i)$ in the numerator. In addition, notice that (\ref{eq1}) simplifies to the expression in Corollary \ref{coro_bias} in the absence of additional heterogeneity $\nu_{i0}$.

%

\section{Simulations}

\paragraph{Model of wages and participation (see (\ref{eq_ex_prob2})).} We model the initial condition as: $Y_{i0}=\boldsymbol{1}\left\{u(\alpha_{i0})\geq c(1;
\theta_0)+U_{i0}\right\}$, with $U_{i0}$ standard normal, independent of $\alpha_{i0}$. We set $c(0;\theta_0)=0$ and $c(1;\theta_0)=-1$. We set $\alpha_{i0}$ and $V_{it}$ to be independent standard normals. In the simulations based on models (\ref{eq_ex_prob2}) and (\ref{eq_ex_prob}) we weight the moments by the share of between-i variance to total variance.\footnote{Specifically, we demean and rescale $h_i$ so that all its components $h_{i\ell}$ have zero mean and unit variance, and multiply each component $h_{i\ell}$ by: $\limfunc{max}\left(\frac{\sum_{i}h_{i\ell}^2-\frac{1}{T^2}\sum_{i,t} (h_{it\ell}-h_{i\ell})^2}{\sum_{i}h_{i\ell}^2},0\right)$. Using equal weights instead has small effects in these simulations, however we observed that this particular weighting can improve performance when some moments are substantially less informative about the heterogeneity than others. 
} To compute the variance $\widehat{V}_h$ to set the number of groups in this dynamic model, we use a Newey-West expression with one lag. Lastly, for kmeans computation we use Lloyd's algorithm with 100 random starting values. Table \ref{Tab_Wages_1000} shows additional simulation results for this model.

\paragraph{Probit model with time-varying heterogeneity (see (\ref{eq_ex_prob})).}
The $U_{it}$'s are standard normal independent of the $X_{it}$'s and the $\alpha_{it0}$'s. The data generating process (DGP) for the scalar covariate is: $X_{it}=\mu_{it0}+V_{it}$, where $V_{it}$ are i.i.d. standard normal independent of the $U_{it}$'s, $\alpha_{it0}$'s, and $\mu_{it0}$'s, and $\mu_{it0}=\alpha_{it0}$. We set $\theta_0=1$, and set $\xi_{i0}$ and $\lambda_{t0}$ to be i.i.d. Gamma(1,1) draws, independent of each other. Table \ref{Tab_TV_1000} shows additional simulation results for this model, including for the two-way GFE estimator based on both the cross-sectional moments $(\frac{1}{N}\sum_{i}Y_{it},\frac{1}{N}\sum_{i}X_{it})'$, and the individual-specific moments $(\overline{Y}_i,\overline{X}_i)'$.

\paragraph{Conditional moments: an example.} Consider the following probit model:	$Y_{it}=\boldsymbol{1}\{X_{it}'\theta_0+\alpha_{i0}+U_{it}\geq 0\}$,
where the $U_{it}$ are i.i.d. standard normal independent of the $X_{it}$'s and $\alpha_{i0}$, and $\theta_0$ is a vector of ones. The DGP for the $k$-th covariate is: $X_{itk}=\boldsymbol{1}\{\mu_{i0k}+V_{itk}>0\}$, where $V_{itk}$ are i.i.d. standard normal independent of the $U_{it}$'s, $\alpha_{i0}$, and the $\mu_{i0k}$'s, and $\alpha_{i0}$ and the $\mu_{i0k}$'s follow independent standard normals. We vary the number of covariates between $1$ and $3$, so the total dimension of heterogeneity varies between $2$ and $4$. In this model, we expect the bias of FE to be moderate given the time horizon we consider ($T=20$), since $\alpha_{i0}$ is scalar and FE is a conditional approach. The question we ask here is how much the use of conditional moments can help reduce the bias of GFE due to the presence of additional heterogeneity in the covariates and the increased dimensionality of heterogeneity (see Subsection \ref{subsec_1st_ext}).

\underline{Consider first} using $h_i=(\overline{Y}_i,\overline{X}_i')'$ as moments. In Table \ref{Tab_probit_3cov} we show the biases, standard deviations, and root mean squared errors of FE and GFE among 1000 simulations, for N=1000 and T=20. In the top panel we report GFE estimates as a function of the number of groups $K$. We see that, while the bias of GFE remains moderate with one covariate, the bias increases substantially with the dimension of heterogeneity, in agreement with our theory. By comparison, the bias of FE in the bottom panel is indeed quite small, and it only increases moderately with the number of covariates.

The situation is rather different when using \underline{conditional moments} in GFE. In the middle panel in Table \ref{Tab_probit_3cov} we show simulation results for GFE based on covariates-specific conditional means $\overline{Y}_i(x)=\sum_{t=1}^T\boldsymbol{1}\{X_{it}=x\}Y_{it}/\sum_{t=1}^T\boldsymbol{1}\{X_{it}=x\}$. Importantly, in large samples these moments are only informative about $\alpha_{i0}$, not $\mu_{i0}$. We see that the bias of GFE with conditional moments increases only moderately with the number of covariates, and that FE and GFE with conditional moments have comparable --- and quite small --- biases.

Regarding implementation, note that, for a given $i$, all moments $\overline{Y}_i(x)$ may not be available since $i$'s covariates may never take the value $x$ in the sample. In Table \ref{Tab_probit_3cov}, whenever $\overline{Y}_i(x)$ is not available, we set the moment to an imputed value, the overall conditional mean $\overline{Y}(x)=\sum_{i,t}\boldsymbol{1}\{X_{it}=x\}Y_{it}/\sum_{i,t}\boldsymbol{1}\{X_{it}=x\}$. The imputation does not affect the theory, provided the event that any of the $\overline{Y}_i(x)$'s is not available tends to zero with probability approaching one in large samples.\footnote{To provide intuition in a simple case, suppose that $X_{it}$ are binary, i.i.d. over time given $\mu_{i0}$, with $\Pr(X_{it}=1\,|\, \mu_{i0}=\mu)\in(\epsilon,1-\epsilon)$ for all $\mu$, for some $\epsilon>0$. Then $\Pr(\exists i\, :\, X_{i1}=...=X_{iT}=0)\leq N (1-\epsilon)^T$, which tends to zero whenever $(\ln N)/T\rightarrow 0$.} Moreover, we have obtained similar results using an alternative conditional first step implementation that does not rely on imputations.\footnote{This implementation is as follows. Let $I_i(x)$ be the indicator that there exists a $t$ such that $X_{it}=x$, and let $x_1,...,x_M$ denote the points of support of $X_{it}$. In the first step, we use a Lloyd's-like algorithm to minimize the function $\sum_{i=1}^N\sum_{m=1}^MI_i(x_m)\left(\overline{Y}_i(x_m)-g(x_m,k_i)\right)^2$, 
	with respect to $k_1,...,k_N$ and $g(x_1,1)$, ..., $g(x_M,K)$.}

\begin{table}[h!]
	
	\caption{Model (\ref{eq_ex_prob2}) of wages and participation\label{Tab_Wages_1000}}
	
	\begin{center}
		\resizebox{.8\textwidth}{!}{
			\begin{tabular}{c r r r r r r r r r r} 
\toprule 
T & \multicolumn{1}{c}{ Bias } & \multicolumn{1}{c}{ std } & \multicolumn{1}{c}{ RMSE } & \multicolumn{1}{c}{ se/std } & \multicolumn{1}{c}{ Bias } & \multicolumn{1}{c}{ std } & \multicolumn{1}{c}{ RMSE } & \multicolumn{1}{c}{ se/std }\\[2pt]
 & \multicolumn{ 4 }{c}{ GFE, $\eta=1$} & \multicolumn{ 4 }{c}{ FE, $\eta=1$}\\[-3pt]
 \cmidrule(lr){2-5}  \cmidrule(lr){6-9}
5 &     -0.570 &      0.058 &      0.573 &      1.082 &     -0.835 &      0.064 &      0.837 &      1.066\\
10 &     -0.207 &      0.040 &      0.211 &      1.003 &     -0.418 &      0.040 &      0.420 &      1.041\\
20 &     -0.088 &      0.027 &      0.092 &      0.993 &     -0.209 &      0.026 &      0.211 &      1.064\\
30 &     -0.055 &      0.023 &      0.060 &      0.960 &     -0.140 &      0.023 &      0.142 &      0.991\\
40 &     -0.040 &      0.019 &      0.044 &      1.000 &     -0.105 &      0.019 &      0.106 &      1.034\\
50 &     -0.031 &      0.017 &      0.036 &      0.982 &     -0.084 &      0.017 &      0.086 &      1.022\\[0pt]
 & \multicolumn{ 4 }{c}{ GFE, $\eta=2$} & \multicolumn{ 4 }{c}{ FE, $\eta=2$}\\[-3pt]
 \cmidrule(lr){2-5}  \cmidrule(lr){6-9}
5 &     -0.519 &      0.063 &      0.523 &      1.052 &     -0.876 &      0.068 &      0.879 &      1.063\\
10 &     -0.163 &      0.043 &      0.169 &      0.985 &     -0.442 &      0.041 &      0.444 &      1.070\\
20 &     -0.049 &      0.031 &      0.058 &      0.929 &     -0.225 &      0.028 &      0.227 &      1.042\\
30 &     -0.032 &      0.024 &      0.040 &      0.964 &     -0.153 &      0.022 &      0.154 &      1.068\\
40 &     -0.019 &      0.020 &      0.028 &      0.981 &     -0.113 &      0.019 &      0.115 &      1.045\\
50 &     -0.015 &      0.019 &      0.024 &      0.944 &     -0.091 &      0.018 &      0.093 &      1.000\\[0pt]
\bottomrule 
\end{tabular}

		}
	\end{center}
	
	{\footnotesize\textit{Notes: $1000$ simulations, $N=1000$. ``RMSE'' is root mean squared error, ``se'' is the average of standard error estimates across simulations, ``std'' is the standard deviation of the estimator across simulations. $\eta$ is the risk aversion parameter.}}

\end{table}

	\begin{table}
		
		\caption{Probit model  (\ref{eq_ex_prob}) with time-varying heterogeneity\label{Tab_TV_1000}}
		
		\begin{center}
			\begin{sideways}
				\resizebox{1.4\textwidth}{!}{
					\begin{tabular}{c r r r r r r r r r r r r r r r r} 
\toprule 
T & \multicolumn{1}{c}{ Bias } & \multicolumn{1}{c}{ std } & \multicolumn{1}{c}{ RMSE } & \multicolumn{1}{c}{ se/std } & \multicolumn{1}{c}{ Bias } & \multicolumn{1}{c}{ std } & \multicolumn{1}{c}{ RMSE } & \multicolumn{1}{c}{ se/std } & \multicolumn{1}{c}{ Bias } & \multicolumn{1}{c}{ std } & \multicolumn{1}{c}{ RMSE } & \multicolumn{1}{c}{ se/std } & \multicolumn{1}{c}{ Bias } & \multicolumn{1}{c}{ std } & \multicolumn{1}{c}{ RMSE } & \multicolumn{1}{c}{ se/std }\\[2pt]
 & \multicolumn{ 4 }{c}{ 2-way GFE, $\sigma{=}-1$0} & \multicolumn{ 4 }{c}{ GFE, $\sigma{=}-1$0} & \multicolumn{ 4 }{c}{ FE, $\sigma{=}-1$0} & \multicolumn{ 4 }{c}{ IFE, $\sigma{=}-1$0}\\[-3pt]
 \cmidrule(lr){2-5}  \cmidrule(lr){6-9}  \cmidrule(lr){10-13}  \cmidrule(lr){14-17}
5 &     -0.045 &      0.035 &      0.057 &      0.927 &     -0.044 &      0.035 &      0.056 &      0.926 &      0.442 &      0.071 &      0.448 &      0.706 &      0.116 &      0.064 &      0.133 &      0.473\\
10 &     -0.016 &      0.024 &      0.028 &      0.939 &     -0.014 &      0.024 &      0.028 &      0.939 &      0.198 &      0.036 &      0.201 &      0.762 &      0.100 &      0.036 &      0.107 &      0.488\\
20 &     -0.003 &      0.016 &      0.016 &      1.014 &     -0.000 &      0.016 &      0.016 &      1.014 &      0.098 &      0.019 &      0.100 &      0.911 &      0.087 &      0.020 &      0.089 &      0.596\\
30 &     -0.000 &      0.013 &      0.013 &      1.009 &      0.003 &      0.013 &      0.013 &      1.013 &      0.069 &      0.014 &      0.070 &      0.966 &      0.059 &      0.014 &      0.061 &      0.675\\
40 &      0.001 &      0.011 &      0.011 &      1.021 &      0.005 &      0.011 &      0.012 &      1.016 &      0.055 &      0.012 &      0.057 &      0.949 &      0.044 &      0.012 &      0.046 &      0.676\\
50 &      0.001 &      0.010 &      0.010 &      0.995 &      0.006 &      0.010 &      0.012 &      0.994 &      0.048 &      0.011 &      0.049 &      0.947 &      0.036 &      0.011 &      0.037 &      0.677\\[0pt]
 & \multicolumn{ 4 }{c}{ 2-way GFE, $\sigma{=}0$} & \multicolumn{ 4 }{c}{ GFE, $\sigma{=}0$} & \multicolumn{ 4 }{c}{ FE, $\sigma{=}0$} & \multicolumn{ 4 }{c}{ IFE, $\sigma{=}0$}\\[-3pt]
 \cmidrule(lr){2-5}  \cmidrule(lr){6-9}  \cmidrule(lr){10-13}  \cmidrule(lr){14-17}
5 &     -0.044 &      0.042 &      0.060 &      0.883 &     -0.043 &      0.042 &      0.060 &      0.882 &      0.488 &      0.091 &      0.497 &      0.654 &      0.152 &      0.089 &      0.176 &      0.398\\
10 &     -0.022 &      0.026 &      0.035 &      0.951 &     -0.021 &      0.026 &      0.034 &      0.949 &      0.226 &      0.045 &      0.231 &      0.710 &      0.118 &      0.040 &      0.125 &      0.494\\
20 &     -0.009 &      0.018 &      0.020 &      0.969 &     -0.006 &      0.018 &      0.019 &      0.964 &      0.108 &      0.023 &      0.110 &      0.843 &      0.099 &      0.023 &      0.101 &      0.571\\
30 &     -0.004 &      0.014 &      0.015 &      1.001 &      0.001 &      0.014 &      0.014 &      1.000 &      0.072 &      0.017 &      0.074 &      0.918 &      0.068 &      0.017 &      0.070 &      0.615\\
40 &     -0.002 &      0.013 &      0.013 &      0.989 &      0.004 &      0.013 &      0.013 &      0.985 &      0.056 &      0.014 &      0.058 &      0.922 &      0.051 &      0.014 &      0.052 &      0.639\\
50 &     -0.001 &      0.011 &      0.012 &      0.965 &      0.005 &      0.012 &      0.013 &      0.961 &      0.046 &      0.012 &      0.047 &      0.926 &      0.040 &      0.012 &      0.042 &      0.643\\[0pt]
 & \multicolumn{ 4 }{c}{ 2-way GFE, $\sigma{=}1$} & \multicolumn{ 4 }{c}{ GFE, $\sigma{=}1$} & \multicolumn{ 4 }{c}{ FE, $\sigma{=}1$} & \multicolumn{ 4 }{c}{ IFE, $\sigma{=}1$}\\[-3pt]
 \cmidrule(lr){2-5}  \cmidrule(lr){6-9}  \cmidrule(lr){10-13}  \cmidrule(lr){14-17}
5 &     -0.049 &      0.056 &      0.074 &      0.754 &     -0.048 &      0.056 &      0.074 &      0.754 &      0.565 &      0.146 &      0.583 &      0.506 &      0.207 &      0.117 &      0.238 &      0.359\\
10 &     -0.032 &      0.029 &      0.043 &      0.981 &     -0.030 &      0.029 &      0.042 &      0.979 &      0.251 &      0.062 &      0.258 &      0.603 &      0.141 &      0.043 &      0.147 &      0.513\\
20 &     -0.014 &      0.020 &      0.024 &      0.986 &     -0.010 &      0.020 &      0.022 &      0.983 &      0.114 &      0.027 &      0.117 &      0.825 &      0.125 &      0.029 &      0.128 &      0.514\\
30 &     -0.007 &      0.016 &      0.017 &      0.996 &     -0.001 &      0.016 &      0.016 &      0.992 &      0.074 &      0.019 &      0.077 &      0.889 &      0.085 &      0.021 &      0.088 &      0.561\\
40 &     -0.005 &      0.014 &      0.015 &      0.985 &      0.001 &      0.014 &      0.014 &      0.980 &      0.055 &      0.016 &      0.057 &      0.915 &      0.063 &      0.016 &      0.065 &      0.611\\
50 &     -0.003 &      0.012 &      0.012 &      1.027 &      0.004 &      0.012 &      0.013 &      1.025 &      0.044 &      0.014 &      0.046 &      0.951 &      0.050 &      0.014 &      0.052 &      0.632\\[0pt]
 & \multicolumn{ 4 }{c}{ 2-way GFE, $\sigma{=}1$0} & \multicolumn{ 4 }{c}{ GFE, $\sigma{=}1$0} & \multicolumn{ 4 }{c}{ FE, $\sigma{=}1$0} & \multicolumn{ 4 }{c}{ IFE, $\sigma{=}1$0}\\[-3pt]
 \cmidrule(lr){2-5}  \cmidrule(lr){6-9}  \cmidrule(lr){10-13}  \cmidrule(lr){14-17}
5 &     -0.016 &      0.075 &      0.076 &      0.691 &     -0.015 &      0.075 &      0.076 &      0.692 &      0.706 &      0.262 &      0.753 &      0.386 &      0.300 &      0.255 &      0.394 &      0.218\\
10 &     -0.013 &      0.035 &      0.037 &      0.946 &     -0.010 &      0.035 &      0.036 &      0.947 &      0.323 &      0.097 &      0.337 &      0.486 &      0.183 &      0.060 &      0.192 &      0.458\\
20 &     -0.002 &      0.024 &      0.024 &      0.975 &      0.003 &      0.024 &      0.024 &      0.967 &      0.150 &      0.037 &      0.154 &      0.719 &      0.168 &      0.036 &      0.172 &      0.474\\
30 &      0.002 &      0.019 &      0.019 &      0.991 &      0.008 &      0.019 &      0.021 &      0.983 &      0.100 &      0.025 &      0.104 &      0.814 &      0.121 &      0.030 &      0.125 &      0.456\\
40 &      0.003 &      0.016 &      0.016 &      0.989 &      0.010 &      0.016 &      0.019 &      0.985 &      0.076 &      0.020 &      0.079 &      0.851 &      0.091 &      0.021 &      0.093 &      0.536\\
50 &      0.002 &      0.014 &      0.015 &      0.995 &      0.010 &      0.014 &      0.018 &      0.996 &      0.061 &      0.017 &      0.063 &      0.911 &      0.073 &      0.017 &      0.075 &      0.592\\[0pt]
\bottomrule 
\end{tabular}

				}
			\end{sideways}
		\end{center}
		
		{\footnotesize\textit{Notes: See notes to Table \ref{Tab_Wages_1000}. IFE is interacted fixed-effects with one factor. $\sigma$ is the substitution parameter.}}

	\end{table}

\clearpage

%
%

	\begin{table}
		\caption{Probit model with binary covariates\label{Tab_probit_3cov}}		
		\begin{center}
			\begin{tabular}{c r r r r r r r r r} 
\toprule 
K & \multicolumn{1}{c}{ Bias } & \multicolumn{1}{c}{ std } & \multicolumn{1}{c}{ RMSE } & \multicolumn{1}{c}{ Bias } & \multicolumn{1}{c}{ std } & \multicolumn{1}{c}{ RMSE } & \multicolumn{1}{c}{ Bias } & \multicolumn{1}{c}{ std } & \multicolumn{1}{c}{ RMSE }\\[2pt]
 & \multicolumn{ 3 }{c}{ GFE, 1 covariate} & \multicolumn{ 3 }{c}{ GFE, 2 covariates} & \multicolumn{ 3 }{c}{ GFE, 3 covariates}\\[-3pt]
 \cmidrule(lr){2-4}  \cmidrule(lr){5-7}  \cmidrule(lr){8-10}
5 &     -0.189 &      0.029 &      0.191 &     -0.293 &      0.031 &      0.295 &     -0.362 &      0.042 &      0.365\\
10 &     -0.083 &      0.027 &      0.088 &     -0.205 &      0.032 &      0.207 &     -0.275 &      0.035 &      0.278\\
20 &     -0.017 &      0.029 &      0.033 &     -0.118 &      0.030 &      0.122 &     -0.206 &      0.033 &      0.209\\
30 &      0.006 &      0.029 &      0.030 &     -0.081 &      0.030 &      0.086 &     -0.166 &      0.032 &      0.169\\
40 &      0.018 &      0.029 &      0.035 &     -0.056 &      0.030 &      0.064 &     -0.136 &      0.033 &      0.140\\
50 &      0.026 &      0.030 &      0.039 &     -0.040 &      0.031 &      0.051 &     -0.116 &      0.033 &      0.120\\[4pt]
 & \multicolumn{ 3 }{c}{ Cond. GFE, 1 covariate} & \multicolumn{ 3 }{c}{ Cond. GFE, 2 covariates} & \multicolumn{ 3 }{c}{ Cond. GFE, 3 covariates}\\[-3pt]
 \cmidrule(lr){2-4}  \cmidrule(lr){5-7}  \cmidrule(lr){8-10}
5 &     -0.060 &      0.035 &      0.069 &     -0.085 &      0.037 &      0.093 &     -0.111 &      0.039 &      0.117\\
10 &     -0.045 &      0.033 &      0.056 &     -0.073 &      0.043 &      0.085 &     -0.100 &      0.044 &      0.109\\
20 &     -0.015 &      0.034 &      0.038 &     -0.046 &      0.037 &      0.059 &     -0.075 &      0.045 &      0.087\\
30 &      0.008 &      0.036 &      0.036 &     -0.031 &      0.036 &      0.047 &     -0.061 &      0.043 &      0.075\\
40 &      0.025 &      0.035 &      0.043 &     -0.020 &      0.037 &      0.041 &     -0.050 &      0.042 &      0.065\\
50 &      0.034 &      0.035 &      0.049 &     -0.012 &      0.036 &      0.038 &     -0.040 &      0.041 &      0.057\\[4pt]
 & \multicolumn{ 3 }{c}{ FE, 1 covariate} & \multicolumn{ 3 }{c}{ FE, 2 covariates} & \multicolumn{ 3 }{c}{ FE, 3 covariates}\\[-3pt]
 \cmidrule(lr){2-4}  \cmidrule(lr){5-7}  \cmidrule(lr){8-10}
- &      0.062 &      0.031 &      0.069 &      0.074 &      0.034 &      0.081 &      0.088 &      0.039 &      0.097\\[0pt]
\bottomrule 
\end{tabular}

		\end{center}
		{\footnotesize\textit{Notes: $1000$ simulations, $N=1000$, $T=20$. In the top panel we show GFE estimates based on unconditional moments for different $K$ values, in the middle panel we show GFE estimates based on conditional moments for different $K$ values, in the bottom row we show FE estimates.}}
		
	\end{table}

\end{document}